\documentclass[aps, prx, superscriptaddress, 10pt, nofootinbib]{revtex4-2}

\usepackage{graphicx}%
\usepackage{multirow}%
\usepackage{amsmath,amssymb,amsfonts}%
\usepackage{amsthm}%
\usepackage{mathrsfs}%
\usepackage[title]{appendix}%
\usepackage{xcolor}%
\usepackage{soul}
\usepackage{textcomp}%
\usepackage{booktabs}%
\usepackage{algorithm}%
\usepackage{algorithmicx}%
\usepackage{algpseudocode}%
\usepackage{listings}%

\usepackage{enumitem}

\usepackage{bm}

\newcommand{\argmin}{\mathop{\mathrm{argmin}}}
\newcommand{\Tr}{\mathop{\mathrm{Tr}}}

\newcommand{\E}{\mathbb{E}_\theta}

\usepackage{caption}
\usepackage{subcaption}
\usepackage{amsmath,amssymb}
\usepackage{mathtools}
\mathtoolsset{showonlyrefs=true}
\usepackage[italicdiff]{physics}
\usepackage{hyperref}
\hypersetup{
citecolor=teal
}

\def\Spec{\mathop{\mathrm{Spec}}\nolimits}
\def\Proj{\mathop{\mathrm{Proj}}\nolimits}

\usepackage{color}

\theoremstyle{plain}
\newtheorem{thm}{Theorem}[section]
\newtheorem{lem}[thm]{Lemma}
\newtheorem{cor}[thm]{Corollary}
\newtheorem{prop}[thm]{Proposition}
\newtheorem{conj}[thm]{Conjecture}
\theoremstyle{definition}

\newtheorem{rem}[thm]{Remark}

\newtheorem{defn}[thm]{Definition}

\newtheorem{ex}[thm]{Example}
\newtheorem{customthm}{Assumption}
\newtheorem{customfc}{Fundamendal condition}

\usepackage{theoremref}
\usepackage{thmtools}
\usepackage{lmodern}
\usepackage{fix-cm}

\newenvironment{asswithname}[1]
  {\renewcommand\thecustomthm{#1}\customthm}
  {\endcustomthm}
\newenvironment{fcwithname}[1]
  {\renewcommand\thecustomfc{#1}\customfc}
  {\endcustomfc}
\newenvironment{proofsketch}{%
  \proof}{\endproof}

\begin{document}


\title{Statistical inference for quantum singular models}


\author{Hiroshi Yano\textsuperscript{\textsection}}
    \email[]{hiroshi.yano@keio.jp}
    \affiliation{Quantum Computing Center, Keio University, Hiyoshi 3-14-1, Kohoku,
223-8522, Yokohama, Japan}
\author{Yota Maeda\textsuperscript{\textsection}}
    \email[]{yota.maeda@sony.com}
    \affiliation{Quantum Computing Center, Keio University, Hiyoshi 3-14-1, Kohoku,
223-8522, Yokohama, Japan}
    \affiliation{Advanced Research Laboratory, Research Platform, Sony Group
Corporation, 1-7-1 Konan, Minato-ku, 108-0075, Tokyo, Japan}
\author{Naoki Yamamoto}
\email[]{yamamoto@appi.keio.ac.jp}
    \affiliation{Quantum Computing Center, Keio University, Hiyoshi 3-14-1, Kohoku,
223-8522, Yokohama, Japan}
    \affiliation{Department of Applied Physics and Physico-Informatics, Keio
University, Hiyoshi 3-14-1, Kohoku, 223-8522, Yokohama, Japan}

{\renewcommand{\thefootnote}{\fnsymbol{footnote}}
\footnotetext[4]{These authors contributed equally to this work.}}


\begin{abstract}
    Deep learning has seen substantial achievements, with numerical and theoretical evidence suggesting that singularities of statistical models are considered a contributing factor to its performance.
    From this remarkable success of classical statistical models, it is naturally expected that quantum singular models will play a vital role in many quantum statistical tasks.
    However, while the theory of quantum statistical models in regular cases has been established, theoretical understanding of quantum singular models is still limited.
    To investigate the statistical properties of quantum singular models, we focus on two prominent tasks in quantum statistical inference: quantum state estimation and model selection.
    In particular, we base our study on classical singular learning theory and seek to extend it within the framework of Bayesian quantum state estimation.
    To this end, we define quantum generalization and training loss functions and give their asymptotic expansions through algebraic geometrical methods.
    The key idea of the proof is the introduction of a quantum analog of the likelihood function using classical shadows.
    Consequently, we construct an asymptotically unbiased estimator of the quantum generalization loss, the \textit{quantum widely applicable information criterion (QWAIC)}, as a computable model selection metric from given measurement outcomes.
\end{abstract}


\maketitle


\section{Introduction}


Statistical inference \cite{cox2006Principles,casella2024statistical}
 is a paradigmatic methodology for understanding unknown systems based on given sample data and predicting their future behavior. 
The success of inference relies on the choice of probabilistic models, particularly 
parametric models, which use a finite number of parameters to represent a model probability distribution. 
An important prerequisite for parametric models is the regularity condition; roughly speaking, a parametric model is said to be \textit{regular} if its Hessian matrix of the Kullback-Leibler (KL) divergence between the true and model probability distribution is positive definite.
If a model is regular, for instance, in the case of estimation theory, the Maximum Likelihood Estimation (MLE) \cite{fisher1922mathematical} asymptotically achieves the lower bound of the Cramer-Rao inequality \cite{cramer1999mathematical} and offers desirable properties such as the asymptotic normality \cite{van2000asymptotic,lecam2000Asymptotics}. 
Moreover, the regularity is critical to formulating the \textit{model selection} procedure, with \textit{information criteria} being one of the main tools used to select an appropriate parametric model among candidates based on the observed data \cite{anderson2004model,burnham2004multimodel,claeskens2008model}.
In particular, based on the regular MLE property, Akaike Information Criterion (AIC) \cite{akaike1973Information,akaike1974New}, historically the first proposal of information criteria, enables us to choose an appropriate regular model that takes into account the model complexity for predicting its general performance.

However, statistical models in machine learning are typically non-regular parametric models \cite{amari2003learning,watanabe2007almost}, 
namely \textit{singular models}. 
Singularities are prevalent in models such as mixture models \cite{yamazaki2003singularities,sato2019bayesian,watanabe2021waic,kariya2022asymptotic,watanabe2022asymptotic}, neural networks \cite{aoyagi2012learning}, and deep learning \cite{wei2022deep} like Transformer-based models \cite{hoogland2024stagewise,wang2024loss}. 
Importantly, the beautiful theoretical guarantees for regular models, such as the asymptotic normality of MLE, do not generally hold anymore for singular models. 
The absence of such theoretical guarantees for singular models poses challenges in understanding their behavior and ensuring reliable performance.
On the other hand, with the recent computational advancements, numerous studies of deep learning \cite{lecun2015deep} have numerically shown that it is possible to train such singular models and learn complex behaviors, suggesting that, surprisingly, singularity might be a key to enhancing the performance \cite{lau2024locallearningcoefficientsingularityaware}. 
However, the theoretical foundation for understanding the effective learnability of singular models remains inadequately explored, underscoring the need for rigorous analysis \cite{zhang2021understanding}. 
This gap is critical, especially when considering the safety and consistent functioning of those models in practical applications \cite{bereska2024mechanistic,anwar2024foundationalchallengesassuringalignment}.

Yet there exist many theoretical studies on singular models. 
To provide theoretical insights into deep neural networks, which could be considered as singular models, several useful analyzing tools have been developed, such as the neural tangent kernel (NTK) \cite{jacot2018neural,lee2019wide,lee2020finite}, the mean-field theory \cite{yang2017mean,lee2017deep,mei2018mean}, and the double descent \cite{mei2022Generalization,hastie2022Surprises}.
A notable point is that these approaches take advantage of specific architectures of neural networks. 
On the other hand, there are few studies focusing on general singular models.
For example, algebraic statistics \cite{drton2008lectures} leverages algebraic techniques to study singular statistical parametric models, resulting in descriptions of the free energy. However, it requires the constraint that models are defined by polynomials.
A particularly successful framework is Watanabe's singular learning theory \cite{watanabe2009algebraic,watanabe2018mathematical}, which provides a general method for analyzing singularities based on algebraic geometry. 
Although its application to deep learning is still in the early stages, recent developments \cite{nagayasu2023bayesian,nagayasu2023freeenergybayesianconvolutional,wei2024variational} show promise, particularly in specific models like convolutional neural networks (CNN). 
Within this theory, an information criterion for singular models has been obtained, referred to as Widely Applicable Information Criterion (WAIC) \cite{watanabe2009algebraic}, as a generalization of AIC.

Now, let us move on to the case of quantum. 
In quantum statistical inference \cite{hayashi2005Asymptotic,jencova2006Sufficiency,gill2013Asymptotic}, significant theoretical developments have already emerged for regular models, including quantum Cramer-Rao bound \cite{helstrom1969Quantum} and quantum local asymptotic normality \cite{kahn2009Local} for quantum parameter estimation problems. 
Also, several studies have addressed the model selection problem in quantum state estimation \cite{usami2003Accuracy,yin2011Information,guta2012Rankbased,enk2013When,langford2013Errors,moroder2013Certifying,schwarz2013Error,knips2015How,scholten2018Behavior}, which mainly use techniques from classical hypothesis testing including AIC, based on a fixed measurement and a resultant regular classical probability distribution. 
On the other hand, in the practical scenarios of quantum statistical inference, singular models have been successfully used to deal with even traditionally challenging tasks. 
For example, in recent years, numerous expressive statistical models, such as neural network quantum states (restricted Boltzmann machine \cite{torlai2018Neuralnetwork}, recurrent neural networks \cite{carrasquilla2019Reconstructing}, CNN \cite{schmale2022Efficient}, Transformer \cite{cha2021Attentionbased}), quantum Boltzmann machines \cite{kieferova2017Tomography}, and quantum neural network \cite{kieferova2021Quantum}, have been proposed to enable efficient quantum state tomography.
Moreover, we are witnessing an investigation to explore the possibility of modern machine learning techniques such as dynamical Lie algebra \cite{larocca2023theory}, Gaussian processes \cite{garcia-martin2023Deep}, and NTK \cite{liu2022Representation} for analyzing quantum singular models. 
Definitely, now is the time to build a solid theory of singular models for quantum statistical inference.

This paper is an endeavor of this grand journey. 
We build upon the celebrated singular learning theory for classical Bayesian statistics \cite{watanabe2009algebraic,watanabe2018mathematical} established by Watanabe to 
formulate the Bayesian state tomography framework for quantum singular models. 
Specifically, we investigate the generalization performance of the estimated state through a quantum information-theoretic quantity. 
For regular models, two of the authors \cite{yano2023Quantuma} proposed a quantum generalization of AIC as a preliminary result. 
Following our previous work, we formulate the Bayesian quantum state estimation and define a quantum generalization and training loss functions based on the quantum relative entropy.
The most significant challenge in this work is how to properly define a quantum analog of the likelihood ratio, a key information-theoretic quantity to conduct statistical inference. 
Although the classical likelihood ratio in quantum state estimation can be naturally introduced as the ratio of the probability of obtaining a measurement outcome on two different quantum states, the ``quantum likelihood ratio'' does not yet have a confirmed definition; see for example \cite{yamagata2013Quantum}.
We avoid these formulations by introducing \textit{classical shadows}, originally developed for efficiently estimating specific features of a quantum state \cite{huang2020Predicting}. 
The classical shadow has since become an important object used in many subsequent learning methods for quantum information processing and quantum many-body physics \cite{huang2021Power,huang2022Provably}. 
Based on its efficient classical representation, we use the classical shadow to define a quantum analog of the log-likelihood ratio based on the quantum relative entropy. 
Our main result is the explicit form of asymptotic expansions of the generalization and training losses based on algebraic geometrical methods, which hold even for quantum singular models.
These formulas reflect the complexity of singularities arising from quantum singular models and the penalty of the used measurement (with respect to parameter estimation) through the simultaneous desingularities of the KL divergence and quantum relative entropy.
Using this representation, we propose an asymptotically unbiased estimator for the quantum generalization loss, which we term the \textit{quantum widely applicable information criterion (QWAIC)}.
It can be seen as a quantum generalization of WAIC \cite{watanabe2009algebraic}, allowing us to evaluate the trade-off between the model's adaptability to the observed data and the model complexity.

The rest of this paper is structured as follows.
Section \ref{sec:Preliminaries} is devoted to providing readers with introductory concepts about statistical inference in quantum state models and singular learning theory.
In Section \ref{section:main results}, we introduce the concrete problem setting in quantum state estimation and our main results, including the asymptotic statistical properties of quantum generalization/training loss and QWAIC, with a sketch of proof and their interpretations.
In Section \ref{sec:concrete_examples}, we demonstrate the analytic calculation of important quantities that appear in our main results and numerical simulation through several toy examples.
Section \ref{sec:Conclusion and open questions} presents conclusions and future perspectives. 
Proving our main theorems requires extensive mathematical preparation, which we would like to defer in the appendices.
Appendix \ref{app:Fundamental conditions} summarises the assumptions on which the theory is developed. 
In particular, we see here how the regularity condition can be generalized. 
We then introduce relevant mathematical tools in Appendix \ref{app:Mathematical tools}, consisting of two parts: algebraic geometry and empirical process.
In Appendix \ref{section:Revisiting singular learning theory and WAIC}, we review singular learning theory by Watanabe based on a part of the above assumptions and see that WAIC is an asymptotically unbiased estimator.
Appendix \ref{app:Proof of main results} includes the complete proof of our main results.
For our notation, see \hyperref[app:glossary]{Glossary}.


\section{Preliminaries}
\label{sec:Preliminaries}
This section provides preliminaries for our study.
First, we will formulate the problem of our interest: statistical inference of quantum state models.
Then, we review singular learning theory \cite{watanabe2009algebraic,watanabe2018mathematical}, focusing on the idea behind the derivation of WAIC.
Readers familiar with these notions can skip this section.

\subsection{Statistical inference in quantum state models}
The goal of statistical inference is to make a ``good'' guess about some properties of the underlying structure, using a finite number of observations.
Two important concepts of statistical inference are estimation and model selection.
This subsection aims to highlight these two concepts and outline the formulation of statistical inference in quantum state estimation.

Traditionally, estimation theory mainly focuses on the problem of optimizing the estimator for a given statistical model.
In simple terms, the primary concern is minimizing the following estimation error of a parameter by constructing a good estimator $\hat{\theta}$ given $n$ i.i.d. samples $x^n \coloneqq \{x_1,...,x_n\}$:
\begin{equation}
    L(\hat{\theta}) = d_{p(\cdot|\theta_0)}(\hat{\theta},\theta_0) = \mathbb{E}_{p(X^n|\theta_0)}\left[ \| \hat{\theta}(X^n) - \theta_0 \|^2 \right], \quad p(x^n|\theta_0) = \prod_{i=1}^n p(x_i|\theta_0),
\end{equation}
where $\theta_0$ is a true parameter, $\hat{\theta}$ an estimator, and $p(\cdot|\theta)$ a statistical model \footnote{In this paper, the notations $X$ and $X^n$ are random variables subject to the true probability distribution.}.
The second equality shows the expression when using the mean squared error as the distance measure between parameters.
Under the regularity condition, MLE attains the Cramer-Rao bound asymptotically.

However, there are often cases where not only the true parameter but also the true statistical model are unknown.
In such cases, notions such as \textit{statistical inference} or \textit{learning} come into play.
Now, the goal turns to minimizing the estimation error of a statistical model by constructing a suitable statistical model $p$ and an estimator $\hat{\theta}$:
\begin{equation}
    L(p, \hat{\theta}) = d_q(q, p(\cdot|\hat{\theta})) = \mathbb{E}_{q(X^n)}\left[ \mathrm{KL}(q||p(\cdot|\hat{\theta}(X^n))) \right], \quad q(x^n) = \prod_{i=1}^n q(x_i), 
\end{equation}
where $q$ is a true statistical model.
The second equality shows the expression when using the KL divergence as the distance measure between probability distributions.
Minimizing $L(p,\hat{\theta})$ requires choosing a proper statistical model as well as a parameter estimator, but this can be addressed, for example, by evaluating the KL divergence $\mathrm{KL}(\tilde{q}||p(\cdot|\hat{\theta}(X^n)))$, where $\tilde{q}(x) = (1/n) \sum_i \delta(x-x_i)$ is the empirical distribution, to update $p(\cdot|\theta)$.
Note that $p(\cdot|\theta)$ with many parameters is likely to yield a lower value of $\mathrm{KL}(\tilde{q}||p(\cdot|\hat{\theta}(X^n)))$; however, this often leads to overfitting.
Therefore, a suitable choice of $p(\cdot|\theta)$ is essential.
See \cite{burnham2004multimodel} for the details of model selection.

While model selection is a common topic in classical statistics, it is rarely discussed in quantum statistics.
The reason is considered to lie in the difficulties of quantum estimation.
In quantum estimation, it is necessary to optimize not only estimators but also measurements in order to minimize the estimation error of a parameter:
\begin{equation}
    L(\Pi, \hat{\theta}) = d_{p(\cdot|\theta_0)}(\hat{\theta},\theta_0) = \mathbb{E}_{p(X^n|\theta_0)}\left[ \| \hat{\theta}(X^n) - \theta_0 \|^2 \right], \quad p(x^n|\theta_0) = \Tr(\sigma(\theta_0)^{\otimes n} \Pi_{x^n}), 
\end{equation}
where $\sigma(\theta)$ is a quantum state model and $\Pi$ is a POVM on the space of $\sigma(\theta)^{\otimes n}$.
When the parameter is high dimensional, it is generally known that constructing a measurement and an estimator that minimizes the estimation error is highly challenging.
Thus, quantum estimation theory is still an active research area.

In this work, we consider the problem of \textit{statistical inference in quantum state models}.
In the context of quantum state estimation (or tomography), there are cases where not only the true parameter and the optimal measurement but also the true quantum state model are unknown.
In such a case, analogous to statistical inference in classical statistics, we need to minimize the following estimation error of a quantum state model:
\begin{equation}
    L(\sigma, \Pi, \hat{\theta}) = d_{\rho}(\rho,\sigma(\hat{\theta})) = \mathbb{E}_{q(X^n)}\left[ D(\rho||\sigma(\hat{\theta}(X^n))) \right], \quad q(x^n) = \Tr(\rho^{\otimes n} \Pi_{x^n}), 
\end{equation}
where $\rho$ is a true quantum state.
The second equation shows the expression when using the quantum relative entropy $D(\cdot||\cdot)$ as the distance measure between quantum states (the formal definition of $D(\cdot||\cdot)$ will be given later). 
However, because handling the triple $(\sigma, \Pi, \hat{\theta})$ is challenging, we restrict $\Pi$ to a separable and non-adaptive tomographic complete measurement and fix it, i.e., $q(x^n) = \prod_i \Tr(\rho \Pi_{x_i})$ ($\Pi$ is a POVM on the space of $\rho$), allowing the remaining two $(\sigma, \hat{\theta})$ to be our objectives, similar to classical statistical inference.
Roughly speaking, this makes it easier to evaluate $D(\tilde{\rho}||\sigma(\hat{\theta}(X^n)))$ with $\tilde{\rho}$ as an ``empirical'' state of $\rho$, similar to $\tilde{q}$ in classical statistics. 
In our work, we will consider the classical shadow of $\rho$ as $\tilde{\rho}$, although the classical shadow is more accurately the linear inversion estimator of $\rho$ than an empirical state.
It also may be interpreted as the estimation of an observable $\log \sigma(\hat{\theta})$ with the classical shadow.
Finally, it should also be noted that our theory was developed without reliance on the existing studies on \textit{quantum statistical inference} \cite{hayashi2005Asymptotic,jencova2006Sufficiency,gill2013Asymptotic}, in which quantum Gaussian states allow one to deal with many quantum statistical problems (e.g. testing problems \cite{kumagai2013Quantum}).

\subsection{Singular learning theory} 
\label{subsection:singular learning theory}
Shifting our focus back to classical statistics, it is widely known that parameter estimation becomes much more difficult when a statistical model has many parameters. 
Along with the advancement of classical computers, many strategies have been proposed to reduce the estimation error. 
In the context of learning, \textit{singular learning theory} has been recently established to deal with a statistical model with many parameters within the framework of Bayesian statistics. 
Building on these developments, our study attempts to advance singular learning theory in the realm of quantum state estimation.

We briefly summarize the Bayesian statistics for singular models \cite{watanabe2009algebraic,watanabe2018mathematical}; see also for the recent advance \cite{watanabe2024recent}.  
 Our primary interest here is on the asymptotic behavior of the KL divergence between the estimated statistical model and the true model. Under the standard regularity condition, the second-order Taylor expansion of the KL divergence with respect to the parameters effectively captures its asymptotic properties. However, analytical methods for cases where such assumptions do not hold are largely unknown. In singular learning theory, the use of mathematical tools such as algebraic geometry paves the way for such cases.
We aim to extend this theory to investigate quantum state estimation, which will be discussed in the next section. 
In the remainder of this subsection, we introduce several basic notions and key ideas in singular learning theory.

Given $n$ i.i.d. samples $x^n = \{x_1,...,x_n\}$ from an unknown probability distribution $q(x)$, one wants to predict $q(x)$ using a pair of a statistical model $p(x|\theta)$ and a prior distribution $\pi(\theta)$, where $\theta \in \Theta \subset \mathbb{R}^d$ and $\theta \sim \pi(\theta)$.
Then, the posterior distribution and posterior predictive distribution are defined by 
\begin{align*}
    p(\theta|x^n) \coloneqq \frac{1}{p(x^n)} \pi(\theta) \prod_{i=1}^{n} p(x_i|\theta), \quad
    p(x|x^n) \coloneqq \int_\Theta p(x|\theta) p(\theta|x^n) d\theta,
\end{align*}
where 
\[p(x^n) \coloneqq \int \pi(\theta) \prod_{i=1}^{n} p(x_i|\theta) d\theta\]
is the marginal likelihood. 
To quantify the error of prediction, based on the KL divergence
\[\mathrm{KL}(q\|p(\cdot|\theta)) \coloneq \mathbb{E}_X\left[ \log \frac{q(X)}{p(X|\theta)} \right],\]
the classical generalization loss $G_n$ and training loss $T_n$ are defined by 
\begin{align}
\label{eq:def of Gn and Tn}
    G_n \coloneqq - \mathbb{E}_X[ \log p(X|x^n) ], \quad T_n \coloneqq - \frac{1}{n} \sum_{i=1}^{n} \log p(x_i|x^n).
\end{align}
Here $\mathbb{E}_X[\cdot]$ represents the expectation with respect to $q(X)$.
To study the asymptotic behavior of $G_n$ and $T_n$, let us introduce an optimal parameter set 
\begin{align}
    \label{eq:Theta_0}
\Theta_0 \coloneq \left\{\theta_0\in\Theta \Big\vert\ \theta_0 = \argmin_{\theta \in \Theta} \mathrm{KL}(q\|p(\cdot|\theta)) \right\}.
\end{align}
In this paper, we denote by $\theta_0$ for an element of $\Theta_0$.
Let us recall the assumption commonly made in classical statistics.
\begin{defn}[{Definition \ref{def:regular for classical}}]
The function $q(x)$ is said to be \textit{regular} for $p(x|\theta)$ if the following conditions are satisfied:
\begin{enumerate}
    \item $\Theta_0$ consists of a single element $\theta_0$,
    \item the Hessian matrix $\nabla^2 \mathrm{KL}(q\|p(\cdot|\theta_0))$ is positive definite, and 
    \item there is an open neighborhood of $\theta_0$ in $\Theta$.
\end{enumerate}
\end{defn}
The regularity condition ensures that we can analyze the losses in traditional statistics; in such cases, the generalization loss is described by the dimension of parameters \cite{amari1992four,amari1993statistical,murata1994network}. 
One of its fruits is the construction of AIC \cite{akaike1973Information,akaike1974New}, the most famous information criterion. 
It has been widely applied in the analysis of model selection of statistical models in machine learning \cite{claeskens2008model,burnham2004multimodel,anderson2004model,fraley2002model}, biostatistics \cite{posada1998modeltest,posada2004model}, and econometrics \cite{sclove1987application} etc.

If the regularity condition is not satisfied, we call \textit{singular} \cite{watanabe2023mathematical}.
It is known that singular models usually appear in various statistical models in real problems; mixture models \cite{kariya2022asymptotic,yamazaki2003singularities,sato2019bayesian,watanabe2022asymptotic,watanabe2021waic}, neural networks \cite{aoyagi2012learning}, deep learning \cite{wei2022deep}, Gaussian processes \cite{seeger2004gaussian}, reduced rank regressions \cite{aoyagi2005stochastic,hayashi2017upper,zellner1976bayesian}, and hidden Markov models \cite{yamazaki2005algebraic,zwiernik2011asymptotic} etc.
Notably, in singular cases, the posterior distribution cannot be approximated by any normal distribution even in the asymptotic limit, and moreover, $\Theta_0$ contains singular points in general, which forces the estimation to be difficult.
Technically, the existence of singularities complicates the integral calculations that contribute to the marginal likelihood.
To address this issue, Watanabe proposed using desingularization of the KL divergence, inspired by the technique used to study the local zeta functions \cite{atiyah1970resolution,igusa2000introduction} in mathematics. 
The core concept of singular learning theory relies on an application of Hironaka's theorem on the resolution of singularities \cite{hironaka1964resolutionI,hironaka1964resolutionII} (Theorems \ref{thm:resolution_original} and 
 \ref{thm:resolution_Watanabe}). 
This theorem allows us to attribute the problem to integrals in spaces that do not contain singularities after taking the resolution of singularities. 

Through these observations, even in such a complex situation, that is, even if $q(x)$ is singular for $p(x|\theta)$, Watanabe's theory describes the asymptotic behaviors of the expectations of $G_n$ and $T_n$ as follows:
\begin{align}
    \mathbb{E}_{X^n}[G_n] &= - \mathbb{E}_X[ \log p(X|\theta_0) ] + \frac{\lambda}{n} + o\left(\frac{1}{n}\right), 
    \label{eq:exp_gen_loss} \\ 
    \mathbb{E}_{X^n}[T_n] &= - \mathbb{E}_X[\log p(X|\theta_0)] + \frac{\lambda-2\nu}{n}+ o\left(\frac{1}{n}\right), 
    \label{eq:exp_emp_loss}
\end{align}
where 
$\mathbb{E}_{X^n}[\cdot]$ represents the expectation with respect to $q(X^n) = \prod_{i=1}^n q(X_i)$.
The quantities $\lambda$ and $\nu$ are called the \textit{real log canonical threshold (RLCT)} and the \textit{singular fluctuation}, respectively; see also Proposition \ref{prop:singular for Gn and Tn} and Definition \ref{defn:invariants in singular learning theory}.
The former quantity $\lambda$ represents how bad the singularity is in algebraic geometry or minimal model program \cite{kollar1998birational,hacon2014acc}. 
The geometrical meaning of the latter one, $\nu$, is still unknown.
In the case a pair $(q(x), p(x|\theta))$ satisfies the regular and realizable conditions, it is known that $\lambda = \nu = d/2$, while they are generally different from each other for singular models.
These asymptotic expansion formulas are interesting in that they describe how different the generalization loss $G_n$ is from a computable training loss $T_n$ in view of algebraic geometrical quantities $\lambda$ and $\nu$. 
Beyond that, as a remarkable application of these observations, Watanabe \cite{watanabe2009algebraic,watanabe2010asymptotic} proposed an information criterion called $\mathrm{WAIC}$ (Widely Applicable Information Criterion) for singular models, 
\[\mathrm{WAIC}\coloneqq T_n + \frac{1}{n}\sum_{i=1}^n\mathbb{V}_{\theta}[\log p(x_i|\theta)],\] 
where the second term is defined by the posterior variance (Definition \ref{defn:posterior mean for classical}).
It is indeed an asymptotically unbiased estimator of $G_n$: 
\begin{align*}
    \mathbb{E}_{X^n}[G_n] = \mathbb{E}_{X^n}[\mathrm{WAIC}] + o\left(\frac{1}{n}\right).
\end{align*} 
As in the situations of AIC, WAIC has been applied in a wide range of areas: model selection \cite{vehtari2017practical,burkner2017brms,gronau2019limitations,watanabe2013waic}, anomaly detections \cite{choi2018waic}, analysis of missing data \cite{du2024comparing}, basics of Bayesian estimation \cite{gelman2014understanding,yao2018using}, and phylogenetics \cite{lartillot2023identifying} etc.

\section{Main results}\label{section:main results}

\subsection{Generalization and training losses in quantum state estimation}
\begin{figure}[t]
    \centering
    \includegraphics[width=0.45\textwidth]{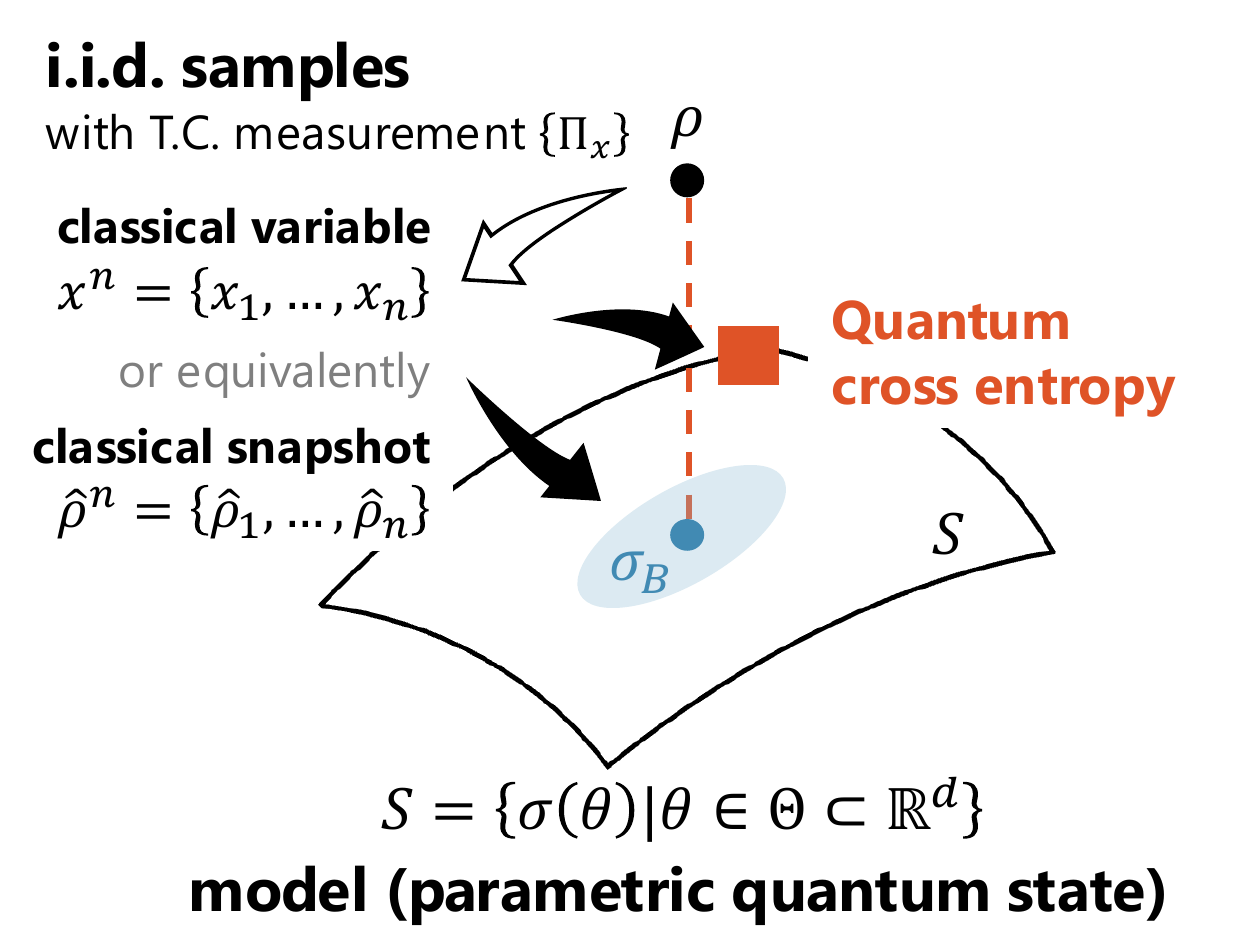}
    \caption{Our setting in quantum state estimation.}
    \label{fig:setting}
\end{figure}
In the present work, we restrict our attention to finite-dimensional systems and formulate the task of Bayesian quantum state estimation as follows (see Fig. \ref{fig:setting}). 
Let $\rho$ be an unknown true state and $x^n = \{ x_1, ..., x_n \}$ be a finite number of measurement data obtained through a tomographically complete (T.C.) measurement $\{\Pi_x\}$ (i.e. there exists $x$ such that $\Tr(\rho_A \Pi_x) \neq \Tr(\rho_B \Pi_x)$ for $\rho_A \neq \rho_B$) with uniform weights on $\rho$. 
In other words, the measurement data $x^n$ is a set of the i.i.d. samples from the corresponding true probability distribution $q(x) \coloneqq \Tr(\Pi_x \rho)$.
To predict $\rho$, one prepares a pair of a parametric quantum state model $\sigma(\theta)$ and a prior distribution $\pi(\theta)$.
Let us denote by $p(x|\theta) \coloneqq \Tr(\Pi_x \sigma(\theta))$ the corresponding probability distribution of the model $\sigma(\theta)$. 
In addition, we introduce an alternative representation of the measurement data using the classical shadow \cite{huang2020Predicting} by $\hat{\rho}^n \coloneqq \{ \hat{\rho}_{x_1}, ..., \hat{\rho}_{x_n} \}$ where $\hat{\rho}_{x_i}$ is called a classical snapshot, corresponding to a measurement outcome $x_i$ for each $i$.
In general, to construct the classical shadow, we repeatedly apply a random unitary $U$ taken from a set $\mathcal{U}$ to rotate the quantum state and perform a computational basis measurement to get the measurement outcome $\ket{b}$.
Two prominent examples are random Clifford measurement and random Pauli basis measurement.
After measurement, the measurement outcome is stored in the form of the classical snapshot $\hat{\rho}_{x} = \mathcal{M}^{-1}(U^\dagger \ketbra{b}{b} U)$ via the classical post-processing of applying the inverse of $U$ and the quantum channel $\mathcal{M}$ (corresponding to $\mathcal{U}$) to $\ketbra{b}{b}$, indicating that $x$ corresponds to $(U, b)$.
This formulation of the classical shadow facilitates efficient estimation of many non-commuting observables.
With the posterior distribution $p(\theta|x^n)$ defined in Section \ref{subsection:singular learning theory}, the posterior predictive quantum state, or simply the Bayesian mean, $\sigma_B$ is naturally defined by
    \[\sigma_B \coloneqq \int_\Theta \sigma(\theta) p(\theta|x^n) d\theta.\]
Now, instead of KL divergence in the classical learning theory, we use the quantum relative entropy between the two quantum states $\rho$ and $\sigma_B$:
\[
   D(\rho||\sigma_B) \coloneq \Tr\left[\rho (\log \rho - \log \sigma_B) \right],
\]
implicitly assuming $\mathrm{supp}(\rho) \subseteq \mathrm{supp} (\sigma_B)$.
Since the first term of $D(\rho||\sigma_B)$ does not depend on  $\sigma_B$, it is enough to evaluate the second term, which we call the \textit{quantum cross entropy} (QCE).
Using QCE, we further define the \textit{quantum generalization loss} $G_n^{Q}$ and \textit{training loss} $T_n^{Q}$: 
\begin{align}
\label{eq:G_n^Q and T_n^Q}
    G_n^{Q} \coloneqq - \Tr(\rho \log \sigma_B),  \quad
    T_n^{Q} \coloneqq - \frac{1}{n} \sum_{i=1}^{n} \Tr(\hat{\rho}_{x_i} \log \sigma_B).
\end{align}
Note that, while properly defining $T_n^Q$ is a non-trivial task, we opted to use the classical shadow due to its unbiasedness, i.e., $\mathbb{E}[\hat{\rho}] = \rho$.
We remark that $G_n^Q$ and $T_n^Q$ serve as our definitions of the quantum analogs of $G_n$ and $T_n$.

Before proceeding to our main results, we note that the Fundamental conditions (Fundamental conditions \ref{ass:fundamental condition} and \ref{ass:fundamental conditionII}) are assumed to hold throughout this paper as written in Appendix \ref{app:Fundamental conditions}.
They provide an appropriate framework for the Bayesian estimation in quantum singular models.
In Section \ref{subsec:main_result_Asymptotic_behaviors}, we begin by establishing the basic theorem for Bayesian quantum state estimation, which allows us to study the asymptotic behaviors of $G_n^Q$ and $T_n^Q$.
Next, we investigate the quantum generalization and training losses when our model is regular for the true state. 
In other words, we assume both the regularity condition between quantum states, where the Hessian matrix of the quantum relative entropy is positive definite, and the regularity condition between the associated probability distributions, where the Hessian matrix of the KL divergence is positive definite. 
For details on the assumptions in regular cases, refer to Assumptions \ref{ass:R1} and \ref{ass:R2}.
Subsequently, we discuss singular models using tools from algebraic geometry.
In such cases, we impose a weaker assumption than the regularity condition, namely a \textit{relatively finite variance}, introduced in classical singular learning theory.
For details on the assumptions in singular cases, refer to Assumptions \ref{ass:S1} and \ref{ass:S2}.
In Section \ref{subsec:QWAIC}, based on these asymptotic descriptions, we establish QWAIC.

\subsection{Asymptotic behaviors of \texorpdfstring{$G_n^{Q}$}{GnQ} and \texorpdfstring{$T_n^{Q}$}{TnQ}}
\label{subsec:main_result_Asymptotic_behaviors}
Investigating the asymptotic expansions of the quantum generalization loss $G_n^Q$ and training loss $T_n^Q$ defined in the previous subsection will provide valuable insights into the generalization performance in quantum state estimation.
Here, we show their asymptotic behaviors in both regular and singular cases.

As a starting point, analogous to the classical learning theory, examining cumulant generating functions in Bayesian statistics is essential for analysis, and the following theorem can be derived using these functions associated with $G_n^Q$ and $T_n^Q$.
It leads us to the following \textit{basic theorem} for the Bayesian quantum state estimation.
\begin{thm}[Basic theorem; informal version of Theorem \ref{thm:q_basic}]
\label{mainthm:expamsion of GnQ and TnQ}
    The generalization loss $G_n^{Q}$ and training loss $T_n^{Q}$ can be expanded as follows:
    \begin{align}
        G_n^Q &= - \Tr(\rho \mathbb{E}_\theta[\log \sigma(\theta)]) - \frac{1}{2} \Tr(\rho \mathbb{V}_{\theta}[\log \sigma(\theta)]) + o_p\left(\frac{1}{n}\right), 
        \label{eq:main_basic_theorem_G_n^Q}\\
        T_n^Q &= - \Tr( \left(\frac{1}{n}\sum_{i=1}^{n} \hat{\rho}_{x_i}\right) \mathbb{E}_\theta[\log \sigma(\theta)]) - \frac{1}{2} \Tr(\left(\frac{1}{n}\sum_{i=1}^{n} \hat{\rho}_{x_i}\right) \mathbb{V}_{\theta}[\log \sigma(\theta)]) + o_p\left(\frac{1}{n}\right).
        \label{eq:main_basic_theorem_T_n^Q}
    \end{align}
\end{thm}
Here, $\mathbb{E}_\theta[\cdot]$ and $\mathbb{V}_{\theta}[\cdot]$ are used to denote the posterior mean and variance for matrices, respectively; see Definition \ref{defn:posterior mean for classical} for the precise definition.
For sequences of random variables $X_n$ and $\epsilon_n$, the order in probability notation $X_n = o_p(\epsilon_n)$ means that $X_n/\epsilon_n$ converges to zero in probability as $n \to \infty$.
\begin{proofsketch}
    Let us introduce a cumulant generating function $s^Q(\hat{\rho}, \alpha) = \Tr(\hat{\rho} \log \mathbb{E}_\theta[\sigma(\theta)^\alpha])$ by noting the fact that the expectation and empirical sum of $s^Q(\hat{\rho},1)$ corresponds to $G_n^Q$ and $T_n^Q$, respectively.
    Then we analyze the cumulants of $\Tr(\hat{\rho} \log \mathbb{E}_\theta[\sigma(\theta)])$ with respect to the posterior distribution.
    The first terms of Eqs. \eqref{eq:main_basic_theorem_G_n^Q} and \eqref{eq:main_basic_theorem_T_n^Q} are derived from the first cumulant, while the second terms are derived from the second cumulant.
    The full proof is in Theorem \ref{thm:q_basic}, whose assumptions on the scaling of the higher-order terms are separately proved in Lemma \ref{lem:higher_order_scaling_regular} for regular cases and Lemma \ref{lem:higher_order_scaling_singular} for singular cases, respectively. 
\end{proofsketch}
Note that Theorem \ref{mainthm:expamsion of GnQ and TnQ} holds for both regular and singular cases, and thus serves as a basic theorem for further expansions, where regular or singular conditions will later be imposed.

Under the regularity condition, we obtain the following detailed description of the asymptotic expansions:
\begin{thm}[Regular asymptotic expansion; informal version of Theorem \ref{thm:q_expectations for regular cases}]
\label{mainthm:quantum expansion regular}
    Let $\theta_0$ be the unique element of $\Theta_0$.
    When a pair of a parametric quantum state model $\sigma(\theta)$ and a true state $\rho$ satisfies the classical and quantum regularity conditions, 
    the expectations of the generalization loss $G_n^Q$ and training loss $T_n^Q$ can be expanded as     
        \begin{align}
            \mathbb{E}_{X^n}[G_n^Q] &= - \Tr(\rho \log \sigma(\theta_0)) + \frac{1}{n}\left( \lambda^Q + \nu'^Q - \nu^Q \right) + o\left(\frac{1}{n}\right), \\
            \mathbb{E}_{X^n}[T_n^Q] &= - \Tr(\rho \log \sigma(\theta_0)) + \frac{1}{n}\left( \lambda^Q + \nu'^Q - 2\chi^Q - \nu^Q \right) + o\left(\frac{1}{n}\right),
        \end{align}
    with constants $\lambda^Q$, $\nu^Q$, and $\nu'^Q$, and $\chi^Q = c + o(1)$ for a constant $c$.
\end{thm}
Note that $\theta_0$ is a parameter that minimizes the KL divergence defined in Section \ref{subsection:singular learning theory}.
The specific expressions for $\lambda^Q$ and $\nu^Q$ are given in Corollary \ref{cor:expansion of GnQ and TnQ for regular cases}, and those for $\nu'^Q$ and $\chi^Q$ in Theorem \ref{thm:q_expectations for regular cases}.
\begin{proofsketch}
    We consider the Taylor expansion of a function $\Tr(\hat{\rho} \log \sigma(\theta))$ around $\theta_0$ to the second-order and plug it in Eqs. \eqref{eq:main_basic_theorem_G_n^Q} and \eqref{eq:main_basic_theorem_T_n^Q} in Theorem \ref{mainthm:expamsion of GnQ and TnQ}.
    Then, we utilize the convergence of the posterior distribution to evaluate the asymptotic properties of the parameter estimation.
    For the expansion of $T_n^Q$, a quantum analog of the empirical process is utilized to deal with the fluctuation of a function $\Tr(\hat{\rho} \log \sigma(\theta))$.
    Taking the expectation $\mathbb{E}_{X^n}[\cdot]$ completes the proof.
    The full proof consists of Theorem \ref{thm:q_regular expansion formula for G and T} and Theorem \ref{thm:q_expectations for regular cases}.
\end{proofsketch}
We remark that $\lambda^Q$ is exactly the ratio of quantum Fisher information to classical Fisher information in the realizable case, where there exists a parameter $\theta \in \Theta$ such that $\sigma(\theta) = \rho$, which is also observed in \cite{yano2023Quantuma}.
The other quantities $\nu^Q$, $\nu'^Q$, and $\chi^Q$ should also be related to these values.
The analysis in regular cases offers valuable insights and interpretations of the generalization error in quantum state estimation from a quantum information theoretic perspective.

However, the regularity condition is frequently no longer satisfied in practical situations, in particular, if our quantum state model has many parameters. 
Thus, we will consider our quantum state model as a singular model in general. 
Below, we shall introduce the method to analyze the singular models in quantum state estimation.

To study the behavior of these losses for quantum singular models, let us formulate a geometric setting.
First, we introduce a parameter set whose parameter minimizes the quantum relative entropy, analogous to $\Theta_0$, as 
\[\Theta_0^Q \coloneq \left\{\theta_0^Q\in\Theta \Big\vert\ \theta_0^Q = \argmin_{\theta \in \Theta} D(\rho||\sigma(\theta)) \right\},\]
where we denote by $\theta_0^Q$ an element of this set.
Let us define the \textit{average log loss function} and \textit{average quantum log loss function}
\begin{align}
    K(\theta) &\coloneq \mathrm{KL}(p(x|\theta_0)\|p(x|\theta)) = \mathrm{KL}(\Tr(\Pi_x \sigma(\theta_0))\|\Tr(\Pi_x \sigma(\theta))), \label{eq:average log loss function}\\
    K^{Q}(\theta) &\coloneqq D(\sigma(\theta_0^Q)\|\sigma(\theta)), \label{eq:average quantum log loss function}
\end{align}
according to \cite{watanabe2018mathematical}.
In singular cases, $\Theta_0$ and $\Theta_0^Q$ generally contain singular points, making the analysis difficult.
In our theory, we consider the desingularization for both the KL divergence $K$ and the quantum relative entropy $K^Q$.
See Fig. \ref{fig:blowup} for an example of the desingularization of the nodal curve $y^2=x^2(x+1)$.

\begin{figure}[t]
    \centering
    \includegraphics[width=0.65\textwidth]{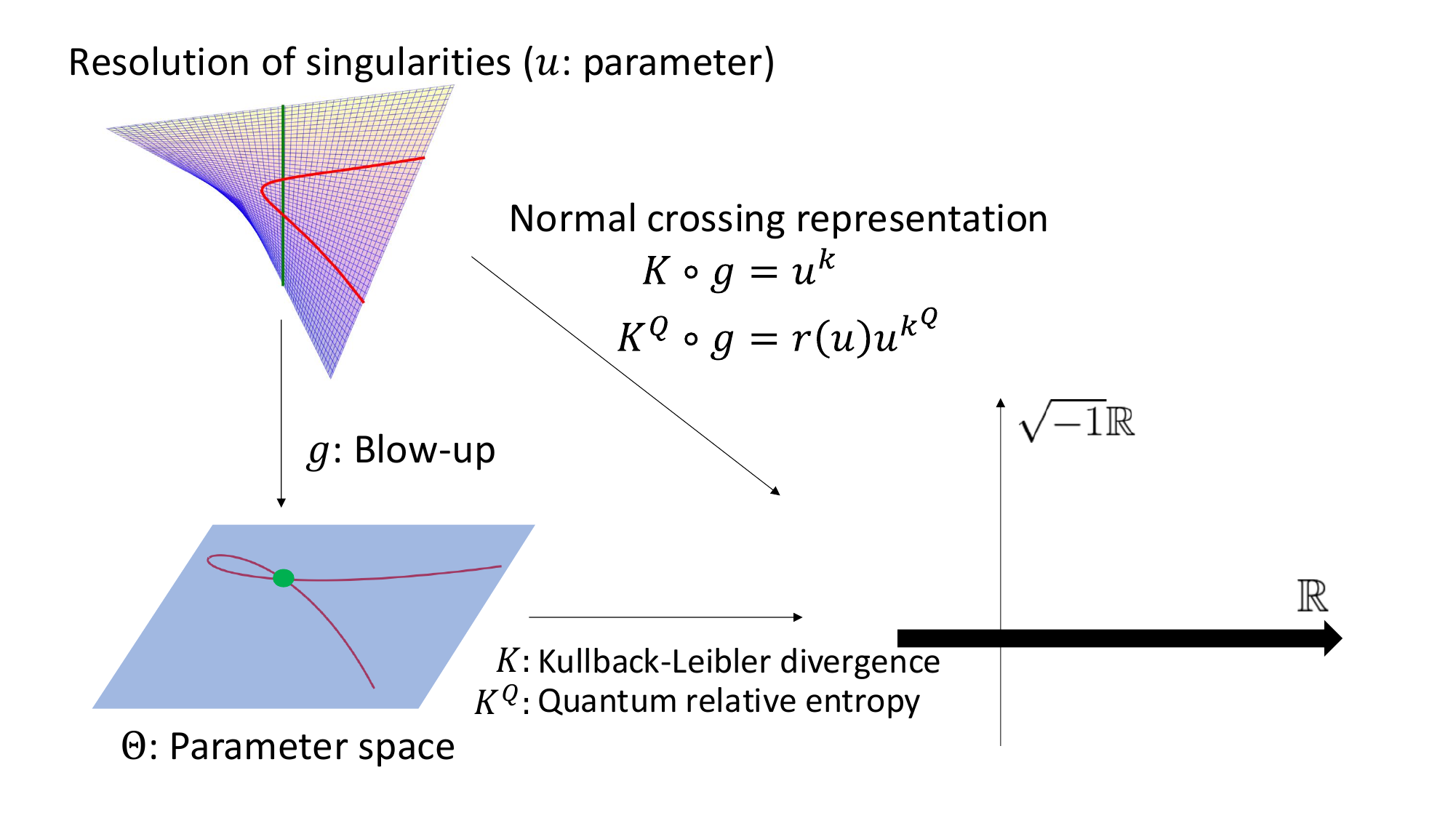}
    \caption{Resolution of singularities of a parameter space. 
    Technically, obtaining the normal crossing representation requires blowups to be repeated; however, to better convey the essence of the desingularization, a single application of the blowup is depicted in this figure.}
    \label{fig:blowup}
\end{figure}

Within our theory, we simultaneously resolve the singularities that appear in these two sets. More precisely, we see that there is a proper holomorphic morphism $g:\widetilde{\Theta}\to\Theta$ so that $g^{-1}(\Theta_0)$ and $g^{-1}(\Theta_0^Q)$ are simple normal crossing divisors; that is, the average loss functions are rewritten as 
\begin{equation}
\label{eq:normal_crossing_rep}
    K(\theta) = K(g(u)) = u^{2k}, \quad K^{Q}(\theta) = K^{Q}(g(u)) = r(u)u^{2k^{Q}}, 
\end{equation}
for $k, k^{Q}\in\mathbb{Z}^d$ by using multi-indices, a parameter $u$ of $\widetilde{\Theta}$ (Theorem \ref{thm:resolution_Watanabe}) and a real-analytic function $r(u)$ with $r(0)\neq 0$.
We remark that these normal crossing representations allow us to define the posterior distribution and empirical processes on the parameter space $\widetilde{\Theta}$, which also play crucial roles in formulating singular learning theory. 
These expressions generalize the quadratic representations around the optimal parameters when the Fisher information matrix degenerates.
Based on the normal crossing representations of both the KL divergence and the quantum relative entropy, we can derive the following asymptotic expansions, taking into account  singularities:

\begin{thm}[Singular asymptotic expansion; informal version of {Theorem \ref{thm:q_expansion formulas for the expectations for singular cases}}]
\label{mainthm:quantum expansion singular}
    Let $\theta_0$ be any element of $\Theta_0 \cap \Theta_0^Q$.
    Even when a pair of a parametric quantum state model $\sigma(\theta)$ and a true state $\rho$ satisfy neither the classical nor quantum regularity conditions, 
    the expectations of the generalization loss $G_n^Q$ and training loss $T_n^Q$ can be expanded as
    \begin{align*}
        \mathbb{E}_{X^n}[G_n^Q] &= - \Tr(\rho \log \sigma(\theta_0)) + \frac{1}{n}\left( r_{CQ} \lambda + r_{CQ} \nu - \nu^Q \right) + o\left(\frac{1}{n}\right), \\
        \mathbb{E}_{X^n}[T_n^Q] &= - \Tr(\rho \log \sigma(\theta_0)) + \frac{1}{n}\left( r_{CQ} \lambda + r_{CQ} \nu - 2\chi^Q - \nu^Q \right) + o\left(\frac{1}{n}\right).
    \end{align*}
    with constants $\lambda$, $\nu$, and $r_{CQ}$, and $\nu^Q = c_1 + o(1)$, and $\chi^Q = c_2 + o(1)$ for constants $c_1$ and $c_2$.
\end{thm}
The quantities $\lambda$ and $\nu$ are defined by the average log loss function $K(\theta)$ as in singular learning theory (Definition \ref{defn:invariants in singular learning theory}).
The specific expressions for $\nu^Q$ and $\chi^Q$ are introduced in Theorem \ref{thm:q_expansion formulas for the expectations for singular cases}.
\begin{proofsketch}
    First, we construct a renormalized log-likelihood function defined in the space $\widetilde{\Theta}$.
    This allows us to expand a function $\Tr(\hat{\rho} \log \sigma(\theta))$ in Theorem \ref{mainthm:expamsion of GnQ and TnQ} with respect to a parameter $u$.
    Then, as in regular cases, we utilize the convergence of renormalized posterior distribution to evaluate the asymptotic properties of parameter estimation.
    For the expansion of $T_n^Q$, we also introduce a quantum analog of the renormalized empirical process.
    Taking the expectation $\mathbb{E}_{X^n}[\cdot]$ completes the proof.
    The full proof consists of Lemma \ref{lem:intro_a^Q}, Theorem \ref{thm:q_expansion formulas of G^Q and T^Q for singular cases}, and Theorem \ref{thm:q_expansion formulas for the expectations for singular cases}.
\end{proofsketch}
Again, $\lambda$ and $\nu$ are called the real log canonical threshold and the singular fluctuation already introduced in Section \ref{subsection:singular learning theory}, both of which do not depend on the choice of a pair of $(\widetilde{\Theta}, g)$. 
The remaining quantities are those that newly arise in the quantum setting.
The quantity $r_{CQ}\lambda$ generalizes $\lambda^Q$.
This comes from the fact that in the presentation of the normal crossing presentations, the ratio of the classical and quantum Fisher information matrices is encoded in the real analytic function $r(u)$.
$\chi^Q$ can be regarded as a singular fluctuation, especially for quantum state estimation, whereas the geometric interpretation of $\nu^Q$ is not studied in this work.

We observe that, from Theorems \ref{mainthm:quantum expansion regular} and \ref{mainthm:quantum expansion singular}, the difference of the expectations of the quantum generalization and training loss, $\mathbb{E}_{X^n}[G_n^Q - T_n^Q]$, is $2\chi^Q/n$ up to $o(1/n)$ in both regular and singular cases.
Since $T_n^Q$ is a quantity that can be computed from the data, if there is an estimator for $2\chi^Q/n$ up to $o(1/n)$, an asymptotically unbiased estimator for $G_n^Q$ can be constructed.
This is precisely what we introduce in the next as QWAIC.

\subsection{Quantum Widely Applicable Information Criterion}
\label{subsec:QWAIC}
Let us first define QWAIC (Quantum Widely Applicable Information Criterion) 
\begin{align}
    \mathrm{QWAIC} &\coloneqq T_n^{Q} + C_n^Q, \\
    C_n^Q &\coloneqq \frac{1}{n} \sum_{i=1}^{n} \mathrm{Cov}_{\theta}\left[ \log p (x_i|\theta), \Tr(\hat{\rho}_{x_i} \log \sigma(\theta)) \right], 
\end{align}
where $\mathrm{Cov}_{\theta}[\cdot,\cdot]$ is the posterior covariance; see also Definition \ref{def:QWAIC}.
This can be seen as a quantum extension of WAIC defined in Section \ref{subsection:singular learning theory}, by introducing a quantum state model $\sigma(\theta)$ and a classical snapshot $\hat{\rho}$.
We will discuss the computational efficiency issue caused by the matrix logarithm $\log \sigma(\theta)$ in Section \ref{sec:Conclusion and open questions}.
With this definition, we can prove that QWAIC is an asymptotically unbiased estimator for $G_n^Q$ in both regular and singular cases.
\begin{thm}[Asymptotic unbiasedness of QWAIC; informal version of {Theorems \ref{thm:q_QWAIC for regular cases} and  \ref{thm:q_QWAIC is unbiased estimator for singular cases}}]
\label{mainthm:QWAIC}
    Even when a pair of a parametric quantum state model $\sigma(\theta)$ and a true state $\rho$ satisfy neither the classical nor quantum regularity conditions, $\mathrm{QWAIC}$ is an asymptotically unbiased estimator of $G_n^{Q}$. 
    In other words, 
    \begin{equation}
        \mathbb{E}_{X^n}[G_n^{Q}] = \mathbb{E}_{X^n}[\mathrm{QWAIC}] + o\left(\frac{1}{n}\right)
    \end{equation}
    holds.
\end{thm}
\begin{proofsketch}
    This theorem holds for both regular and singular cases.
    For both cases, it suffices to prove $\chi^Q = (1/2) \mathbb{E}_{X^n}[ C_n^Q ] + o(1)$ by utilizing the properties of the empirical processes.
    The full proof is in Theorem \ref{thm:q_QWAIC for regular cases} for regular cases and Theorem \ref{thm:q_QWAIC is unbiased estimator for singular cases} for singular cases.
\end{proofsketch}
In practical applications, WAIC can be employed for model selection among singular statistical models.
We expect QWAIC to play the same role in model selection among singular quantum state models, which is one of our greatest interests.
However, in the next section, we will instead focus on investigating the non-trivial quantities $\lambda$, $\lambda^Q$, $\nu$, $\nu^Q$, $\nu'^Q$, and $r_{CQ}$ in our main expansion formulas (Theorms \ref{mainthm:quantum expansion regular} and \ref{mainthm:quantum expansion singular}) and leave the study of model selection for future work \footnote{In our previous work, we numerically demonstrated that QAIC, only valid in regular cases, can serve as a tool for quantum state model selection.
For interested readers, refer to \cite{yano2023Quantuma}.}.
It is also an intriguing problem to define an appropriate complexity measure for the task of quantum state estimation that captures the influence of singularities of quantum state models.
In the original singular learning theory, the RLCT $\lambda$ is considered as a complexity measure that reflects the model's intrinsic geometric properties and determines the speed of learning according to the asymptotic expansion formula of $\mathbb{E}_{X^n}[G^n]$ in Section \ref{subsection:singular learning theory}.
For this reason, we acknowledge that $\lambda$ is sometimes referred to as the learning coefficient.
On the other hand, from the asymptotic expansion formula of $\mathbb{E}_{X^n}[G_n^Q]$ in Theorems \ref{mainthm:quantum expansion regular} (resp. \ref{mainthm:quantum expansion singular}), the quantity $\lambda^Q + \nu'^Q - \nu^Q$ (resp. $r_{CQ}\lambda + r_{CQ}\nu - \nu^Q$) determines the speed of learning for quantum regular (resp. singular) models.
Therefore, throughout this paper, we refer to $\lambda^Q + \nu'^Q - \nu^Q$ (or $r_{CQ}\lambda + r_{CQ}\nu - \nu^Q$) as the learning coefficient in quantum state estimation.

\section{Concrete examples}
\label{sec:concrete_examples}
Our main results are described by various quantities.
Such quantities include ($\lambda^Q, \nu'^Q$) first appeared in quantum state estimation as well as ($\lambda, \nu$) in singular learning theory.
In Section \ref{subsec:analytic_calculation}, we explicitly present what values they take through concrete models.
The purpose of this section is not to conduct large-scale computations involving realistic quantum state models but to introduce the calculations and assumptions of the mathematical tools that appear in our paper and to provide the reader with a concrete picture.
As a result, we show the concrete value of the learning coefficient for regular models in Example \ref{ex:Regular situations}.
Notably it differs from the one for classical notion, that is the second term of AIC, even though it is a classical system model.
In Section \ref{subsec:numerical_simulation}, we numerically compute QWAIC and check consistency with the analytic results.
This supports the results that QWAIC is an asymptotically unbiased estimator for $G_n^Q$.

\subsection{Analytic and algebraic calculation}
\label{subsec:analytic_calculation}
We illustrate an instance of regular cases in Example \ref{ex:Regular situations} and an instance of singular cases in Example \ref{ex:analy-1Q-depol-2-b}.
In these two examples, only classical states are considered for the simplicity of calculation, though enough to demonstrate the essence of the resolution of singularities via an algebraic geometrical method.
Finally, in Example \ref{ex:Singular situations}, we address a quantum state model (but the true state is still classical). 
We note that similar calculations should be performed when the true state is a quantum state or the system size is large, but since the calculation of the explicit form of $K^Q$ becomes complicated and there is little theoretical improvement, the discussion here is dedicated to these settings.

We first calculate the average log loss functions $K(\theta)$ (Eq. \eqref{eq:average log loss function}) and $K^Q(\theta)$ (Eq. \eqref{eq:average quantum log loss function}) to check whether the regularity condition is satisfied.
For simplicity, we work on the 1-qubit parametric quantum models that satisfy the realizability condition (Definition \ref{defn:realizability}).
Notably, we will confirm that Conjectures \ref{conj:qregular to cregular} and \ref{conj:k=k^Q} hold in all the examples.
For the computation of the average log loss functions, we use the Pauli basis measurement $\{\Pi_{m}\}_{m=1}^{6}$ with $6$ outcomes, i.e., 
\begin{align*}
    \Pi_{1} = \frac{1}{3} \ketbra{0}{0}, \,
    \Pi_{2} = \frac{1}{3} \ketbra{1}{1}, \,
    \Pi_{3} = \frac{1}{3} \ketbra{+}{+}, \,
    \Pi_{4} = \frac{1}{3} \ketbra{-}{-}, \,
    \Pi_{5} = \frac{1}{3} \ketbra{r}{r}, \,
    \Pi_{6} = \frac{1}{3} \ketbra{l}{l}, 
\end{align*}
where $\{\ket{0}, \ket{1}\}$, $\{\ket{+}, \ket{-}\}$, and $\{\ket{r}, \ket{l}\}$ are the eigenbases of the Pauli operators $Z$, $X$, and $Y$, respectively.

\vspace{3mm}
\begin{ex}[A regular case; a classical state model]
\label{ex:Regular situations}
    First, we work on a regular case 
    \[\rho = \frac{I}{2},\quad \sigma(\theta) = \begin{pmatrix}
    \cos^2(\theta) & 0\\
    0 & \sin^2(\theta) \\
    \end{pmatrix},\quad  \theta \in \Theta\coloneqq \left[0,\frac{\pi}{2}\right].
    \]
    Then, since $\rho = \sigma (\pi/4)$, the realizability condition is satisfied.
    The average log loss functions are given by 
    \begin{align*}
        K^Q(\theta) &= -\frac{1}{2}(\log(\sin^2(\theta)\cos^2(\theta)) + 2\log(2) ), \\
            K(\theta) &= -\frac{1}{6}(\log(\sin^2(\theta)\cos^2(\theta)) + 2\log(2) ).
        \end{align*}
        Note that we have $K^Q(\theta) = 3 K(\theta)$, which coincides with the situation in Proposition \ref{prop:Liouville and k=k^Q}.    
    Since the Hessians of these functions are positive, that is 
    \[\frac{d^2}{d\theta^2}K^Q(\theta)>0,\quad \frac{d^2}{d\theta^2}K(\theta)>0,\]
    this example satisfies the classical and quantum regularity conditions.
    Though Remark \ref{rem:regular_RLCT} with $d=1$ implies that $\lambda=1/2$ in this case, we show an explicit computation of RLCT in detail for readers' understanding.   
    In this paper, we introduce the statement of Hironaka's theorem in detail in Appendix \ref{app:algebraic geometry} for algebraic varieties, that is, a set defined as a zero set of polynomials,
    but a similar argument holds for analytic spaces.
    Since we can check that
    \[K\left(\frac{\pi}{4}\right) = 0,\quad \frac{d}{d\theta} K(\theta)\Bigm\vert_{\theta = \frac{\pi}{4}} = 0,\quad  \frac{d^2}{d\theta^2} K(\theta)\Bigm\vert_{\theta = \frac{\pi}{4}} \neq 0,\]
    the average log loss function $K(\theta)$ can be rewritten as 
        \[K(\theta) = \left(\theta-\frac{\pi}{4}\right)^2 b(\theta)\]
        for an analytic function $b(\theta)$ with $b(\pi/4) \neq 0$.
    In this case, since $\Theta=[0,\pi/2]$ is one dimensional, that is $d=1$, the set $\Theta_0 = \{\pi/4\}$ is already considered as a simple normal crossing divisor.
    In other words, if we take a variable transformation $\iota:\mathbb{R}_{u}\xrightarrow{\sim}\mathbb{R}_{\theta}$ with $\iota(u) \coloneqq u + \pi/4 = \theta$, then
    \begin{equation}
    \label{eq:example_A_regular}
        K(\iota(u)) = u^2\cdot b(\iota(u)).
    \end{equation}
    This is exactly a normal crossing representation.
    Hence, we can compute the RLCT without taking a resolution of singularities, which forces $h=0$ since the Jacobian of the identity map is trivial.
    Comparison of Eq. \eqref{eq:example_A_regular} with Eq. \eqref{eq:normal_crossing_rep} implies that $2k=2$
    \footnote{In analytic or geometric notion, for a real analytic function $f$, we denote by $\mathrm{div}(f)\coloneqq \sum_{i} a_i V_i$ where $V_i$ is a divisorial irreducible component of $f=0$, and $a_i$ is an associated vanishing order.
    Accordingly, we have $\mathrm{div}(K) = 2 \Theta_0$.}.
    Note that the existence of the real analytic function $b(\iota(u))$ does not affect the computation of $\lambda$ because it satisfies $b(\iota(0)) \neq 0$, that is, the Taylor expansion of $b(\iota(u))$ around $u=0$ has a non-zero constant term.  
    It concludes that the RLCT is
    \[\lambda = \frac{1}{2},\]
    which is consistent with Proposition \ref{prop:generalization loss of G_n and T_n}. 
    Since $\lambda = \nu$ in the regular and realizable condition, singular fluctuation is given by 
    \[\nu = \frac{1}{2}.\]
    In the above, the algebraic computation is demonstrated for $\lambda$ and $\nu$, although the analytic computation for the classical Fisher information matrix gives the same result.
    For $\lambda^Q$ and $\nu'^Q$, which appear only in regular cases, the explicit calculation of classical and quantum Fisher information matrices leads to
    \[\lambda^Q = \nu'^Q = 3.\]
    Also, we deduce 
    \[\nu^Q = 4\]
    from the explicit computation of the Fisher information matrix.
    Hence, the learning coefficient in this case is 
    \[\lambda^Q + \nu'^Q - \nu^Q = 2.\]
    Since the state model and the true state are classical, we can just consider this example as a problem in the classical learning theory.
    Contrary to the learning coefficient in the quantum case, the classical one is $\lambda = 1/2$ with $d=1$ as described in Section \ref{subsection:singular learning theory}.
\end{ex}

\begin{ex}[a singular case; a classical state model]
\label{ex:analy-1Q-depol-2-b}
    Now, we take an example such that $\dim \Theta_0^Q >0$, which forces the situation to become singular, i.e., it violates Definition \ref{def:quantum_regularity} (2).
    Let us define
    \[\sigma(\theta) \coloneqq \begin{pmatrix}
        \cos^2(f(\theta_1, \theta_2) + \pi/3) & 0\\
        0 & \sin^2(f(\theta_1, \theta_2) + \pi/3)
    \end{pmatrix},\quad f(\theta_1,\theta_2)\coloneqq \theta_1-\theta_2 .\]
    Take a true state $\rho \coloneqq \sigma(0,0)$ within a parameter space \[\Theta = \left\{(\theta_1,\theta_2)\in \mathbb{R}^2\ \middle|\ -\frac{\pi}{3} \leq \theta_1 - \theta_2 \leq \frac{\pi}{2}, -\pi \leq \theta_1 \leq \pi\right\},\]
    which is a compact set $\Theta\subset \mathbb{R}^2$, we can consider the realizable situation so that
    \[\Theta_0 = \Theta_0^Q = \{(\theta_1,\theta_2) \in \Theta\mid \theta_1 = \theta_2\},\]
    where the quantum relative entropy and the KL divergence are
    \begin{align}
        K^Q(\theta) &= - \frac{3}{4} \log(\sin^2(\theta_1-\theta_2+\frac{\pi}{3})) - \frac{1}{4} \log(\cos^2(\theta_1-\theta_2+\frac{\pi}{3})) - 2 \log(2) + \frac{3\log(3)}{4}, \\
        K(\theta) &= \frac{1}{3} K^Q(\theta).
    \end{align}
    Hence, the model is singular in both classical and quantum senses.
    In this case, since $\Theta_0 = \Theta_0^Q$ is already a simple normal crossing divisor, there is no need to take a log resolution to compute the invariants appearing in QWAIC.
    In fact, taking a variable transformation $x = \theta_1 - \theta_2$, $y = \theta_1 + \theta_2$, which induces an isomorphism $\iota:\mathbb{R}^2_{x,y} \xrightarrow{\sim}  \mathbb{R}^2_{\theta_1,\theta_2}$, the set of optimal parameters $\Theta_0$ is defined by the equation $(x = 0)$ in $\Theta$.
    Through this isomorphism, we can rewrite
    \[K(x,y) \coloneqq K\circ \iota (x,y) = x^2 \cdot b(x,y)\]
    where $b(x,y)$ is an analytic function with $b(0,0) \neq 0$.
    The identity map $\Theta \to \Theta$ being already a log resolution, the quantity $2k=2k^Q$ is 2 and a corresponding Jacobian is represented by a constant as Example \ref{ex:Regular situations}.
    Combining these observations, we obtain
    \[\lambda = \frac{1}{2}.\]
    This is consistent with the numerical experiments executed in Section \ref{subsec:numerical_simulation}.
    If the model were regular, it would be concluded as $\lambda = d/2 = 1$ with $d=2$ only from the information of the dimension of the parameter space.
    However, our computation above shows that it is incorrect, compatible with the fact that the model is singular.
    The computation of the quantum-related quantities in singular cases, such as $r_{CQ}$, is non-trivial, and we leave this for future work.
        Note that since $K^Q(\theta) = 3K(\theta)$ in this case, the function $r(u)$ is the constant function $3$, which implies that 
    \[r_{CQ} = 3.\]
\end{ex}

\begin{ex}[A singular case; a quantum state model]
\label{ex:Singular situations}
    While both the true state and the state model are classical states in the previous two examples, this example addresses a quantum state model, although the true state is still a classical state.
    For smooth functions $f$ and $g$, let 
    \[\rho = \frac{I}{2} = \begin{pmatrix}
         \frac{1}{2}  & 0\\
         0& \frac{1}{2}
    \end{pmatrix},\quad \sigma(\theta) = f(\theta_2)\ketbra{\phi(\theta_1)}{\phi(\theta_1)} + (1-f(\theta_2)) \frac{I}{2},\]
    where 
    \[\ket{\phi(\theta_1)} = \begin{pmatrix}
        \cos(g(\theta_1)) \\
        \sin(g(\theta_1))
    \end{pmatrix}.\]
    Note that $\sigma(\theta)$ can be regarded as a mixed state after applying the depolarizing channel, a famous example of quantum noise channels.
    Furthermore, we consider highly singular optimal parameter sets.
    Assuming that there exists a $\theta_2\in\Theta$ so that $f(\theta_2)=0$, this is an analog of a normal mixture.
    In particular, the realizability condition is satisfied and hence, we have
    \[\Theta_0 = \Theta_0^Q = \{\theta = (\theta_1,\theta_2)\in\mathbb{R}^{d_1+d_2}\mid f(\theta_2) =0 \},\]
    including singular points in general.
    
    In singular learning theory, variables must be transformed because the average log loss functions are not analytic along the endpoint, which realizes the true distribution.
    In other words, we cannot treat this model as itself because it does not satisfy the fundamental condition.
    To overcome this difficulty, it is useful to change the variables and consider them as the coordinates of the hyperspherical surface; see \cite[Remark 6.1 (8)]{watanabe2009algebraic}. 
    As an analog of the discussion in Example \ref{ex:normal mixture}, we shall replace the model with 
    \[\sigma(\theta) = \sin^2(f(\theta_2))\ketbra{\phi(\theta_1)}{\phi(\theta_1)} + \cos^2(f(\theta_2)) \frac{I}{2}.\]
    Note that with this change, the description of $\Theta_0$ and $\Theta_0^Q$ remains the same.
    
    Now, we shall compute algebraic invariants $k,k^Q,\lambda$ appearing in Theorem \ref{mainthm:QWAIC}.
    For this purpose, let us see the degree of decrease around the zero of $f(\theta_2)$.
    Below, we use the relation $0 \cdot \log 0 = 0$.
    On the one hand, the quantum relative entropy in this case is
    \begin{align*}
     K^Q(\theta) = D(\rho||\sigma(\theta)) =
    -\frac{1}{2}(\log(2-\cos^2(f(\theta_2))) +\log(\cos^2(f(\theta_2))))
    \end{align*}
    This expression implies that $K^Q(\theta)$ can be considered as $f(\theta_2)^4$ near the locus $\Theta^Q_0$.
    
    Now, let us move the computation of $K$.
    Then, on the other hand, the KL divergence is calculated as 
    \begin{align*}
    K(\theta) &= 
        \mathrm{KL}(q(x)||p(x|\theta))\\
        &=
    -\frac{1}{6}\{\log(2\sin^2(g(\theta_1))\sin^2(f(\theta_2)) + \cos^2(f(\theta_2)))+ \log(2\cos^2(g(\theta_1))\sin^2(f(\theta_2)) + \cos^2(f(\theta_2)))\\ &+ \log(-2\sin(g(\theta_1))\sin^2(f(\theta_2))\cos(g(\theta_1)) + 1)+ \log(2\sin(g(\theta_1))\sin^2(f(\theta_2))\cos(g(\theta_1)) + 1)\}
    \end{align*}
    This description implies that the KL divergence $K(\theta) = \mathrm{KL}(q(x)||p(x|\theta))$ also decreases like $f(\theta_2)^4$.
    Hence, like $K^Q$, it follows that $K$ is also approximated as $f(\theta_2)^4$ near $\Theta_0$ which concludes that the quantum relative entropy $D(\rho||\sigma(\theta))$ and the KL divergence $\mathrm{KL}(q(x)||p(x|\theta))$ have the same vanishing order at $\Theta_0^Q=\Theta_0$.
    Hence, the depolarizing noise model satisfies Fundamental conditions \ref{ass:fundamental condition} and \ref{ass:fundamental conditionII}, and Assumption \ref{ass:S2}.
    
    Now we shall apply the above to the two famous examples in algebraic geometry.
    They have the meaning in singular learning theory that they are in fact, singular models.
    Below, for simplicity, let us put $g$ be a constant function so that $d_1=0$.
    First case is the multiparameter function  $f(\theta_2) = \theta_{21}^2 + \cdots + \theta_{2d_2}^2$ where $\theta_2 = (\theta_{21},\cdots, \theta_{2d_2})\in\mathbb{R}^{d_2}$, which forces $\Theta_0 = \{(0,\cdots,0)\}$.
    Then we can use a log resolution of singularities taken in Example \ref{ex:primitive blowup example}.
    Note that the vanishing orders of $K$ and $K^Q$ at $o=(0,\cdots,0)$ are 4 and hence, $2k_2=2k_2^Q = (0,\cdots, 0,8,0,\cdots,0)$ on each affine open subset taken there.
    The computation associated with Jacobians remains the same, which concludes that 
    \[\lambda = \frac{1}{4}.\]
    
    \begin{figure}[t]
        \centering
        \includegraphics[width=0.65\textwidth]{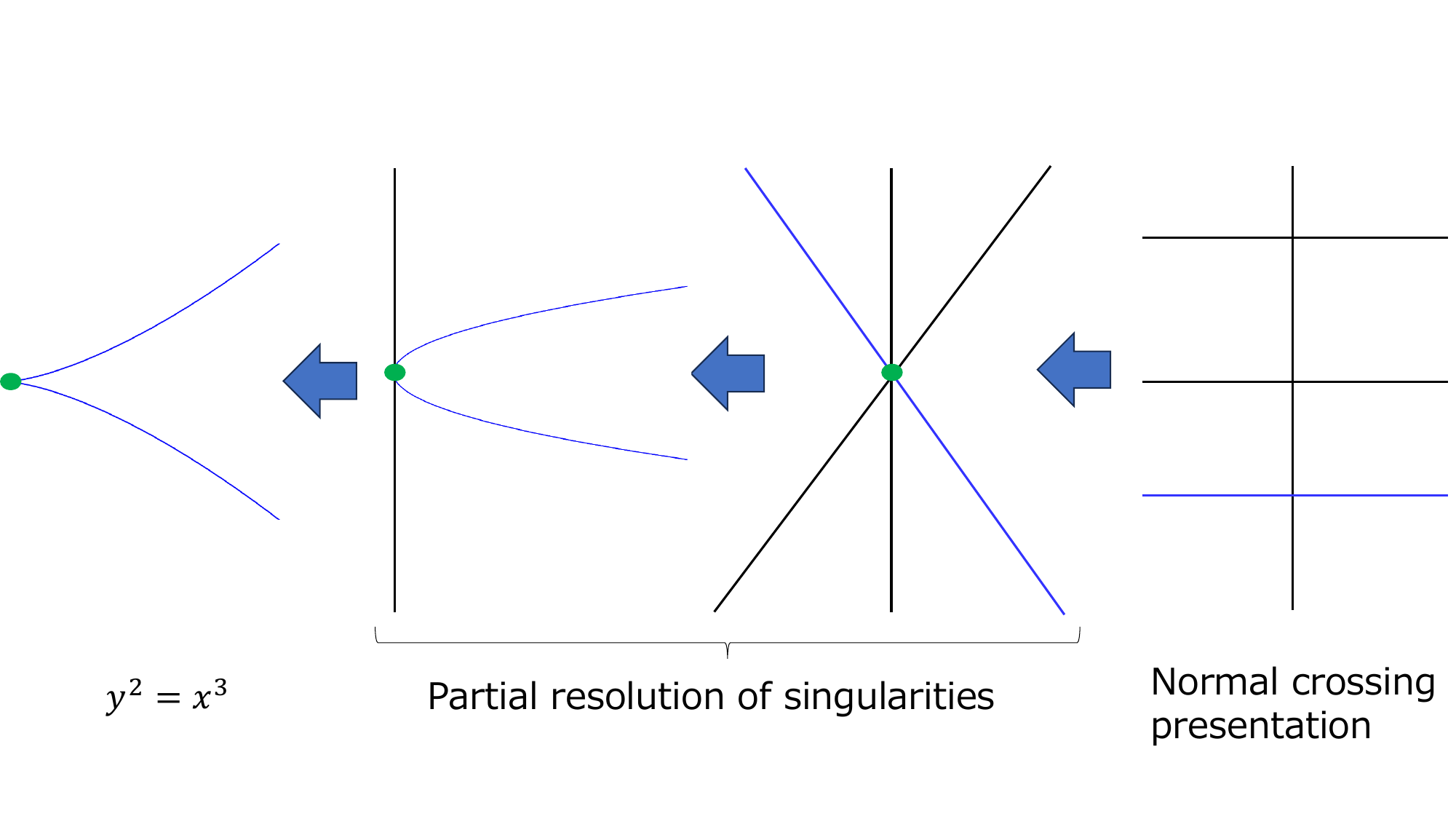}
        \caption{Resolution of the singularity of $y^2 = x^3$, called the monoidal transformation \cite{hartshorne2013algebraic}.}
        \label{fig:Monoidal transformation of the plane cuspidal curve}
    \end{figure}

    Secondally, let us choose the singular model $f(\theta_2) = (\theta_{21}^2-\theta_{22}^3)^2$ where $\theta_2=(\theta_{21},\theta_{22})\in \mathbb{R}^2$.
    This is the cuspidal plane curve and its resolution can be computed explicitly, related to the monoidal transformations; see Figure \ref{fig:Monoidal transformation of the plane cuspidal curve} and \cite{hartshorne2013algebraic}.
    We note that in our case, one has to treat $f(\theta_2)^4$, not $f(\theta_2)$, which forces that $k_2$ and $k_2^Q$ are multiplied by 4 as in the above case.
    In other words, taking a suitable affine cover, $2k_2=2k_2^Q$ are given by
    \[(24,0), (0,16), (16,48), (48,24).\]
    The quantity $h$, associated to the Jacobian, is given by
    \[(2,0), (0,1), (1,4), (4,2)\]
    as usual.
    Hence, 
    \[\lambda = \frac{5}{48},\quad m = 1\]
    where $m$ is the multiplicity.
    The multiplicity is needed to describe the asymptotic behavior of the free energy \cite[Theorem 11]{watanabe2018mathematical}, which we do not treat in this paper.
\end{ex}

\subsection{Numerical simulation}
\label{subsec:numerical_simulation}
In Section \ref{subsec:QWAIC}, we proposed QWAIC as an asymptotically unbiased estimator for $G_n^Q$.
Here, we numerically evaluate QWAIC and check the asymptotic unbiasedness for $G_n^Q$ (Thorem \ref{mainthm:QWAIC}) with an emphasis on the value of $C_n^Q$.
Below, we conduct measurement and estimation processes for different numbers of measurements (from $2,000$ to $8,000$) and repeat them $100$ times to see the statistical fluctuation of QWAIC for certain concrete parametric models of a 1-qubit system.
The random Pauli basis measurement, as described in Section \ref{subsec:analytic_calculation}, was chosen for a tomographically complete measurement in these numerical experiments.
The standard Metropolis-Hastings algorithm was used to conduct the Bayesian parameter estimation (we took $5,000$ samples, of which the first $500$ samples were set aside for the burn-in period).

First, we study a quantum regular model.
Suppose that 
\begin{align}
    \sigma(\theta) = \cos^2(\pi/32) |\phi(\theta)\rangle \langle\phi(\theta)| + \sin^2(\pi/32) \frac{I}{2}, \quad 
    \ket{\phi(\theta)} = \begin{pmatrix} \cos(\theta) \\ \sin(\theta) \end{pmatrix},
\end{align}
and let the true state $\rho = \sigma(\pi/4)$, which represents a superposition state with a slight addition of depolarizing noise.
The regularity property is confirmed by straightforward calculation.
Fig. \ref{fig:num_1Q-depol-1-a_AU} shows that the empirical average (over $100$ independent measurement and estimation processes) of the error $G_n^Q - \mathrm{QWAIC}$ (green line) converges to zero as the number of measurements increases, confirming the asymptotic unbiasedness of QWAIC.
Fig. \ref{fig:num_1Q-depol-1-a_C} shows the value of $C_n^Q$ (blue dot) compared with $8.08 \times 1/n$ (red dot), where $8.08$ is the ratio of the quantum and classical Fisher information computed analytically based on the value of $2\lambda^Q = \Tr(J^Q J^{-1})$ (Corollary \ref{cor:expansion of GnQ and TnQ for regular cases}) with $J^Q \approx 10.565$ and $J \approx 1.308$, and $1$ is the number of parameters.
We observe that the value of $C_n^Q$ is around the value of $8.08 \times 1/n$.
This may imply that, in regular cases, the second term of QWAIC reflects the model complexity (the number of parameters) and the measurement penalty (the ratio of quantum and classical Fisher information).
\begin{figure}[t]
    \centering
    \begin{subfigure}[b]{0.45\textwidth}
        \includegraphics[width=\textwidth]{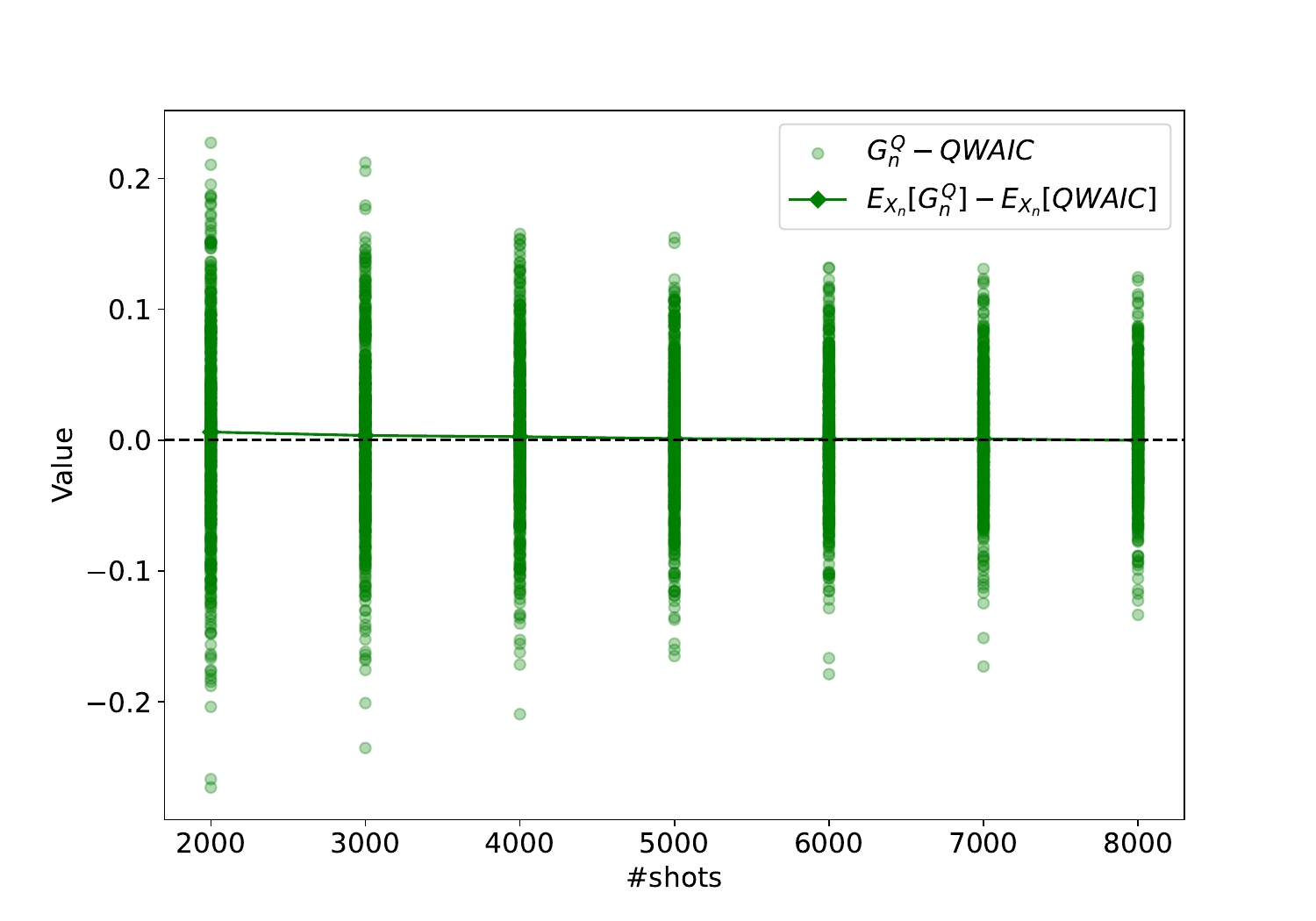}
        \caption{Asymptotic unbiasedness.}
        \label{fig:num_1Q-depol-1-a_AU}
    \end{subfigure}
    \begin{subfigure}[b]{0.45\textwidth}
        \includegraphics[width=\textwidth]{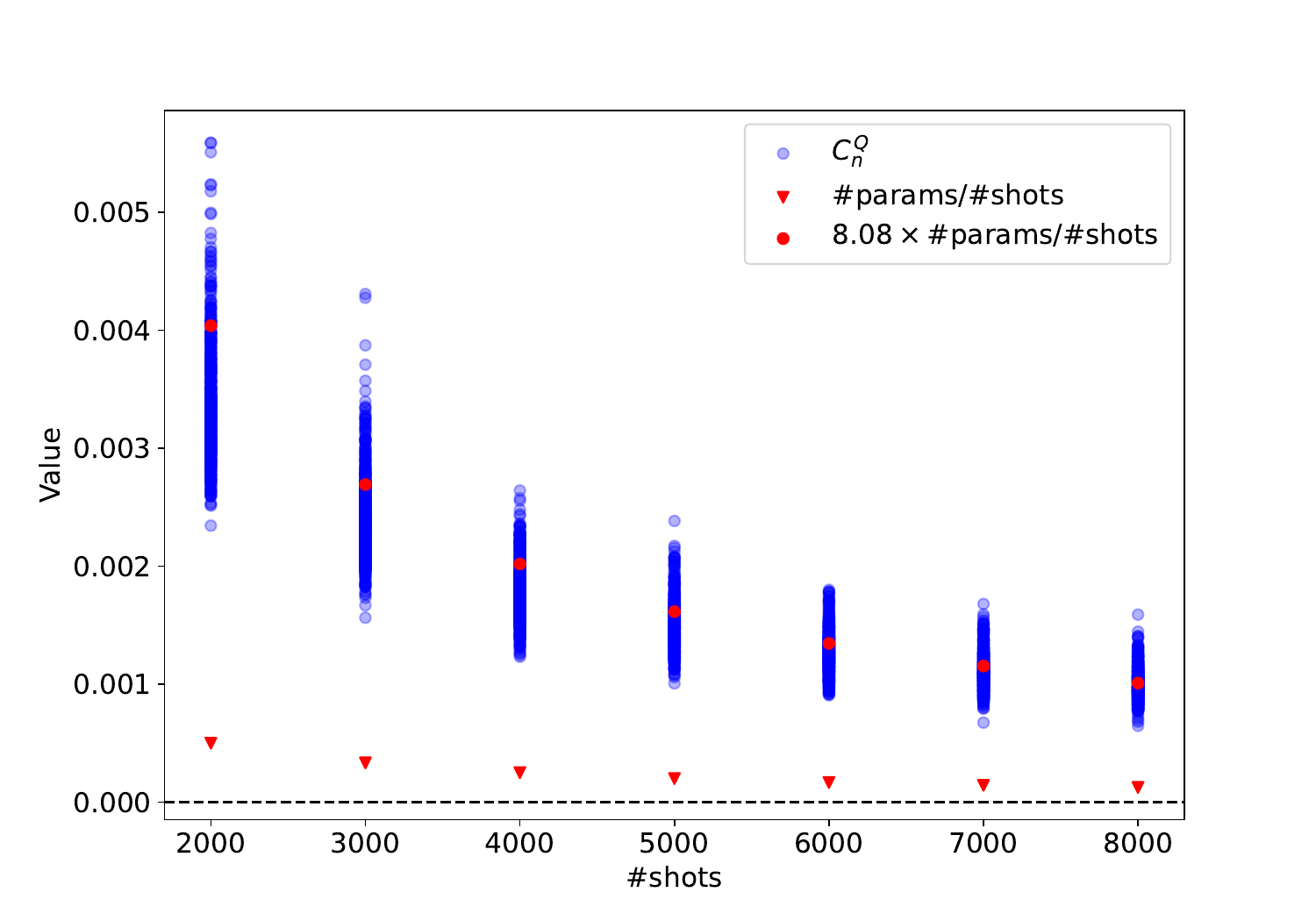}
        \caption{$C_n^Q$.}
        \label{fig:num_1Q-depol-1-a_C}
    \end{subfigure}
    \caption{Numerical results for a regular model.}
    \label{fig:num_1Q-depol-1-a}
\end{figure}

The second example focuses on the same singular (but classical) model as in Example \ref{ex:analy-1Q-depol-2-b}.
This example is suitable for comparing the analytical value of RLCT $\lambda$ calculated in Example \ref{ex:analy-1Q-depol-2-b} and the numerically evaluated value of $C_n^Q$.
The key difference from regular cases is the value of $C_n^Q$.
In Fig. \ref{fig:num_1Q-depol-2-b_C}, it is compared with the value of $3 \times 2/n$, where $2$ is the number of parameters and $3$ is the value of $r_{CQ}$.
Notably, it is observed that $C_n^Q$ reflects the value of $\lambda$ ($1/2$) calculated in Example \ref{ex:analy-1Q-depol-2-b} and takes almost half the value of $3 \times 2/n$.
\begin{figure}[t]
    \centering
    \includegraphics[width=0.5\textwidth]{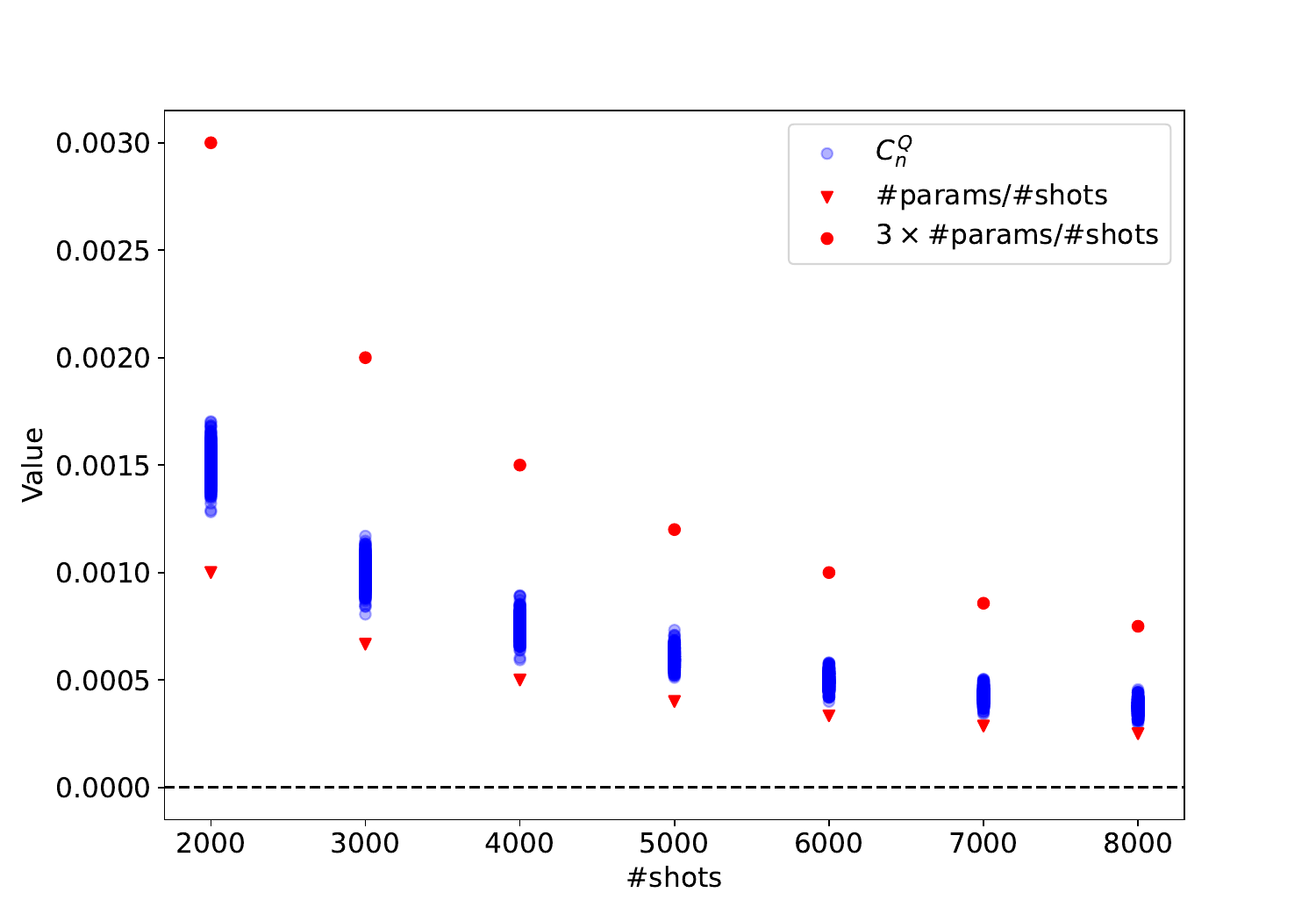}
    \caption{$C_n^Q$ for a singular model.}
    \label{fig:num_1Q-depol-2-b_C}
\end{figure}

We remark that the numerical results here align with our expectation that $C_n^Q$, the second term of QWAIC, represents the ``effective'' number of parameters in singular models (i.e. $\lambda$) as well as the penalty caused by the measurement with respect to parameter estimation (i.e. $r_{CQ}$), although there has been no clear theoretical guarantee yet.
Furthermore, since most of the concrete examples in this section are classical states for the simplicity of analytic calculation, the analytic and numerical validations for quantum states need to be done in order to demonstrate the practical effectiveness of QWAIC.

\section{Conclusion and open questions}
\label{sec:Conclusion and open questions}
This work initiated the study of statistical inference for quantum singular models. 
In our definition, quantum statistical models are singular if the Hessian matrix of the quantum relative entropy between the true and model states is degenerate, in analogy with classical statistics. 
Quantum singular models are commonly encountered in various contexts, such as neural network quantum states and quantum Boltzmann machines.
These models have been numerically shown to solve tasks that are challenging for conventional methods; however, a theoretical comprehension of these models remains insufficient.
To tackle this demanding issue, we formulate the problem of Bayesian quantum state estimation and model selection, representative examples of quantum statistical inference, and develop a theoretical framework that explains the learning behavior of quantum singular models, building upon singular learning theory by Watanabe  \cite{watanabe2009algebraic,watanabe2018mathematical}.
First, we described the asymptotic behavior of the generalization loss and training loss in both regular and singular cases. 
As a result, we observed the emergence of new quantities that do not appear in classical cases, particularly those related to the quantum Fisher information matrix. 
We demonstrated that these quantities are influenced by the measurement and the complexity of the singularities. 
Leveraging these findings, we also established the Quantum Widely Applicable Information Criterion (QWAIC), an information criterion that can be used for parametric model selection in quantum state estimation, even in singular situations. 
This quantity is a computable metric derived from given measurement outcomes, and we verified the validity of each term through numerical experiments.


Expanding the scope of our study reveals a variety of intriguing challenges.
First, the applicability of QWAIC can be better clarified by resolving some of the conjectures made in this paper (Conjectures \ref{conj:qregular to cregular} and \ref{conj:k=k^Q}).
This is related to the complexity of the singularities that arise in quantum information theory, highlighting the need for further research.
Secondly, singular learning theory, which underpins our work, employs profound mathematical concepts; zeta functions are a topic we have not dealt with in this paper. The relationship between the zeta function and algebraic geometry in the context of quantum state estimation will be explored in our next paper. 
The characterization of the learning coefficient that appears in the expansion formulas of $\mathbb{E}[G_n^Q]$ in terms of zeta functions and algebraic geometry is still an open question, unlike singular learning theory in the classical setting.
Thirdly, the computational aspect of QWAIC should also be improved for practical use.
Its computation involves the matrix logarithm of density operators by definition, which generally takes exponential time in the system size.
Some approaches \cite{verdon2019Quantuma,kieferova2021Quantum} have already been discussed, but improvements tailored for QWAIC are necessary. 
Lastly, despite the primary focus of this work on quantum state tomography, many other tasks in quantum statistical inference concerning singular models, such as quantum process tomography, have not yet been explored.
We anticipate that the insights gained from applying singular learning theory to quantum statistical problems will facilitate a better understanding of complex quantum information processing techniques.

\section*{Acknowledgements}
This work is supported by MEXT Quantum Leap Flagship Program (MEXT Q-LEAP) Grants No. JPMXS0118067285 and No. JPMXS0120319794.

\bibliography{main}

\providecommand{\noopsort}[1]{}\providecommand{\singleletter}[1]{#1}%
\begin{thebibliography}{121}%
\makeatletter
\providecommand \@ifxundefined [1]{%
 \@ifx{#1\undefined}
}%
\providecommand \@ifnum [1]{%
 \ifnum #1\expandafter \@firstoftwo
 \else \expandafter \@secondoftwo
 \fi
}%
\providecommand \@ifx [1]{%
 \ifx #1\expandafter \@firstoftwo
 \else \expandafter \@secondoftwo
 \fi
}%
\providecommand \natexlab [1]{#1}%
\providecommand \enquote  [1]{``#1''}%
\providecommand \bibnamefont  [1]{#1}%
\providecommand \bibfnamefont [1]{#1}%
\providecommand \citenamefont [1]{#1}%
\providecommand \href@noop [0]{\@secondoftwo}%
\providecommand \href [0]{\begingroup \@sanitize@url \@href}%
\providecommand \@href[1]{\@@startlink{#1}\@@href}%
\providecommand \@@href[1]{\endgroup#1\@@endlink}%
\providecommand \@sanitize@url [0]{\catcode `\\12\catcode `\$12\catcode `\&12\catcode `\#12\catcode `\^12\catcode `\_12\catcode `\%12\relax}%
\providecommand \@@startlink[1]{}%
\providecommand \@@endlink[0]{}%
\providecommand \url  [0]{\begingroup\@sanitize@url \@url }%
\providecommand \@url [1]{\endgroup\@href {#1}{\urlprefix }}%
\providecommand \urlprefix  [0]{URL }%
\providecommand \Eprint [0]{\href }%
\providecommand \doibase [0]{https://doi.org/}%
\providecommand \selectlanguage [0]{\@gobble}%
\providecommand \bibinfo  [0]{\@secondoftwo}%
\providecommand \bibfield  [0]{\@secondoftwo}%
\providecommand \translation [1]{[#1]}%
\providecommand \BibitemOpen [0]{}%
\providecommand \bibitemStop [0]{}%
\providecommand \bibitemNoStop [0]{.\EOS\space}%
\providecommand \EOS [0]{\spacefactor3000\relax}%
\providecommand \BibitemShut  [1]{\csname bibitem#1\endcsname}%
\let\auto@bib@innerbib\@empty
\bibitem [{\citenamefont {Cox}(2006)}]{cox2006Principles}%
  \BibitemOpen
  \bibfield  {author} {\bibinfo {author} {\bibfnamefont {D.~R.}\ \bibnamefont {Cox}},\ }\href {https://doi.org/10.1017/CBO9780511813559} {\emph {\bibinfo {title} {Principles of {{Statistical Inference}}}}}\ (\bibinfo  {publisher} {Cambridge University Press},\ \bibinfo {address} {Cambridge},\ \bibinfo {year} {2006})\BibitemShut {NoStop}%
\bibitem [{\citenamefont {Casella}\ and\ \citenamefont {Berger}(2024)}]{casella2024statistical}%
  \BibitemOpen
  \bibfield  {author} {\bibinfo {author} {\bibfnamefont {G.}~\bibnamefont {Casella}}\ and\ \bibinfo {author} {\bibfnamefont {R.}~\bibnamefont {Berger}},\ }\bibfield  {title} {\bibinfo {title} {Statistical inference},\ }\href@noop {} {\bibfield  {journal} {\bibinfo  {journal} {CRC Press}\ } (\bibinfo {year} {2024})}\BibitemShut {NoStop}%
\bibitem [{\citenamefont {Fisher}(1922)}]{fisher1922mathematical}%
  \BibitemOpen
  \bibfield  {author} {\bibinfo {author} {\bibfnamefont {R.~A.}\ \bibnamefont {Fisher}},\ }\bibfield  {title} {\bibinfo {title} {On the mathematical foundations of theoretical statistics},\ }\href@noop {} {\bibfield  {journal} {\bibinfo  {journal} {Philosophical transactions of the Royal Society of London. Series A, containing papers of a mathematical or physical character}\ }\textbf {\bibinfo {volume} {222}},\ \bibinfo {pages} {309} (\bibinfo {year} {1922})}\BibitemShut {NoStop}%
\bibitem [{\citenamefont {Cram{\'e}r}(1999)}]{cramer1999mathematical}%
  \BibitemOpen
  \bibfield  {author} {\bibinfo {author} {\bibfnamefont {H.}~\bibnamefont {Cram{\'e}r}},\ }\bibfield  {title} {\bibinfo {title} {Mathematical methods of statistics},\ }\href@noop {} {\bibfield  {journal} {\bibinfo  {journal} {Princeton university press}\ }\textbf {\bibinfo {volume} {26}} (\bibinfo {year} {1999})}\BibitemShut {NoStop}%
\bibitem [{\citenamefont {Van~der Vaart}(2000)}]{van2000asymptotic}%
  \BibitemOpen
  \bibfield  {author} {\bibinfo {author} {\bibfnamefont {A.~W.}\ \bibnamefont {Van~der Vaart}},\ }\bibfield  {title} {\bibinfo {title} {Asymptotic statistics},\ }\href@noop {} {\bibfield  {journal} {\bibinfo  {journal} {Cambridge university press}\ }\textbf {\bibinfo {volume} {3}} (\bibinfo {year} {2000})}\BibitemShut {NoStop}%
\bibitem [{\citenamefont {Le~Cam}\ and\ \citenamefont {Lo~Yang}(2000)}]{lecam2000Asymptotics}%
  \BibitemOpen
  \bibfield  {author} {\bibinfo {author} {\bibfnamefont {L.}~\bibnamefont {Le~Cam}}\ and\ \bibinfo {author} {\bibfnamefont {G.}~\bibnamefont {Lo~Yang}},\ }\href {https://doi.org/10.1007/978-1-4612-1166-2} {\emph {\bibinfo {title} {Asymptotics in {{Statistics}}}}},\ Springer {{Series}} in {{Statistics}}\ (\bibinfo  {publisher} {Springer},\ \bibinfo {address} {New York, NY},\ \bibinfo {year} {2000})\BibitemShut {NoStop}%
\bibitem [{\citenamefont {Anderson}\ and\ \citenamefont {Burnham}(2004)}]{anderson2004model}%
  \BibitemOpen
  \bibfield  {author} {\bibinfo {author} {\bibfnamefont {D.}~\bibnamefont {Anderson}}\ and\ \bibinfo {author} {\bibfnamefont {K.}~\bibnamefont {Burnham}},\ }\bibfield  {title} {\bibinfo {title} {Model selection and multi-model inference},\ }\href@noop {} {\bibfield  {journal} {\bibinfo  {journal} {Second. NY: Springer-Verlag}\ }\textbf {\bibinfo {volume} {63}},\ \bibinfo {pages} {10} (\bibinfo {year} {2004})}\BibitemShut {NoStop}%
\bibitem [{\citenamefont {Burnham}\ and\ \citenamefont {Anderson}(2004)}]{burnham2004multimodel}%
  \BibitemOpen
  \bibfield  {author} {\bibinfo {author} {\bibfnamefont {K.~P.}\ \bibnamefont {Burnham}}\ and\ \bibinfo {author} {\bibfnamefont {D.~R.}\ \bibnamefont {Anderson}},\ }\bibfield  {title} {\bibinfo {title} {Multimodel inference: understanding {AIC} and {BIC} in model selection},\ }\href@noop {} {\bibfield  {journal} {\bibinfo  {journal} {Sociological methods \& research}\ }\textbf {\bibinfo {volume} {33}},\ \bibinfo {pages} {261} (\bibinfo {year} {2004})}\BibitemShut {NoStop}%
\bibitem [{\citenamefont {Claeskens}\ and\ \citenamefont {Hjort}(2008)}]{claeskens2008model}%
  \BibitemOpen
  \bibfield  {author} {\bibinfo {author} {\bibfnamefont {G.}~\bibnamefont {Claeskens}}\ and\ \bibinfo {author} {\bibfnamefont {N.~L.}\ \bibnamefont {Hjort}},\ }\bibfield  {title} {\bibinfo {title} {Model selection and model averaging},\ }\href@noop {} {\bibfield  {journal} {\bibinfo  {journal} {Cambridge books}\ } (\bibinfo {year} {2008})}\BibitemShut {NoStop}%
\bibitem [{\citenamefont {Akaike}(1973)}]{akaike1973Information}%
  \BibitemOpen
  \bibfield  {author} {\bibinfo {author} {\bibfnamefont {H.}~\bibnamefont {Akaike}},\ }\bibfield  {title} {\bibinfo {title} {Information theory and an extension of the maximum likelihood principle},\ }\href@noop {} {\bibfield  {journal} {\bibinfo  {journal} {Second {{International Symposium}} on {{Information Theory}} (Eds. {{B}}. {{N}}. {{Petrov}} and {{F}}. {{Csaki}})}\ ,\ \bibinfo {pages} {267}} (\bibinfo {year} {1973})}\BibitemShut {NoStop}%
\bibitem [{\citenamefont {Akaike}(1974)}]{akaike1974New}%
  \BibitemOpen
  \bibfield  {author} {\bibinfo {author} {\bibfnamefont {H.}~\bibnamefont {Akaike}},\ }\bibfield  {title} {\bibinfo {title} {A new look at the statistical model identification},\ }\href {https://doi.org/10.1109/TAC.1974.1100705} {\bibfield  {journal} {\bibinfo  {journal} {IEEE Transactions on Automatic Control}\ }\textbf {\bibinfo {volume} {19}},\ \bibinfo {pages} {716} (\bibinfo {year} {1974})}\BibitemShut {NoStop}%
\bibitem [{\citenamefont {Amari}\ \emph {et~al.}(2003)\citenamefont {Amari}, \citenamefont {Ozeki},\ and\ \citenamefont {Park}}]{amari2003learning}%
  \BibitemOpen
  \bibfield  {author} {\bibinfo {author} {\bibfnamefont {S.-i.}\ \bibnamefont {Amari}}, \bibinfo {author} {\bibfnamefont {T.}~\bibnamefont {Ozeki}},\ and\ \bibinfo {author} {\bibfnamefont {H.}~\bibnamefont {Park}},\ }\bibfield  {title} {\bibinfo {title} {Learning and inference in hierarchical models with singularities},\ }\href@noop {} {\bibfield  {journal} {\bibinfo  {journal} {Systems and Computers in Japan}\ }\textbf {\bibinfo {volume} {34}},\ \bibinfo {pages} {34} (\bibinfo {year} {2003})}\BibitemShut {NoStop}%
\bibitem [{\citenamefont {Watanabe}(2007)}]{watanabe2007almost}%
  \BibitemOpen
  \bibfield  {author} {\bibinfo {author} {\bibfnamefont {S.}~\bibnamefont {Watanabe}},\ }\bibfield  {title} {\bibinfo {title} {Almost all learning machines are singular},\ }in\ \href@noop {} {\emph {\bibinfo {booktitle} {2007 IEEE Symposium on Foundations of Computational Intelligence}}}\ (\bibinfo {organization} {IEEE},\ \bibinfo {year} {2007})\ pp.\ \bibinfo {pages} {383--388}\BibitemShut {NoStop}%
\bibitem [{\citenamefont {Yamazaki}\ and\ \citenamefont {Watanabe}(2003)}]{yamazaki2003singularities}%
  \BibitemOpen
  \bibfield  {author} {\bibinfo {author} {\bibfnamefont {K.}~\bibnamefont {Yamazaki}}\ and\ \bibinfo {author} {\bibfnamefont {S.}~\bibnamefont {Watanabe}},\ }\bibfield  {title} {\bibinfo {title} {Singularities in mixture models and upper bounds of stochastic complexity},\ }\href@noop {} {\bibfield  {journal} {\bibinfo  {journal} {Neural networks}\ }\textbf {\bibinfo {volume} {16}},\ \bibinfo {pages} {1029} (\bibinfo {year} {2003})}\BibitemShut {NoStop}%
\bibitem [{\citenamefont {Sato}\ and\ \citenamefont {Watanabe}(2019)}]{sato2019bayesian}%
  \BibitemOpen
  \bibfield  {author} {\bibinfo {author} {\bibfnamefont {K.}~\bibnamefont {Sato}}\ and\ \bibinfo {author} {\bibfnamefont {S.}~\bibnamefont {Watanabe}},\ }\bibfield  {title} {\bibinfo {title} {Bayesian generalization error of {P}oisson mixture and simplex {V}andermonde matrix type singularity},\ }\href@noop {} {\bibfield  {journal} {\bibinfo  {journal} {arXiv preprint arXiv:1912.13289}\ } (\bibinfo {year} {2019})}\BibitemShut {NoStop}%
\bibitem [{\citenamefont {Watanabe}(2021)}]{watanabe2021waic}%
  \BibitemOpen
  \bibfield  {author} {\bibinfo {author} {\bibfnamefont {S.}~\bibnamefont {Watanabe}},\ }\bibfield  {title} {\bibinfo {title} {{WAIC} and {WBIC} for mixture models},\ }\href@noop {} {\bibfield  {journal} {\bibinfo  {journal} {Behaviormetrika}\ }\textbf {\bibinfo {volume} {48}},\ \bibinfo {pages} {5} (\bibinfo {year} {2021})}\BibitemShut {NoStop}%
\bibitem [{\citenamefont {Kariya}\ and\ \citenamefont {Watanabe}(2022)}]{kariya2022asymptotic}%
  \BibitemOpen
  \bibfield  {author} {\bibinfo {author} {\bibfnamefont {N.}~\bibnamefont {Kariya}}\ and\ \bibinfo {author} {\bibfnamefont {S.}~\bibnamefont {Watanabe}},\ }\bibfield  {title} {\bibinfo {title} {Asymptotic analysis of singular likelihood ratio of normal mixture by {B}ayesian learning theory for testing homogeneity},\ }\href@noop {} {\bibfield  {journal} {\bibinfo  {journal} {Communications in Statistics-Theory and Methods}\ }\textbf {\bibinfo {volume} {51}},\ \bibinfo {pages} {5873} (\bibinfo {year} {2022})}\BibitemShut {NoStop}%
\bibitem [{\citenamefont {Watanabe}\ and\ \citenamefont {Watanabe}(2022)}]{watanabe2022asymptotic}%
  \BibitemOpen
  \bibfield  {author} {\bibinfo {author} {\bibfnamefont {T.}~\bibnamefont {Watanabe}}\ and\ \bibinfo {author} {\bibfnamefont {S.}~\bibnamefont {Watanabe}},\ }\bibfield  {title} {\bibinfo {title} {Asymptotic behavior of {B}ayesian generalization error in multinomial mixtures},\ }\href@noop {} {\bibfield  {journal} {\bibinfo  {journal} {arXiv preprint arXiv:2203.06884}\ } (\bibinfo {year} {2022})}\BibitemShut {NoStop}%
\bibitem [{\citenamefont {Aoyagi}\ and\ \citenamefont {Nagata}(2012)}]{aoyagi2012learning}%
  \BibitemOpen
  \bibfield  {author} {\bibinfo {author} {\bibfnamefont {M.}~\bibnamefont {Aoyagi}}\ and\ \bibinfo {author} {\bibfnamefont {K.}~\bibnamefont {Nagata}},\ }\bibfield  {title} {\bibinfo {title} {Learning coefficient of generalization error in {B}ayesian estimation and {V}andermonde matrix-type singularity},\ }\href@noop {} {\bibfield  {journal} {\bibinfo  {journal} {Neural Computation}\ }\textbf {\bibinfo {volume} {24}},\ \bibinfo {pages} {1569} (\bibinfo {year} {2012})}\BibitemShut {NoStop}%
\bibitem [{\citenamefont {Wei}\ \emph {et~al.}(2022)\citenamefont {Wei}, \citenamefont {Murfet}, \citenamefont {Gong}, \citenamefont {Li}, \citenamefont {Gell-Redman},\ and\ \citenamefont {Quella}}]{wei2022deep}%
  \BibitemOpen
  \bibfield  {author} {\bibinfo {author} {\bibfnamefont {S.}~\bibnamefont {Wei}}, \bibinfo {author} {\bibfnamefont {D.}~\bibnamefont {Murfet}}, \bibinfo {author} {\bibfnamefont {M.}~\bibnamefont {Gong}}, \bibinfo {author} {\bibfnamefont {H.}~\bibnamefont {Li}}, \bibinfo {author} {\bibfnamefont {J.}~\bibnamefont {Gell-Redman}},\ and\ \bibinfo {author} {\bibfnamefont {T.}~\bibnamefont {Quella}},\ }\bibfield  {title} {\bibinfo {title} {Deep learning is singular, and that’s good},\ }\href@noop {} {\bibfield  {journal} {\bibinfo  {journal} {IEEE Transactions on Neural Networks and Learning Systems}\ } (\bibinfo {year} {2022})}\BibitemShut {NoStop}%
\bibitem [{\citenamefont {Hoogland}\ \emph {et~al.}(2024)\citenamefont {Hoogland}, \citenamefont {Carroll},\ and\ \citenamefont {Murfet}}]{hoogland2024stagewise}%
  \BibitemOpen
  \bibfield  {author} {\bibinfo {author} {\bibfnamefont {J.}~\bibnamefont {Hoogland}}, \bibinfo {author} {\bibfnamefont {L.}~\bibnamefont {Carroll}},\ and\ \bibinfo {author} {\bibfnamefont {D.}~\bibnamefont {Murfet}},\ }\bibfield  {title} {\bibinfo {title} {Stagewise development in neural networks},\ }\href@noop {} {\bibfield  {journal} {\bibinfo  {journal} {AI Alignment Forum}\ } (\bibinfo {year} {2024})}\BibitemShut {NoStop}%
\bibitem [{\citenamefont {Wang}\ \emph {et~al.}(2024)\citenamefont {Wang}, \citenamefont {Farrugia-Roberts}, \citenamefont {Hoogland}, \citenamefont {Carroll}, \citenamefont {Wei},\ and\ \citenamefont {Murfet}}]{wang2024loss}%
  \BibitemOpen
  \bibfield  {author} {\bibinfo {author} {\bibfnamefont {G.}~\bibnamefont {Wang}}, \bibinfo {author} {\bibfnamefont {M.}~\bibnamefont {Farrugia-Roberts}}, \bibinfo {author} {\bibfnamefont {J.}~\bibnamefont {Hoogland}}, \bibinfo {author} {\bibfnamefont {L.}~\bibnamefont {Carroll}}, \bibinfo {author} {\bibfnamefont {S.}~\bibnamefont {Wei}},\ and\ \bibinfo {author} {\bibfnamefont {D.}~\bibnamefont {Murfet}},\ }\bibfield  {title} {\bibinfo {title} {Loss landscape geometry reveals stagewise development of transformers},\ }in\ \href@noop {} {\emph {\bibinfo {booktitle} {High-dimensional Learning Dynamics 2024: The Emergence of Structure and Reasoning}}}\ (\bibinfo {year} {2024})\BibitemShut {NoStop}%
\bibitem [{\citenamefont {LeCun}\ \emph {et~al.}(2015)\citenamefont {LeCun}, \citenamefont {Bengio},\ and\ \citenamefont {Hinton}}]{lecun2015deep}%
  \BibitemOpen
  \bibfield  {author} {\bibinfo {author} {\bibfnamefont {Y.}~\bibnamefont {LeCun}}, \bibinfo {author} {\bibfnamefont {Y.}~\bibnamefont {Bengio}},\ and\ \bibinfo {author} {\bibfnamefont {G.}~\bibnamefont {Hinton}},\ }\bibfield  {title} {\bibinfo {title} {Deep learning},\ }\href@noop {} {\bibfield  {journal} {\bibinfo  {journal} {nature}\ }\textbf {\bibinfo {volume} {521}},\ \bibinfo {pages} {436} (\bibinfo {year} {2015})}\BibitemShut {NoStop}%
\bibitem [{\citenamefont {Lau}\ \emph {et~al.}(2024)\citenamefont {Lau}, \citenamefont {Furman}, \citenamefont {Wang}, \citenamefont {Murfet},\ and\ \citenamefont {Wei}}]{lau2024locallearningcoefficientsingularityaware}%
  \BibitemOpen
  \bibfield  {author} {\bibinfo {author} {\bibfnamefont {E.}~\bibnamefont {Lau}}, \bibinfo {author} {\bibfnamefont {Z.}~\bibnamefont {Furman}}, \bibinfo {author} {\bibfnamefont {G.}~\bibnamefont {Wang}}, \bibinfo {author} {\bibfnamefont {D.}~\bibnamefont {Murfet}},\ and\ \bibinfo {author} {\bibfnamefont {S.}~\bibnamefont {Wei}},\ }\href {https://arxiv.org/abs/2308.12108} {\bibinfo {title} {The local learning coefficient: A singularity-aware complexity measure}} (\bibinfo {year} {2024}),\ \Eprint {https://arxiv.org/abs/2308.12108} {arXiv:2308.12108 [stat.ML]} \BibitemShut {NoStop}%
\bibitem [{\citenamefont {Zhang}\ \emph {et~al.}(2021)\citenamefont {Zhang}, \citenamefont {Bengio}, \citenamefont {Hardt}, \citenamefont {Recht},\ and\ \citenamefont {Vinyals}}]{zhang2021understanding}%
  \BibitemOpen
  \bibfield  {author} {\bibinfo {author} {\bibfnamefont {C.}~\bibnamefont {Zhang}}, \bibinfo {author} {\bibfnamefont {S.}~\bibnamefont {Bengio}}, \bibinfo {author} {\bibfnamefont {M.}~\bibnamefont {Hardt}}, \bibinfo {author} {\bibfnamefont {B.}~\bibnamefont {Recht}},\ and\ \bibinfo {author} {\bibfnamefont {O.}~\bibnamefont {Vinyals}},\ }\bibfield  {title} {\bibinfo {title} {Understanding deep learning (still) requires rethinking generalization},\ }\href@noop {} {\bibfield  {journal} {\bibinfo  {journal} {Communications of the ACM}\ }\textbf {\bibinfo {volume} {64}},\ \bibinfo {pages} {107} (\bibinfo {year} {2021})}\BibitemShut {NoStop}%
\bibitem [{\citenamefont {Bereska}\ and\ \citenamefont {Gavves}(2024)}]{bereska2024mechanistic}%
  \BibitemOpen
  \bibfield  {author} {\bibinfo {author} {\bibfnamefont {L.}~\bibnamefont {Bereska}}\ and\ \bibinfo {author} {\bibfnamefont {S.}~\bibnamefont {Gavves}},\ }\bibfield  {title} {\bibinfo {title} {Mechanistic interpretability for {AI} safety - a review},\ }\href {https://openreview.net/forum?id=ePUVetPKu6} {\bibfield  {journal} {\bibinfo  {journal} {Transactions on Machine Learning Research}\ } (\bibinfo {year} {2024})},\ \bibinfo {note} {survey Certification, Expert Certification}\BibitemShut {NoStop}%
\bibitem [{\citenamefont {Anwar}\ \emph {et~al.}(2024)\citenamefont {Anwar} \emph {et~al.}}]{anwar2024foundationalchallengesassuringalignment}%
  \BibitemOpen
  \bibfield  {author} {\bibinfo {author} {\bibfnamefont {U.}~\bibnamefont {Anwar}} \emph {et~al.},\ }\href {https://arxiv.org/abs/2404.09932} {\bibinfo {title} {Foundational challenges in assuring alignment and safety of large language models}} (\bibinfo {year} {2024}),\ \Eprint {https://arxiv.org/abs/2404.09932} {arXiv:2404.09932 [cs.LG]} \BibitemShut {NoStop}%
\bibitem [{\citenamefont {Jacot}\ \emph {et~al.}(2018)\citenamefont {Jacot}, \citenamefont {Gabriel},\ and\ \citenamefont {Hongler}}]{jacot2018neural}%
  \BibitemOpen
  \bibfield  {author} {\bibinfo {author} {\bibfnamefont {A.}~\bibnamefont {Jacot}}, \bibinfo {author} {\bibfnamefont {F.}~\bibnamefont {Gabriel}},\ and\ \bibinfo {author} {\bibfnamefont {C.}~\bibnamefont {Hongler}},\ }\bibfield  {title} {\bibinfo {title} {Neural tangent kernel: Convergence and generalization in neural networks},\ }\href@noop {} {\bibfield  {journal} {\bibinfo  {journal} {Advances in neural information processing systems}\ }\textbf {\bibinfo {volume} {31}} (\bibinfo {year} {2018})}\BibitemShut {NoStop}%
\bibitem [{\citenamefont {Lee}\ \emph {et~al.}(2019)\citenamefont {Lee}, \citenamefont {Xiao}, \citenamefont {Schoenholz}, \citenamefont {Bahri}, \citenamefont {Novak}, \citenamefont {Sohl-Dickstein},\ and\ \citenamefont {Pennington}}]{lee2019wide}%
  \BibitemOpen
  \bibfield  {author} {\bibinfo {author} {\bibfnamefont {J.}~\bibnamefont {Lee}}, \bibinfo {author} {\bibfnamefont {L.}~\bibnamefont {Xiao}}, \bibinfo {author} {\bibfnamefont {S.}~\bibnamefont {Schoenholz}}, \bibinfo {author} {\bibfnamefont {Y.}~\bibnamefont {Bahri}}, \bibinfo {author} {\bibfnamefont {R.}~\bibnamefont {Novak}}, \bibinfo {author} {\bibfnamefont {J.}~\bibnamefont {Sohl-Dickstein}},\ and\ \bibinfo {author} {\bibfnamefont {J.}~\bibnamefont {Pennington}},\ }\bibfield  {title} {\bibinfo {title} {Wide neural networks of any depth evolve as linear models under gradient descent},\ }\href@noop {} {\bibfield  {journal} {\bibinfo  {journal} {Advances in neural information processing systems}\ }\textbf {\bibinfo {volume} {32}} (\bibinfo {year} {2019})}\BibitemShut {NoStop}%
\bibitem [{\citenamefont {Lee}\ \emph {et~al.}(2020)\citenamefont {Lee}, \citenamefont {Schoenholz}, \citenamefont {Pennington}, \citenamefont {Adlam}, \citenamefont {Xiao}, \citenamefont {Novak},\ and\ \citenamefont {Sohl-Dickstein}}]{lee2020finite}%
  \BibitemOpen
  \bibfield  {author} {\bibinfo {author} {\bibfnamefont {J.}~\bibnamefont {Lee}}, \bibinfo {author} {\bibfnamefont {S.}~\bibnamefont {Schoenholz}}, \bibinfo {author} {\bibfnamefont {J.}~\bibnamefont {Pennington}}, \bibinfo {author} {\bibfnamefont {B.}~\bibnamefont {Adlam}}, \bibinfo {author} {\bibfnamefont {L.}~\bibnamefont {Xiao}}, \bibinfo {author} {\bibfnamefont {R.}~\bibnamefont {Novak}},\ and\ \bibinfo {author} {\bibfnamefont {J.}~\bibnamefont {Sohl-Dickstein}},\ }\bibfield  {title} {\bibinfo {title} {Finite versus infinite neural networks: an empirical study},\ }\href@noop {} {\bibfield  {journal} {\bibinfo  {journal} {Advances in Neural Information Processing Systems}\ }\textbf {\bibinfo {volume} {33}},\ \bibinfo {pages} {15156} (\bibinfo {year} {2020})}\BibitemShut {NoStop}%
\bibitem [{\citenamefont {Yang}\ and\ \citenamefont {Schoenholz}(2017)}]{yang2017mean}%
  \BibitemOpen
  \bibfield  {author} {\bibinfo {author} {\bibfnamefont {G.}~\bibnamefont {Yang}}\ and\ \bibinfo {author} {\bibfnamefont {S.}~\bibnamefont {Schoenholz}},\ }\bibfield  {title} {\bibinfo {title} {Mean field residual networks: On the edge of chaos},\ }\href@noop {} {\bibfield  {journal} {\bibinfo  {journal} {Advances in neural information processing systems}\ }\textbf {\bibinfo {volume} {30}} (\bibinfo {year} {2017})}\BibitemShut {NoStop}%
\bibitem [{\citenamefont {Lee}\ \emph {et~al.}(2017)\citenamefont {Lee}, \citenamefont {Bahri}, \citenamefont {Novak}, \citenamefont {Schoenholz}, \citenamefont {Pennington},\ and\ \citenamefont {Sohl-Dickstein}}]{lee2017deep}%
  \BibitemOpen
  \bibfield  {author} {\bibinfo {author} {\bibfnamefont {J.}~\bibnamefont {Lee}}, \bibinfo {author} {\bibfnamefont {Y.}~\bibnamefont {Bahri}}, \bibinfo {author} {\bibfnamefont {R.}~\bibnamefont {Novak}}, \bibinfo {author} {\bibfnamefont {S.~S.}\ \bibnamefont {Schoenholz}}, \bibinfo {author} {\bibfnamefont {J.}~\bibnamefont {Pennington}},\ and\ \bibinfo {author} {\bibfnamefont {J.}~\bibnamefont {Sohl-Dickstein}},\ }\bibfield  {title} {\bibinfo {title} {Deep neural networks as gaussian processes},\ }\href@noop {} {\bibfield  {journal} {\bibinfo  {journal} {arXiv preprint arXiv:1711.00165}\ } (\bibinfo {year} {2017})}\BibitemShut {NoStop}%
\bibitem [{\citenamefont {Mei}\ \emph {et~al.}(2018)\citenamefont {Mei}, \citenamefont {Montanari},\ and\ \citenamefont {Nguyen}}]{mei2018mean}%
  \BibitemOpen
  \bibfield  {author} {\bibinfo {author} {\bibfnamefont {S.}~\bibnamefont {Mei}}, \bibinfo {author} {\bibfnamefont {A.}~\bibnamefont {Montanari}},\ and\ \bibinfo {author} {\bibfnamefont {P.-M.}\ \bibnamefont {Nguyen}},\ }\bibfield  {title} {\bibinfo {title} {A mean field view of the landscape of two-layer neural networks},\ }\href@noop {} {\bibfield  {journal} {\bibinfo  {journal} {Proceedings of the National Academy of Sciences}\ }\textbf {\bibinfo {volume} {115}},\ \bibinfo {pages} {E7665} (\bibinfo {year} {2018})}\BibitemShut {NoStop}%
\bibitem [{\citenamefont {Mei}\ and\ \citenamefont {Montanari}(2022)}]{mei2022Generalization}%
  \BibitemOpen
  \bibfield  {author} {\bibinfo {author} {\bibfnamefont {S.}~\bibnamefont {Mei}}\ and\ \bibinfo {author} {\bibfnamefont {A.}~\bibnamefont {Montanari}},\ }\bibfield  {title} {\bibinfo {title} {The {{Generalization Error}} of {{Random Features Regression}}: {{Precise Asymptotics}} and the {{Double Descent Curve}}},\ }\href {https://doi.org/10.1002/cpa.22008} {\bibfield  {journal} {\bibinfo  {journal} {Communications on Pure and Applied Mathematics}\ }\textbf {\bibinfo {volume} {75}},\ \bibinfo {pages} {667} (\bibinfo {year} {2022})}\BibitemShut {NoStop}%
\bibitem [{\citenamefont {Hastie}\ \emph {et~al.}(2022)\citenamefont {Hastie}, \citenamefont {Montanari}, \citenamefont {Rosset},\ and\ \citenamefont {Tibshirani}}]{hastie2022Surprises}%
  \BibitemOpen
  \bibfield  {author} {\bibinfo {author} {\bibfnamefont {T.}~\bibnamefont {Hastie}}, \bibinfo {author} {\bibfnamefont {A.}~\bibnamefont {Montanari}}, \bibinfo {author} {\bibfnamefont {S.}~\bibnamefont {Rosset}},\ and\ \bibinfo {author} {\bibfnamefont {R.~J.}\ \bibnamefont {Tibshirani}},\ }\bibfield  {title} {\bibinfo {title} {Surprises in high-dimensional ridgeless least squares interpolation},\ }\href {https://doi.org/10.1214/21-AOS2133} {\bibfield  {journal} {\bibinfo  {journal} {The Annals of Statistics}\ }\textbf {\bibinfo {volume} {50}},\ \bibinfo {pages} {949} (\bibinfo {year} {2022})}\BibitemShut {NoStop}%
\bibitem [{\citenamefont {Drton}\ \emph {et~al.}(2008)\citenamefont {Drton}, \citenamefont {Sturmfels},\ and\ \citenamefont {Sullivant}}]{drton2008lectures}%
  \BibitemOpen
  \bibfield  {author} {\bibinfo {author} {\bibfnamefont {M.}~\bibnamefont {Drton}}, \bibinfo {author} {\bibfnamefont {B.}~\bibnamefont {Sturmfels}},\ and\ \bibinfo {author} {\bibfnamefont {S.}~\bibnamefont {Sullivant}},\ }\bibfield  {title} {\bibinfo {title} {Lectures on algebraic statistics},\ }\href@noop {} {\bibfield  {journal} {\bibinfo  {journal} {Springer Science \& Business Media}\ }\textbf {\bibinfo {volume} {39}} (\bibinfo {year} {2008})}\BibitemShut {NoStop}%
\bibitem [{\citenamefont {Watanabe}(2009)}]{watanabe2009algebraic}%
  \BibitemOpen
  \bibfield  {author} {\bibinfo {author} {\bibfnamefont {S.}~\bibnamefont {Watanabe}},\ }\bibfield  {title} {\bibinfo {title} {Algebraic geometry and statistical learning theory},\ }\href@noop {} {\bibfield  {journal} {\bibinfo  {journal} {Cambridge university press}\ }\textbf {\bibinfo {volume} {25}} (\bibinfo {year} {2009})}\BibitemShut {NoStop}%
\bibitem [{\citenamefont {Watanabe}(2018)}]{watanabe2018mathematical}%
  \BibitemOpen
  \bibfield  {author} {\bibinfo {author} {\bibfnamefont {S.}~\bibnamefont {Watanabe}},\ }\bibfield  {title} {\bibinfo {title} {Mathematical theory of {B}ayesian statistics},\ }\href@noop {} {\bibfield  {journal} {\bibinfo  {journal} {Chapman and Hall/CRC}\ } (\bibinfo {year} {2018})}\BibitemShut {NoStop}%
\bibitem [{\citenamefont {Nagayasu}\ and\ \citenamefont {Watanabe}(2023{\natexlab{a}})}]{nagayasu2023bayesian}%
  \BibitemOpen
  \bibfield  {author} {\bibinfo {author} {\bibfnamefont {S.}~\bibnamefont {Nagayasu}}\ and\ \bibinfo {author} {\bibfnamefont {S.}~\bibnamefont {Watanabe}},\ }\bibfield  {title} {\bibinfo {title} {Bayesian free energy of deep relu neural network in overparametrized cases},\ }\href@noop {} {\bibfield  {journal} {\bibinfo  {journal} {arXiv preprint arXiv:2303.15739}\ } (\bibinfo {year} {2023}{\natexlab{a}})}\BibitemShut {NoStop}%
\bibitem [{\citenamefont {Nagayasu}\ and\ \citenamefont {Watanabe}(2023{\natexlab{b}})}]{nagayasu2023freeenergybayesianconvolutional}%
  \BibitemOpen
  \bibfield  {author} {\bibinfo {author} {\bibfnamefont {S.}~\bibnamefont {Nagayasu}}\ and\ \bibinfo {author} {\bibfnamefont {S.}~\bibnamefont {Watanabe}},\ }\href {https://arxiv.org/abs/2307.01417} {\bibinfo {title} {Free energy of bayesian convolutional neural network with skip connection}} (\bibinfo {year} {2023}{\natexlab{b}}),\ \Eprint {https://arxiv.org/abs/2307.01417} {arXiv:2307.01417 [cs.LG]} \BibitemShut {NoStop}%
\bibitem [{\citenamefont {Wei}\ and\ \citenamefont {Lau}(2024)}]{wei2024variational}%
  \BibitemOpen
  \bibfield  {author} {\bibinfo {author} {\bibfnamefont {S.}~\bibnamefont {Wei}}\ and\ \bibinfo {author} {\bibfnamefont {E.}~\bibnamefont {Lau}},\ }\bibfield  {title} {\bibinfo {title} {Variational bayesian neural networks via resolution of singularities},\ }\href@noop {} {\bibfield  {journal} {\bibinfo  {journal} {Journal of Computational and Graphical Statistics}\ ,\ \bibinfo {pages} {1}} (\bibinfo {year} {2024})}\BibitemShut {NoStop}%
\bibitem [{\citenamefont {Hayashi}(2005)}]{hayashi2005Asymptotic}%
  \BibitemOpen
  \bibinfo {editor} {\bibfnamefont {M.}~\bibnamefont {Hayashi}},\ ed.,\ \href@noop {} {\emph {\bibinfo {title} {Asymptotic Theory of Quantum Statistical Inference: Selected Papers}}}\ (\bibinfo  {publisher} {World Scientific},\ \bibinfo {address} {New Jersey},\ \bibinfo {year} {2005})\BibitemShut {NoStop}%
\bibitem [{\citenamefont {Jen{\v c}ov{\'a}}\ and\ \citenamefont {Petz}(2006)}]{jencova2006Sufficiency}%
  \BibitemOpen
  \bibfield  {author} {\bibinfo {author} {\bibfnamefont {A.}~\bibnamefont {Jen{\v c}ov{\'a}}}\ and\ \bibinfo {author} {\bibfnamefont {D.}~\bibnamefont {Petz}},\ }\bibfield  {title} {\bibinfo {title} {Sufficiency in {{Quantum Statistical Inference}}},\ }\href {https://doi.org/10.1007/s00220-005-1510-7} {\bibfield  {journal} {\bibinfo  {journal} {Communications in Mathematical Physics}\ }\textbf {\bibinfo {volume} {263}},\ \bibinfo {pages} {259} (\bibinfo {year} {2006})}\BibitemShut {NoStop}%
\bibitem [{\citenamefont {Gill}\ and\ \citenamefont {Gu{\c t}{\u a}}(2013)}]{gill2013Asymptotic}%
  \BibitemOpen
  \bibfield  {author} {\bibinfo {author} {\bibfnamefont {R.~D.}\ \bibnamefont {Gill}}\ and\ \bibinfo {author} {\bibfnamefont {M.~I.}\ \bibnamefont {Gu{\c t}{\u a}}},\ }\bibfield  {title} {\bibinfo {title} {On asymptotic quantum statistical inference},\ }in\ \href {https://doi.org/10.1214/12-IMSCOLL909} {\emph {\bibinfo {booktitle} {From {{Probability}} to {{Statistics}} and {{Back}}: {{High-Dimensional Models}} and {{Processes}} -- {{A Festschrift}} in {{Honor}} of {{Jon A}}. {{Wellner}}}}},\ Vol.~\bibinfo {volume} {9}\ (\bibinfo  {publisher} {Institute of Mathematical Statistics},\ \bibinfo {year} {2013})\ pp.\ \bibinfo {pages} {105--128}\BibitemShut {NoStop}%
\bibitem [{\citenamefont {Helstrom}(1969)}]{helstrom1969Quantum}%
  \BibitemOpen
  \bibfield  {author} {\bibinfo {author} {\bibfnamefont {C.~W.}\ \bibnamefont {Helstrom}},\ }\bibfield  {title} {\bibinfo {title} {Quantum detection and estimation theory},\ }\href {https://doi.org/10.1007/BF01007479} {\bibfield  {journal} {\bibinfo  {journal} {Journal of Statistical Physics}\ }\textbf {\bibinfo {volume} {1}},\ \bibinfo {pages} {231} (\bibinfo {year} {1969})}\BibitemShut {NoStop}%
\bibitem [{\citenamefont {Kahn}\ and\ \citenamefont {Gu{\c t}{\u a}}(2009)}]{kahn2009Local}%
  \BibitemOpen
  \bibfield  {author} {\bibinfo {author} {\bibfnamefont {J.}~\bibnamefont {Kahn}}\ and\ \bibinfo {author} {\bibfnamefont {M.}~\bibnamefont {Gu{\c t}{\u a}}},\ }\bibfield  {title} {\bibinfo {title} {Local {{Asymptotic Normality}} for {{Finite Dimensional Quantum Systems}}},\ }\href {https://doi.org/10.1007/s00220-009-0787-3} {\bibfield  {journal} {\bibinfo  {journal} {Communications in Mathematical Physics}\ }\textbf {\bibinfo {volume} {289}},\ \bibinfo {pages} {597} (\bibinfo {year} {2009})}\BibitemShut {NoStop}%
\bibitem [{\citenamefont {Usami}\ \emph {et~al.}(2003)\citenamefont {Usami}, \citenamefont {Nambu}, \citenamefont {Tsuda}, \citenamefont {Matsumoto},\ and\ \citenamefont {Nakamura}}]{usami2003Accuracy}%
  \BibitemOpen
  \bibfield  {author} {\bibinfo {author} {\bibfnamefont {K.}~\bibnamefont {Usami}}, \bibinfo {author} {\bibfnamefont {Y.}~\bibnamefont {Nambu}}, \bibinfo {author} {\bibfnamefont {Y.}~\bibnamefont {Tsuda}}, \bibinfo {author} {\bibfnamefont {K.}~\bibnamefont {Matsumoto}},\ and\ \bibinfo {author} {\bibfnamefont {K.}~\bibnamefont {Nakamura}},\ }\bibfield  {title} {\bibinfo {title} {Accuracy of quantum-state estimation utilizing {{Akaike}}'s information criterion},\ }\href {https://doi.org/10.1103/PhysRevA.68.022314} {\bibfield  {journal} {\bibinfo  {journal} {Physical Review A}\ }\textbf {\bibinfo {volume} {68}},\ \bibinfo {pages} {022314} (\bibinfo {year} {2003})}\BibitemShut {NoStop}%
\bibitem [{\citenamefont {Yin}\ and\ \citenamefont {{van Enk}}(2011)}]{yin2011Information}%
  \BibitemOpen
  \bibfield  {author} {\bibinfo {author} {\bibfnamefont {J.~O.~S.}\ \bibnamefont {Yin}}\ and\ \bibinfo {author} {\bibfnamefont {S.~J.}\ \bibnamefont {{van Enk}}},\ }\bibfield  {title} {\bibinfo {title} {Information criteria for efficient quantum state estimation},\ }\href {https://doi.org/10.1103/PhysRevA.83.062110} {\bibfield  {journal} {\bibinfo  {journal} {Physical Review A}\ }\textbf {\bibinfo {volume} {83}},\ \bibinfo {pages} {062110} (\bibinfo {year} {2011})}\BibitemShut {NoStop}%
\bibitem [{\citenamefont {Gu{\c t}{\u a}}\ \emph {et~al.}(2012)\citenamefont {Gu{\c t}{\u a}}, \citenamefont {Kypraios},\ and\ \citenamefont {Dryden}}]{guta2012Rankbased}%
  \BibitemOpen
  \bibfield  {author} {\bibinfo {author} {\bibfnamefont {M.}~\bibnamefont {Gu{\c t}{\u a}}}, \bibinfo {author} {\bibfnamefont {T.}~\bibnamefont {Kypraios}},\ and\ \bibinfo {author} {\bibfnamefont {I.}~\bibnamefont {Dryden}},\ }\bibfield  {title} {\bibinfo {title} {Rank-based model selection for multiple ions quantum tomography},\ }\href {https://doi.org/10.1088/1367-2630/14/10/105002} {\bibfield  {journal} {\bibinfo  {journal} {New Journal of Physics}\ }\textbf {\bibinfo {volume} {14}},\ \bibinfo {pages} {105002} (\bibinfo {year} {2012})}\BibitemShut {NoStop}%
\bibitem [{\citenamefont {van Enk}\ and\ \citenamefont {{Blume-Kohout}}(2013)}]{enk2013When}%
  \BibitemOpen
  \bibfield  {author} {\bibinfo {author} {\bibfnamefont {S.~J.}\ \bibnamefont {van Enk}}\ and\ \bibinfo {author} {\bibfnamefont {R.}~\bibnamefont {{Blume-Kohout}}},\ }\bibfield  {title} {\bibinfo {title} {When quantum tomography goes wrong: Drift of quantum sources and other errors},\ }\href {https://doi.org/10.1088/1367-2630/15/2/025024} {\bibfield  {journal} {\bibinfo  {journal} {New Journal of Physics}\ }\textbf {\bibinfo {volume} {15}},\ \bibinfo {pages} {025024} (\bibinfo {year} {2013})}\BibitemShut {NoStop}%
\bibitem [{\citenamefont {Langford}(2013)}]{langford2013Errors}%
  \BibitemOpen
  \bibfield  {author} {\bibinfo {author} {\bibfnamefont {N.~K.}\ \bibnamefont {Langford}},\ }\bibfield  {title} {\bibinfo {title} {Errors in quantum tomography: Diagnosing systematic versus statistical errors},\ }\href {https://doi.org/10.1088/1367-2630/15/3/035003} {\bibfield  {journal} {\bibinfo  {journal} {New Journal of Physics}\ }\textbf {\bibinfo {volume} {15}},\ \bibinfo {pages} {035003} (\bibinfo {year} {2013})}\BibitemShut {NoStop}%
\bibitem [{\citenamefont {Moroder}\ \emph {et~al.}(2013)\citenamefont {Moroder}, \citenamefont {Kleinmann}, \citenamefont {Schindler}, \citenamefont {Monz}, \citenamefont {G{\"u}hne},\ and\ \citenamefont {Blatt}}]{moroder2013Certifying}%
  \BibitemOpen
  \bibfield  {author} {\bibinfo {author} {\bibfnamefont {T.}~\bibnamefont {Moroder}}, \bibinfo {author} {\bibfnamefont {M.}~\bibnamefont {Kleinmann}}, \bibinfo {author} {\bibfnamefont {P.}~\bibnamefont {Schindler}}, \bibinfo {author} {\bibfnamefont {T.}~\bibnamefont {Monz}}, \bibinfo {author} {\bibfnamefont {O.}~\bibnamefont {G{\"u}hne}},\ and\ \bibinfo {author} {\bibfnamefont {R.}~\bibnamefont {Blatt}},\ }\bibfield  {title} {\bibinfo {title} {Certifying {{Systematic Errors}} in {{Quantum Experiments}}},\ }\href {https://doi.org/10.1103/PhysRevLett.110.180401} {\bibfield  {journal} {\bibinfo  {journal} {Physical Review Letters}\ }\textbf {\bibinfo {volume} {110}},\ \bibinfo {pages} {180401} (\bibinfo {year} {2013})}\BibitemShut {NoStop}%
\bibitem [{\citenamefont {Schwarz}\ and\ \citenamefont {{van Enk}}(2013)}]{schwarz2013Error}%
  \BibitemOpen
  \bibfield  {author} {\bibinfo {author} {\bibfnamefont {L.}~\bibnamefont {Schwarz}}\ and\ \bibinfo {author} {\bibfnamefont {S.~J.}\ \bibnamefont {{van Enk}}},\ }\bibfield  {title} {\bibinfo {title} {Error models in quantum computation: {{An}} application of model selection},\ }\href {https://doi.org/10.1103/PhysRevA.88.032318} {\bibfield  {journal} {\bibinfo  {journal} {Physical Review A}\ }\textbf {\bibinfo {volume} {88}},\ \bibinfo {pages} {032318} (\bibinfo {year} {2013})}\BibitemShut {NoStop}%
\bibitem [{\citenamefont {Knips}\ \emph {et~al.}(2015)\citenamefont {Knips}, \citenamefont {Schwemmer}, \citenamefont {Klein}, \citenamefont {Reuter}, \citenamefont {T{\'o}th},\ and\ \citenamefont {Weinfurter}}]{knips2015How}%
  \BibitemOpen
  \bibfield  {author} {\bibinfo {author} {\bibfnamefont {L.}~\bibnamefont {Knips}}, \bibinfo {author} {\bibfnamefont {C.}~\bibnamefont {Schwemmer}}, \bibinfo {author} {\bibfnamefont {N.}~\bibnamefont {Klein}}, \bibinfo {author} {\bibfnamefont {J.}~\bibnamefont {Reuter}}, \bibinfo {author} {\bibfnamefont {G.}~\bibnamefont {T{\'o}th}},\ and\ \bibinfo {author} {\bibfnamefont {H.}~\bibnamefont {Weinfurter}},\ }\href {https://doi.org/10.48550/arXiv.1512.06866} {\bibinfo {title} {How long does it take to obtain a physical density matrix?}} (\bibinfo {year} {2015}),\ \Eprint {https://arxiv.org/abs/1512.06866} {arXiv:1512.06866} \BibitemShut {NoStop}%
\bibitem [{\citenamefont {Scholten}\ and\ \citenamefont {{Blume-Kohout}}(2018)}]{scholten2018Behavior}%
  \BibitemOpen
  \bibfield  {author} {\bibinfo {author} {\bibfnamefont {T.~L.}\ \bibnamefont {Scholten}}\ and\ \bibinfo {author} {\bibfnamefont {R.}~\bibnamefont {{Blume-Kohout}}},\ }\bibfield  {title} {\bibinfo {title} {Behavior of the maximum likelihood in quantum state tomography},\ }\href {https://doi.org/10.1088/1367-2630/aaa7e2} {\bibfield  {journal} {\bibinfo  {journal} {New Journal of Physics}\ }\textbf {\bibinfo {volume} {20}},\ \bibinfo {pages} {023050} (\bibinfo {year} {2018})}\BibitemShut {NoStop}%
\bibitem [{\citenamefont {Torlai}\ \emph {et~al.}(2018)\citenamefont {Torlai}, \citenamefont {Mazzola}, \citenamefont {Carrasquilla}, \citenamefont {Troyer}, \citenamefont {Melko},\ and\ \citenamefont {Carleo}}]{torlai2018Neuralnetwork}%
  \BibitemOpen
  \bibfield  {author} {\bibinfo {author} {\bibfnamefont {G.}~\bibnamefont {Torlai}}, \bibinfo {author} {\bibfnamefont {G.}~\bibnamefont {Mazzola}}, \bibinfo {author} {\bibfnamefont {J.}~\bibnamefont {Carrasquilla}}, \bibinfo {author} {\bibfnamefont {M.}~\bibnamefont {Troyer}}, \bibinfo {author} {\bibfnamefont {R.}~\bibnamefont {Melko}},\ and\ \bibinfo {author} {\bibfnamefont {G.}~\bibnamefont {Carleo}},\ }\bibfield  {title} {\bibinfo {title} {Neural-network quantum state tomography},\ }\href {https://doi.org/10.1038/s41567-018-0048-5} {\bibfield  {journal} {\bibinfo  {journal} {Nature Physics}\ }\textbf {\bibinfo {volume} {14}},\ \bibinfo {pages} {447} (\bibinfo {year} {2018})}\BibitemShut {NoStop}%
\bibitem [{\citenamefont {Carrasquilla}\ \emph {et~al.}(2019)\citenamefont {Carrasquilla}, \citenamefont {Torlai}, \citenamefont {Melko},\ and\ \citenamefont {Aolita}}]{carrasquilla2019Reconstructing}%
  \BibitemOpen
  \bibfield  {author} {\bibinfo {author} {\bibfnamefont {J.}~\bibnamefont {Carrasquilla}}, \bibinfo {author} {\bibfnamefont {G.}~\bibnamefont {Torlai}}, \bibinfo {author} {\bibfnamefont {R.~G.}\ \bibnamefont {Melko}},\ and\ \bibinfo {author} {\bibfnamefont {L.}~\bibnamefont {Aolita}},\ }\bibfield  {title} {\bibinfo {title} {Reconstructing quantum states with generative models},\ }\href {https://doi.org/10.1038/s42256-019-0028-1} {\bibfield  {journal} {\bibinfo  {journal} {Nature Machine Intelligence}\ }\textbf {\bibinfo {volume} {1}},\ \bibinfo {pages} {155} (\bibinfo {year} {2019})}\BibitemShut {NoStop}%
\bibitem [{\citenamefont {Schmale}\ \emph {et~al.}(2022)\citenamefont {Schmale}, \citenamefont {Reh},\ and\ \citenamefont {G{\"a}rttner}}]{schmale2022Efficient}%
  \BibitemOpen
  \bibfield  {author} {\bibinfo {author} {\bibfnamefont {T.}~\bibnamefont {Schmale}}, \bibinfo {author} {\bibfnamefont {M.}~\bibnamefont {Reh}},\ and\ \bibinfo {author} {\bibfnamefont {M.}~\bibnamefont {G{\"a}rttner}},\ }\bibfield  {title} {\bibinfo {title} {Efficient quantum state tomography with convolutional neural networks},\ }\href {https://doi.org/10.1038/s41534-022-00621-4} {\bibfield  {journal} {\bibinfo  {journal} {npj Quantum Information}\ }\textbf {\bibinfo {volume} {8}},\ \bibinfo {pages} {1} (\bibinfo {year} {2022})}\BibitemShut {NoStop}%
\bibitem [{\citenamefont {Cha}\ \emph {et~al.}(2021)\citenamefont {Cha}, \citenamefont {Ginsparg}, \citenamefont {Wu}, \citenamefont {Carrasquilla}, \citenamefont {McMahon},\ and\ \citenamefont {Kim}}]{cha2021Attentionbased}%
  \BibitemOpen
  \bibfield  {author} {\bibinfo {author} {\bibfnamefont {P.}~\bibnamefont {Cha}}, \bibinfo {author} {\bibfnamefont {P.}~\bibnamefont {Ginsparg}}, \bibinfo {author} {\bibfnamefont {F.}~\bibnamefont {Wu}}, \bibinfo {author} {\bibfnamefont {J.}~\bibnamefont {Carrasquilla}}, \bibinfo {author} {\bibfnamefont {P.~L.}\ \bibnamefont {McMahon}},\ and\ \bibinfo {author} {\bibfnamefont {E.-A.}\ \bibnamefont {Kim}},\ }\bibfield  {title} {\bibinfo {title} {Attention-based quantum tomography},\ }\href {https://doi.org/10.1088/2632-2153/ac362b} {\bibfield  {journal} {\bibinfo  {journal} {Machine Learning: Science and Technology}\ }\textbf {\bibinfo {volume} {3}},\ \bibinfo {pages} {01LT01} (\bibinfo {year} {2021})}\BibitemShut {NoStop}%
\bibitem [{\citenamefont {Kieferov{\'a}}\ and\ \citenamefont {Wiebe}(2017)}]{kieferova2017Tomography}%
  \BibitemOpen
  \bibfield  {author} {\bibinfo {author} {\bibfnamefont {M.}~\bibnamefont {Kieferov{\'a}}}\ and\ \bibinfo {author} {\bibfnamefont {N.}~\bibnamefont {Wiebe}},\ }\bibfield  {title} {\bibinfo {title} {Tomography and generative training with quantum {{Boltzmann}} machines},\ }\href {https://doi.org/10.1103/PhysRevA.96.062327} {\bibfield  {journal} {\bibinfo  {journal} {Physical Review A}\ }\textbf {\bibinfo {volume} {96}},\ \bibinfo {pages} {062327} (\bibinfo {year} {2017})}\BibitemShut {NoStop}%
\bibitem [{\citenamefont {Kieferova}\ \emph {et~al.}(2021)\citenamefont {Kieferova}, \citenamefont {Carlos},\ and\ \citenamefont {Wiebe}}]{kieferova2021Quantum}%
  \BibitemOpen
  \bibfield  {author} {\bibinfo {author} {\bibfnamefont {M.}~\bibnamefont {Kieferova}}, \bibinfo {author} {\bibfnamefont {O.~M.}\ \bibnamefont {Carlos}},\ and\ \bibinfo {author} {\bibfnamefont {N.}~\bibnamefont {Wiebe}},\ }\href {https://doi.org/10.48550/arXiv.2106.09567} {\bibinfo {title} {Quantum {{Generative Training Using R}}{\textbackslash}'enyi {{Divergences}}}} (\bibinfo {year} {2021}),\ \Eprint {https://arxiv.org/abs/2106.09567} {arXiv:2106.09567} \BibitemShut {NoStop}%
\bibitem [{\citenamefont {Larocca}\ \emph {et~al.}(2023)\citenamefont {Larocca}, \citenamefont {Ju}, \citenamefont {Garc{\'\i}a-Mart{\'\i}n}, \citenamefont {Coles},\ and\ \citenamefont {Cerezo}}]{larocca2023theory}%
  \BibitemOpen
  \bibfield  {author} {\bibinfo {author} {\bibfnamefont {M.}~\bibnamefont {Larocca}}, \bibinfo {author} {\bibfnamefont {N.}~\bibnamefont {Ju}}, \bibinfo {author} {\bibfnamefont {D.}~\bibnamefont {Garc{\'\i}a-Mart{\'\i}n}}, \bibinfo {author} {\bibfnamefont {P.~J.}\ \bibnamefont {Coles}},\ and\ \bibinfo {author} {\bibfnamefont {M.}~\bibnamefont {Cerezo}},\ }\bibfield  {title} {\bibinfo {title} {Theory of overparametrization in quantum neural networks},\ }\href@noop {} {\bibfield  {journal} {\bibinfo  {journal} {Nature Computational Science}\ }\textbf {\bibinfo {volume} {3}},\ \bibinfo {pages} {542} (\bibinfo {year} {2023})}\BibitemShut {NoStop}%
\bibitem [{\citenamefont {{Garc{\'i}a-Mart{\'i}n}}\ \emph {et~al.}(2023)\citenamefont {{Garc{\'i}a-Mart{\'i}n}}, \citenamefont {Larocca},\ and\ \citenamefont {Cerezo}}]{garcia-martin2023Deep}%
  \BibitemOpen
  \bibfield  {author} {\bibinfo {author} {\bibfnamefont {D.}~\bibnamefont {{Garc{\'i}a-Mart{\'i}n}}}, \bibinfo {author} {\bibfnamefont {M.}~\bibnamefont {Larocca}},\ and\ \bibinfo {author} {\bibfnamefont {M.}~\bibnamefont {Cerezo}},\ }\href {https://doi.org/10.48550/arXiv.2305.09957} {\bibinfo {title} {Deep quantum neural networks form {{Gaussian}} processes}} (\bibinfo {year} {2023}),\ \Eprint {https://arxiv.org/abs/2305.09957} {arXiv:2305.09957} \BibitemShut {NoStop}%
\bibitem [{\citenamefont {Liu}\ \emph {et~al.}(2022)\citenamefont {Liu}, \citenamefont {Tacchino}, \citenamefont {Glick}, \citenamefont {Jiang},\ and\ \citenamefont {Mezzacapo}}]{liu2022Representation}%
  \BibitemOpen
  \bibfield  {author} {\bibinfo {author} {\bibfnamefont {J.}~\bibnamefont {Liu}}, \bibinfo {author} {\bibfnamefont {F.}~\bibnamefont {Tacchino}}, \bibinfo {author} {\bibfnamefont {J.~R.}\ \bibnamefont {Glick}}, \bibinfo {author} {\bibfnamefont {L.}~\bibnamefont {Jiang}},\ and\ \bibinfo {author} {\bibfnamefont {A.}~\bibnamefont {Mezzacapo}},\ }\bibfield  {title} {\bibinfo {title} {Representation {{Learning}} via {{Quantum Neural Tangent Kernels}}},\ }\href {https://doi.org/10.1103/PRXQuantum.3.030323} {\bibfield  {journal} {\bibinfo  {journal} {PRX Quantum}\ }\textbf {\bibinfo {volume} {3}},\ \bibinfo {pages} {030323} (\bibinfo {year} {2022})}\BibitemShut {NoStop}%
\bibitem [{\citenamefont {Yano}\ and\ \citenamefont {Yamamoto}(2023)}]{yano2023Quantuma}%
  \BibitemOpen
  \bibfield  {author} {\bibinfo {author} {\bibfnamefont {H.}~\bibnamefont {Yano}}\ and\ \bibinfo {author} {\bibfnamefont {N.}~\bibnamefont {Yamamoto}},\ }\bibfield  {title} {\bibinfo {title} {Quantum information criteria for model selection in quantum state estimation},\ }\href {https://doi.org/10.1088/1751-8121/acf747} {\bibfield  {journal} {\bibinfo  {journal} {Journal of Physics A: Mathematical and Theoretical}\ }\textbf {\bibinfo {volume} {56}},\ \bibinfo {pages} {405301} (\bibinfo {year} {2023})}\BibitemShut {NoStop}%
\bibitem [{\citenamefont {Yamagata}\ \emph {et~al.}(2013)\citenamefont {Yamagata}, \citenamefont {Fujiwara},\ and\ \citenamefont {Gill}}]{yamagata2013Quantum}%
  \BibitemOpen
  \bibfield  {author} {\bibinfo {author} {\bibfnamefont {K.}~\bibnamefont {Yamagata}}, \bibinfo {author} {\bibfnamefont {A.}~\bibnamefont {Fujiwara}},\ and\ \bibinfo {author} {\bibfnamefont {R.~D.}\ \bibnamefont {Gill}},\ }\bibfield  {title} {\bibinfo {title} {Quantum local asymptotic normality based on a new quantum likelihood ratio},\ }\href {https://doi.org/10.1214/13-AOS1147} {\bibfield  {journal} {\bibinfo  {journal} {The Annals of Statistics}\ }\textbf {\bibinfo {volume} {41}},\ \bibinfo {pages} {2197} (\bibinfo {year} {2013})}\BibitemShut {NoStop}%
\bibitem [{\citenamefont {Huang}\ \emph {et~al.}(2020)\citenamefont {Huang}, \citenamefont {Kueng},\ and\ \citenamefont {Preskill}}]{huang2020Predicting}%
  \BibitemOpen
  \bibfield  {author} {\bibinfo {author} {\bibfnamefont {H.-Y.}\ \bibnamefont {Huang}}, \bibinfo {author} {\bibfnamefont {R.}~\bibnamefont {Kueng}},\ and\ \bibinfo {author} {\bibfnamefont {J.}~\bibnamefont {Preskill}},\ }\bibfield  {title} {\bibinfo {title} {Predicting many properties of a quantum system from very few measurements},\ }\href {https://doi.org/10.1038/s41567-020-0932-7} {\bibfield  {journal} {\bibinfo  {journal} {Nature Physics}\ }\textbf {\bibinfo {volume} {16}},\ \bibinfo {pages} {1050} (\bibinfo {year} {2020})}\BibitemShut {NoStop}%
\bibitem [{\citenamefont {Huang}\ \emph {et~al.}(2021)\citenamefont {Huang}, \citenamefont {Broughton}, \citenamefont {Mohseni}, \citenamefont {Babbush}, \citenamefont {Boixo}, \citenamefont {Neven},\ and\ \citenamefont {McClean}}]{huang2021Power}%
  \BibitemOpen
  \bibfield  {author} {\bibinfo {author} {\bibfnamefont {H.-Y.}\ \bibnamefont {Huang}}, \bibinfo {author} {\bibfnamefont {M.}~\bibnamefont {Broughton}}, \bibinfo {author} {\bibfnamefont {M.}~\bibnamefont {Mohseni}}, \bibinfo {author} {\bibfnamefont {R.}~\bibnamefont {Babbush}}, \bibinfo {author} {\bibfnamefont {S.}~\bibnamefont {Boixo}}, \bibinfo {author} {\bibfnamefont {H.}~\bibnamefont {Neven}},\ and\ \bibinfo {author} {\bibfnamefont {J.~R.}\ \bibnamefont {McClean}},\ }\bibfield  {title} {\bibinfo {title} {Power of data in quantum machine learning},\ }\href {https://doi.org/10.1038/s41467-021-22539-9} {\bibfield  {journal} {\bibinfo  {journal} {Nature Communications}\ }\textbf {\bibinfo {volume} {12}},\ \bibinfo {pages} {2631} (\bibinfo {year} {2021})}\BibitemShut {NoStop}%
\bibitem [{\citenamefont {Huang}\ \emph {et~al.}(2022)\citenamefont {Huang}, \citenamefont {Kueng}, \citenamefont {Torlai}, \citenamefont {Albert},\ and\ \citenamefont {Preskill}}]{huang2022Provably}%
  \BibitemOpen
  \bibfield  {author} {\bibinfo {author} {\bibfnamefont {H.-Y.}\ \bibnamefont {Huang}}, \bibinfo {author} {\bibfnamefont {R.}~\bibnamefont {Kueng}}, \bibinfo {author} {\bibfnamefont {G.}~\bibnamefont {Torlai}}, \bibinfo {author} {\bibfnamefont {V.~V.}\ \bibnamefont {Albert}},\ and\ \bibinfo {author} {\bibfnamefont {J.}~\bibnamefont {Preskill}},\ }\bibfield  {title} {\bibinfo {title} {Provably efficient machine learning for quantum many-body problems},\ }\href {https://doi.org/10.1126/science.abk3333} {\bibfield  {journal} {\bibinfo  {journal} {Science}\ }\textbf {\bibinfo {volume} {377}},\ \bibinfo {pages} {eabk3333} (\bibinfo {year} {2022})}\BibitemShut {NoStop}%
\bibitem [{\citenamefont {Kumagai}\ and\ \citenamefont {Hayashi}(2013)}]{kumagai2013Quantum}%
  \BibitemOpen
  \bibfield  {author} {\bibinfo {author} {\bibfnamefont {W.}~\bibnamefont {Kumagai}}\ and\ \bibinfo {author} {\bibfnamefont {M.}~\bibnamefont {Hayashi}},\ }\bibfield  {title} {\bibinfo {title} {Quantum {{Hypothesis Testing}} for {{Gaussian States}}: {{Quantum Analogues}} of {$X$}2, t-, and {{F-Tests}}},\ }\href {https://doi.org/10.1007/s00220-013-1678-1} {\bibfield  {journal} {\bibinfo  {journal} {Communications in Mathematical Physics}\ }\textbf {\bibinfo {volume} {318}},\ \bibinfo {pages} {535} (\bibinfo {year} {2013})}\BibitemShut {NoStop}%
\bibitem [{\citenamefont {Watanabe}(2024)}]{watanabe2024recent}%
  \BibitemOpen
  \bibfield  {author} {\bibinfo {author} {\bibfnamefont {S.}~\bibnamefont {Watanabe}},\ }\bibfield  {title} {\bibinfo {title} {Recent advances in algebraic geometry and {B}ayesian statistics},\ }\href@noop {} {\bibfield  {journal} {\bibinfo  {journal} {Information Geometry}\ }\textbf {\bibinfo {volume} {7}},\ \bibinfo {pages} {187} (\bibinfo {year} {2024})}\BibitemShut {NoStop}%
\bibitem [{\citenamefont {Amari}\ \emph {et~al.}(1992)\citenamefont {Amari}, \citenamefont {Fujita},\ and\ \citenamefont {Shinomoto}}]{amari1992four}%
  \BibitemOpen
  \bibfield  {author} {\bibinfo {author} {\bibfnamefont {S.-i.}\ \bibnamefont {Amari}}, \bibinfo {author} {\bibfnamefont {N.}~\bibnamefont {Fujita}},\ and\ \bibinfo {author} {\bibfnamefont {S.}~\bibnamefont {Shinomoto}},\ }\bibfield  {title} {\bibinfo {title} {Four types of learning curves},\ }\href@noop {} {\bibfield  {journal} {\bibinfo  {journal} {Neural Computation}\ }\textbf {\bibinfo {volume} {4}},\ \bibinfo {pages} {605} (\bibinfo {year} {1992})}\BibitemShut {NoStop}%
\bibitem [{\citenamefont {Amari}\ and\ \citenamefont {Murata}(1993)}]{amari1993statistical}%
  \BibitemOpen
  \bibfield  {author} {\bibinfo {author} {\bibfnamefont {S.-i.}\ \bibnamefont {Amari}}\ and\ \bibinfo {author} {\bibfnamefont {N.}~\bibnamefont {Murata}},\ }\bibfield  {title} {\bibinfo {title} {Statistical theory of learning curves under entropic loss criterion},\ }\href@noop {} {\bibfield  {journal} {\bibinfo  {journal} {Neural Computation}\ }\textbf {\bibinfo {volume} {5}},\ \bibinfo {pages} {140} (\bibinfo {year} {1993})}\BibitemShut {NoStop}%
\bibitem [{\citenamefont {Murata}\ \emph {et~al.}(1994)\citenamefont {Murata}, \citenamefont {Yoshizawa},\ and\ \citenamefont {Amari}}]{murata1994network}%
  \BibitemOpen
  \bibfield  {author} {\bibinfo {author} {\bibfnamefont {N.}~\bibnamefont {Murata}}, \bibinfo {author} {\bibfnamefont {S.}~\bibnamefont {Yoshizawa}},\ and\ \bibinfo {author} {\bibfnamefont {S.-i.}\ \bibnamefont {Amari}},\ }\bibfield  {title} {\bibinfo {title} {Network information criterion-determining the number of hidden units for an artificial neural network model},\ }\href@noop {} {\bibfield  {journal} {\bibinfo  {journal} {IEEE transactions on neural networks}\ }\textbf {\bibinfo {volume} {5}},\ \bibinfo {pages} {865} (\bibinfo {year} {1994})}\BibitemShut {NoStop}%
\bibitem [{\citenamefont {Fraley}\ and\ \citenamefont {Raftery}(2002)}]{fraley2002model}%
  \BibitemOpen
  \bibfield  {author} {\bibinfo {author} {\bibfnamefont {C.}~\bibnamefont {Fraley}}\ and\ \bibinfo {author} {\bibfnamefont {A.~E.}\ \bibnamefont {Raftery}},\ }\bibfield  {title} {\bibinfo {title} {Model-based clustering, discriminant analysis, and density estimation},\ }\href@noop {} {\bibfield  {journal} {\bibinfo  {journal} {Journal of the American statistical Association}\ }\textbf {\bibinfo {volume} {97}},\ \bibinfo {pages} {611} (\bibinfo {year} {2002})}\BibitemShut {NoStop}%
\bibitem [{\citenamefont {Posada}\ and\ \citenamefont {Crandall}(1998)}]{posada1998modeltest}%
  \BibitemOpen
  \bibfield  {author} {\bibinfo {author} {\bibfnamefont {D.}~\bibnamefont {Posada}}\ and\ \bibinfo {author} {\bibfnamefont {K.~A.}\ \bibnamefont {Crandall}},\ }\bibfield  {title} {\bibinfo {title} {Modeltest: testing the model of dna substitution.},\ }\href@noop {} {\bibfield  {journal} {\bibinfo  {journal} {Bioinformatics (Oxford, England)}\ }\textbf {\bibinfo {volume} {14}},\ \bibinfo {pages} {817} (\bibinfo {year} {1998})}\BibitemShut {NoStop}%
\bibitem [{\citenamefont {Posada}\ and\ \citenamefont {Buckley}(2004)}]{posada2004model}%
  \BibitemOpen
  \bibfield  {author} {\bibinfo {author} {\bibfnamefont {D.}~\bibnamefont {Posada}}\ and\ \bibinfo {author} {\bibfnamefont {T.~R.}\ \bibnamefont {Buckley}},\ }\bibfield  {title} {\bibinfo {title} {Model selection and model averaging in phylogenetics: advantages of {A}kaike information criterion and {B}ayesian approaches over likelihood ratio tests},\ }\href@noop {} {\bibfield  {journal} {\bibinfo  {journal} {Systematic biology}\ }\textbf {\bibinfo {volume} {53}},\ \bibinfo {pages} {793} (\bibinfo {year} {2004})}\BibitemShut {NoStop}%
\bibitem [{\citenamefont {Sclove}(1987)}]{sclove1987application}%
  \BibitemOpen
  \bibfield  {author} {\bibinfo {author} {\bibfnamefont {S.~L.}\ \bibnamefont {Sclove}},\ }\bibfield  {title} {\bibinfo {title} {Application of model-selection criteria to some problems in multivariate analysis},\ }\href@noop {} {\bibfield  {journal} {\bibinfo  {journal} {Psychometrika}\ }\textbf {\bibinfo {volume} {52}},\ \bibinfo {pages} {333} (\bibinfo {year} {1987})}\BibitemShut {NoStop}%
\bibitem [{\citenamefont {Watanabe}(2023)}]{watanabe2023mathematical}%
  \BibitemOpen
  \bibfield  {author} {\bibinfo {author} {\bibfnamefont {S.}~\bibnamefont {Watanabe}},\ }\bibfield  {title} {\bibinfo {title} {Mathematical theory of {B}ayesian statistics for unknown information source},\ }\href@noop {} {\bibfield  {journal} {\bibinfo  {journal} {Philosophical Transactions of the Royal Society A}\ }\textbf {\bibinfo {volume} {381}},\ \bibinfo {pages} {20220151} (\bibinfo {year} {2023})}\BibitemShut {NoStop}%
\bibitem [{\citenamefont {Seeger}(2004)}]{seeger2004gaussian}%
  \BibitemOpen
  \bibfield  {author} {\bibinfo {author} {\bibfnamefont {M.}~\bibnamefont {Seeger}},\ }\bibfield  {title} {\bibinfo {title} {Gaussian processes for machine learning},\ }\href@noop {} {\bibfield  {journal} {\bibinfo  {journal} {International journal of neural systems}\ }\textbf {\bibinfo {volume} {14}},\ \bibinfo {pages} {69} (\bibinfo {year} {2004})}\BibitemShut {NoStop}%
\bibitem [{\citenamefont {Aoyagi}\ and\ \citenamefont {Watanabe}(2005)}]{aoyagi2005stochastic}%
  \BibitemOpen
  \bibfield  {author} {\bibinfo {author} {\bibfnamefont {M.}~\bibnamefont {Aoyagi}}\ and\ \bibinfo {author} {\bibfnamefont {S.}~\bibnamefont {Watanabe}},\ }\bibfield  {title} {\bibinfo {title} {Stochastic complexities of reduced rank regression in {B}ayesian estimation},\ }\href@noop {} {\bibfield  {journal} {\bibinfo  {journal} {Neural Networks}\ }\textbf {\bibinfo {volume} {18}},\ \bibinfo {pages} {924} (\bibinfo {year} {2005})}\BibitemShut {NoStop}%
\bibitem [{\citenamefont {Hayashi}\ and\ \citenamefont {Watanabe}(2017)}]{hayashi2017upper}%
  \BibitemOpen
  \bibfield  {author} {\bibinfo {author} {\bibfnamefont {N.}~\bibnamefont {Hayashi}}\ and\ \bibinfo {author} {\bibfnamefont {S.}~\bibnamefont {Watanabe}},\ }\bibfield  {title} {\bibinfo {title} {Upper bound of {B}ayesian generalization error in non-negative matrix factorization},\ }\href@noop {} {\bibfield  {journal} {\bibinfo  {journal} {Neurocomputing}\ }\textbf {\bibinfo {volume} {266}},\ \bibinfo {pages} {21} (\bibinfo {year} {2017})}\BibitemShut {NoStop}%
\bibitem [{\citenamefont {Zellner}(1976)}]{zellner1976bayesian}%
  \BibitemOpen
  \bibfield  {author} {\bibinfo {author} {\bibfnamefont {A.}~\bibnamefont {Zellner}},\ }\bibfield  {title} {\bibinfo {title} {Bayesian and non-{B}ayesian analysis of the regression model with multivariate {S}tudent-t error terms},\ }\href@noop {} {\bibfield  {journal} {\bibinfo  {journal} {Journal of the American Statistical Association}\ }\textbf {\bibinfo {volume} {71}},\ \bibinfo {pages} {400} (\bibinfo {year} {1976})}\BibitemShut {NoStop}%
\bibitem [{\citenamefont {Yamazaki}\ and\ \citenamefont {Watanabe}(2005)}]{yamazaki2005algebraic}%
  \BibitemOpen
  \bibfield  {author} {\bibinfo {author} {\bibfnamefont {K.}~\bibnamefont {Yamazaki}}\ and\ \bibinfo {author} {\bibfnamefont {S.}~\bibnamefont {Watanabe}},\ }\bibfield  {title} {\bibinfo {title} {Algebraic geometry and stochastic complexity of hidden markov models},\ }\href@noop {} {\bibfield  {journal} {\bibinfo  {journal} {Neurocomputing}\ }\textbf {\bibinfo {volume} {69}},\ \bibinfo {pages} {62} (\bibinfo {year} {2005})}\BibitemShut {NoStop}%
\bibitem [{\citenamefont {Zwiernik}(2011)}]{zwiernik2011asymptotic}%
  \BibitemOpen
  \bibfield  {author} {\bibinfo {author} {\bibfnamefont {P.}~\bibnamefont {Zwiernik}},\ }\bibfield  {title} {\bibinfo {title} {An asymptotic behaviour of the marginal likelihood for general markov models},\ }\href@noop {} {\bibfield  {journal} {\bibinfo  {journal} {The Journal of Machine Learning Research}\ }\textbf {\bibinfo {volume} {12}},\ \bibinfo {pages} {3283} (\bibinfo {year} {2011})}\BibitemShut {NoStop}%
\bibitem [{\citenamefont {Atiyah}(1970)}]{atiyah1970resolution}%
  \BibitemOpen
  \bibfield  {author} {\bibinfo {author} {\bibfnamefont {M.~F.}\ \bibnamefont {Atiyah}},\ }\bibfield  {title} {\bibinfo {title} {Resolution of singularities and division of distributions},\ }\href@noop {} {\bibfield  {journal} {\bibinfo  {journal} {Communications on pure and applied mathematics}\ }\textbf {\bibinfo {volume} {23}},\ \bibinfo {pages} {145} (\bibinfo {year} {1970})}\BibitemShut {NoStop}%
\bibitem [{\citenamefont {Igusa}(2000)}]{igusa2000introduction}%
  \BibitemOpen
  \bibfield  {author} {\bibinfo {author} {\bibfnamefont {J.-i.}\ \bibnamefont {Igusa}},\ }\bibfield  {title} {\bibinfo {title} {An introduction to the theory of local zeta functions},\ }\href@noop {} {\bibfield  {journal} {\bibinfo  {journal} {American Mathematical Soc.}\ }\textbf {\bibinfo {volume} {14}} (\bibinfo {year} {2000})}\BibitemShut {NoStop}%
\bibitem [{\citenamefont {Hironaka}(1964{\natexlab{a}})}]{hironaka1964resolutionI}%
  \BibitemOpen
  \bibfield  {author} {\bibinfo {author} {\bibfnamefont {H.}~\bibnamefont {Hironaka}},\ }\bibfield  {title} {\bibinfo {title} {Resolution of singularities of an algebraic variety over a field of characteristic zero. {I}},\ }\href {https://doi.org/10.2307/1970486} {\bibfield  {journal} {\bibinfo  {journal} {Ann. Math. (2)}\ }\textbf {\bibinfo {volume} {79}},\ \bibinfo {pages} {109} (\bibinfo {year} {1964}{\natexlab{a}})}\BibitemShut {NoStop}%
\bibitem [{\citenamefont {Hironaka}(1964{\natexlab{b}})}]{hironaka1964resolutionII}%
  \BibitemOpen
  \bibfield  {author} {\bibinfo {author} {\bibfnamefont {H.}~\bibnamefont {Hironaka}},\ }\bibfield  {title} {\bibinfo {title} {Resolution of singularities of an algebraic variety over a field of characteristic zero. {II}},\ }\href {https://doi.org/10.2307/1970547} {\bibfield  {journal} {\bibinfo  {journal} {Ann. Math. (2)}\ }\textbf {\bibinfo {volume} {79}},\ \bibinfo {pages} {205} (\bibinfo {year} {1964}{\natexlab{b}})}\BibitemShut {NoStop}%
\bibitem [{\citenamefont {Koll{\'a}r}\ and\ \citenamefont {Mori}(1998)}]{kollar1998birational}%
  \BibitemOpen
  \bibfield  {author} {\bibinfo {author} {\bibfnamefont {J.}~\bibnamefont {Koll{\'a}r}}\ and\ \bibinfo {author} {\bibfnamefont {S.}~\bibnamefont {Mori}},\ }\bibfield  {title} {\bibinfo {title} {Birational geometry of algebraic varieties. {With} the collaboration of {C}. {H}. {Clemens} and {A}. {Corti}},\ }\href@noop {} {\bibfield  {journal} {\bibinfo  {journal} {Cambridge: Cambridge University Press}\ }\bibinfo {series} {Camb. Tracts Math.},\ \textbf {\bibinfo {volume} {134}} (\bibinfo {year} {1998})}\BibitemShut {NoStop}%
\bibitem [{\citenamefont {Hacon}\ \emph {et~al.}(2014)\citenamefont {Hacon}, \citenamefont {McKernan},\ and\ \citenamefont {Xu}}]{hacon2014acc}%
  \BibitemOpen
  \bibfield  {author} {\bibinfo {author} {\bibfnamefont {C.~D.}\ \bibnamefont {Hacon}}, \bibinfo {author} {\bibfnamefont {J.}~\bibnamefont {McKernan}},\ and\ \bibinfo {author} {\bibfnamefont {C.}~\bibnamefont {Xu}},\ }\bibfield  {title} {\bibinfo {title} {Acc for log canonical thresholds},\ }\href@noop {} {\bibfield  {journal} {\bibinfo  {journal} {Annals of Mathematics}\ ,\ \bibinfo {pages} {523}} (\bibinfo {year} {2014})}\BibitemShut {NoStop}%
\bibitem [{\citenamefont {Watanabe}\ and\ \citenamefont {Opper}(2010)}]{watanabe2010asymptotic}%
  \BibitemOpen
  \bibfield  {author} {\bibinfo {author} {\bibfnamefont {S.}~\bibnamefont {Watanabe}}\ and\ \bibinfo {author} {\bibfnamefont {M.}~\bibnamefont {Opper}},\ }\bibfield  {title} {\bibinfo {title} {Asymptotic equivalence of {B}ayes cross validation and widely applicable information criterion in singular learning theory.},\ }\href@noop {} {\bibfield  {journal} {\bibinfo  {journal} {Journal of machine learning research}\ }\textbf {\bibinfo {volume} {11}} (\bibinfo {year} {2010})}\BibitemShut {NoStop}%
\bibitem [{\citenamefont {Vehtari}\ \emph {et~al.}(2017)\citenamefont {Vehtari}, \citenamefont {Gelman},\ and\ \citenamefont {Gabry}}]{vehtari2017practical}%
  \BibitemOpen
  \bibfield  {author} {\bibinfo {author} {\bibfnamefont {A.}~\bibnamefont {Vehtari}}, \bibinfo {author} {\bibfnamefont {A.}~\bibnamefont {Gelman}},\ and\ \bibinfo {author} {\bibfnamefont {J.}~\bibnamefont {Gabry}},\ }\bibfield  {title} {\bibinfo {title} {Practical {B}ayesian model evaluation using leave-one-out cross-validation and {WAIC}},\ }\href@noop {} {\bibfield  {journal} {\bibinfo  {journal} {Statistics and computing}\ }\textbf {\bibinfo {volume} {27}},\ \bibinfo {pages} {1413} (\bibinfo {year} {2017})}\BibitemShut {NoStop}%
\bibitem [{\citenamefont {B{\"u}rkner}(2017)}]{burkner2017brms}%
  \BibitemOpen
  \bibfield  {author} {\bibinfo {author} {\bibfnamefont {P.-C.}\ \bibnamefont {B{\"u}rkner}},\ }\bibfield  {title} {\bibinfo {title} {brms: An {R} package for {B}ayesian multilevel models using stan},\ }\href@noop {} {\bibfield  {journal} {\bibinfo  {journal} {Journal of statistical software}\ }\textbf {\bibinfo {volume} {80}},\ \bibinfo {pages} {1} (\bibinfo {year} {2017})}\BibitemShut {NoStop}%
\bibitem [{\citenamefont {Gronau}\ and\ \citenamefont {Wagenmakers}(2019)}]{gronau2019limitations}%
  \BibitemOpen
  \bibfield  {author} {\bibinfo {author} {\bibfnamefont {Q.~F.}\ \bibnamefont {Gronau}}\ and\ \bibinfo {author} {\bibfnamefont {E.-J.}\ \bibnamefont {Wagenmakers}},\ }\bibfield  {title} {\bibinfo {title} {Limitations of {B}ayesian leave-one-out cross-validation for model selection},\ }\href@noop {} {\bibfield  {journal} {\bibinfo  {journal} {Computational brain \& behavior}\ }\textbf {\bibinfo {volume} {2}},\ \bibinfo {pages} {1} (\bibinfo {year} {2019})}\BibitemShut {NoStop}%
\bibitem [{\citenamefont {Watanabe}(2013)}]{watanabe2013waic}%
  \BibitemOpen
  \bibfield  {author} {\bibinfo {author} {\bibfnamefont {S.}~\bibnamefont {Watanabe}},\ }\bibfield  {title} {\bibinfo {title} {{WAIC} and {WBIC} are information criteria for singular statistical model evaluation},\ }in\ \href@noop {} {\emph {\bibinfo {booktitle} {Proceedings of the Workshop on Information Theoretic Methods in Science and Engineering}}}\ (\bibinfo {year} {2013})\ pp.\ \bibinfo {pages} {90--94}\BibitemShut {NoStop}%
\bibitem [{\citenamefont {Choi}\ \emph {et~al.}(2018)\citenamefont {Choi}, \citenamefont {Jang},\ and\ \citenamefont {Alemi}}]{choi2018waic}%
  \BibitemOpen
  \bibfield  {author} {\bibinfo {author} {\bibfnamefont {H.}~\bibnamefont {Choi}}, \bibinfo {author} {\bibfnamefont {E.}~\bibnamefont {Jang}},\ and\ \bibinfo {author} {\bibfnamefont {A.~A.}\ \bibnamefont {Alemi}},\ }\bibfield  {title} {\bibinfo {title} {{WAIC}, but why? generative ensembles for robust anomaly detection},\ }\href@noop {} {\bibfield  {journal} {\bibinfo  {journal} {arXiv preprint arXiv:1810.01392}\ } (\bibinfo {year} {2018})}\BibitemShut {NoStop}%
\bibitem [{\citenamefont {Du}\ \emph {et~al.}(2024)\citenamefont {Du}, \citenamefont {Keller}, \citenamefont {Alacam},\ and\ \citenamefont {Enders}}]{du2024comparing}%
  \BibitemOpen
  \bibfield  {author} {\bibinfo {author} {\bibfnamefont {H.}~\bibnamefont {Du}}, \bibinfo {author} {\bibfnamefont {B.}~\bibnamefont {Keller}}, \bibinfo {author} {\bibfnamefont {E.}~\bibnamefont {Alacam}},\ and\ \bibinfo {author} {\bibfnamefont {C.}~\bibnamefont {Enders}},\ }\bibfield  {title} {\bibinfo {title} {Comparing {DIC} and {WAIC} for multilevel models with missing data},\ }\href@noop {} {\bibfield  {journal} {\bibinfo  {journal} {Behavior Research Methods}\ }\textbf {\bibinfo {volume} {56}},\ \bibinfo {pages} {2731} (\bibinfo {year} {2024})}\BibitemShut {NoStop}%
\bibitem [{\citenamefont {Gelman}\ \emph {et~al.}(2014)\citenamefont {Gelman}, \citenamefont {Hwang},\ and\ \citenamefont {Vehtari}}]{gelman2014understanding}%
  \BibitemOpen
  \bibfield  {author} {\bibinfo {author} {\bibfnamefont {A.}~\bibnamefont {Gelman}}, \bibinfo {author} {\bibfnamefont {J.}~\bibnamefont {Hwang}},\ and\ \bibinfo {author} {\bibfnamefont {A.}~\bibnamefont {Vehtari}},\ }\bibfield  {title} {\bibinfo {title} {Understanding predictive information criteria for {B}ayesian models},\ }\href@noop {} {\bibfield  {journal} {\bibinfo  {journal} {Statistics and computing}\ }\textbf {\bibinfo {volume} {24}},\ \bibinfo {pages} {997} (\bibinfo {year} {2014})}\BibitemShut {NoStop}%
\bibitem [{\citenamefont {Yao}\ \emph {et~al.}(2018)\citenamefont {Yao}, \citenamefont {Vehtari}, \citenamefont {Simpson},\ and\ \citenamefont {Gelman}}]{yao2018using}%
  \BibitemOpen
  \bibfield  {author} {\bibinfo {author} {\bibfnamefont {Y.}~\bibnamefont {Yao}}, \bibinfo {author} {\bibfnamefont {A.}~\bibnamefont {Vehtari}}, \bibinfo {author} {\bibfnamefont {D.}~\bibnamefont {Simpson}},\ and\ \bibinfo {author} {\bibfnamefont {A.}~\bibnamefont {Gelman}},\ }\bibfield  {title} {\bibinfo {title} {Using stacking to average {B}ayesian predictive distributions (with discussion)},\ }\href@noop {} {\bibfield  {journal} {\bibinfo  {journal} {Bayesian Analysis}\ }\textbf {\bibinfo {volume} {13}},\ \bibinfo {pages} {917} (\bibinfo {year} {2018})}\BibitemShut {NoStop}%
\bibitem [{\citenamefont {Lartillot}(2023)}]{lartillot2023identifying}%
  \BibitemOpen
  \bibfield  {author} {\bibinfo {author} {\bibfnamefont {N.}~\bibnamefont {Lartillot}},\ }\bibfield  {title} {\bibinfo {title} {Identifying the best approximating model in {B}ayesian phylogenetics: {B}ayes factors, cross-validation or {WAIC}?},\ }\href@noop {} {\bibfield  {journal} {\bibinfo  {journal} {Systematic Biology}\ }\textbf {\bibinfo {volume} {72}},\ \bibinfo {pages} {616} (\bibinfo {year} {2023})}\BibitemShut {NoStop}%
\bibitem [{\citenamefont {Hartshorne}(2013)}]{hartshorne2013algebraic}%
  \BibitemOpen
  \bibfield  {author} {\bibinfo {author} {\bibfnamefont {R.}~\bibnamefont {Hartshorne}},\ }\bibfield  {title} {\bibinfo {title} {Algebraic geometry},\ }\href@noop {} {\bibfield  {journal} {\bibinfo  {journal} {Springer Science \& Business Media}\ }\textbf {\bibinfo {volume} {52}} (\bibinfo {year} {2013})}\BibitemShut {NoStop}%
\bibitem [{\citenamefont {Verdon}\ \emph {et~al.}(2019)\citenamefont {Verdon}, \citenamefont {Marks}, \citenamefont {Nanda}, \citenamefont {Leichenauer},\ and\ \citenamefont {Hidary}}]{verdon2019Quantuma}%
  \BibitemOpen
  \bibfield  {author} {\bibinfo {author} {\bibfnamefont {G.}~\bibnamefont {Verdon}}, \bibinfo {author} {\bibfnamefont {J.}~\bibnamefont {Marks}}, \bibinfo {author} {\bibfnamefont {S.}~\bibnamefont {Nanda}}, \bibinfo {author} {\bibfnamefont {S.}~\bibnamefont {Leichenauer}},\ and\ \bibinfo {author} {\bibfnamefont {J.}~\bibnamefont {Hidary}},\ }\bibfield  {title} {\bibinfo {title} {Quantum {{Hamiltonian-Based Models}} and the {{Variational Quantum Thermalizer Algorithm}}},\ }\bibfield  {journal} {\bibinfo  {journal} {arXiv:1910.02071}\ }\href {https://doi.org/10.48550/arXiv.1910.02071} {10.48550/arXiv.1910.02071} (\bibinfo {year} {2019}),\ \Eprint {https://arxiv.org/abs/1910.02071} {arxiv:1910.02071} \BibitemShut {NoStop}%
\bibitem [{\citenamefont {Amari}(2016)}]{amari2016information}%
  \BibitemOpen
  \bibfield  {author} {\bibinfo {author} {\bibfnamefont {S.-i.}\ \bibnamefont {Amari}},\ }\bibfield  {title} {\bibinfo {title} {Information geometry and its applications},\ }\href@noop {} {\bibfield  {journal} {\bibinfo  {journal} {Springer}\ }\textbf {\bibinfo {volume} {194}} (\bibinfo {year} {2016})}\BibitemShut {NoStop}%
\bibitem [{\citenamefont {Fukumizu}\ \emph {et~al.}(2004)\citenamefont {Fukumizu}, \citenamefont {Kuriki}, \citenamefont {Takeuchi},\ and\ \citenamefont {Akahira}}]{fukumizu2004statistical}%
  \BibitemOpen
  \bibfield  {author} {\bibinfo {author} {\bibfnamefont {K.}~\bibnamefont {Fukumizu}}, \bibinfo {author} {\bibfnamefont {S.}~\bibnamefont {Kuriki}}, \bibinfo {author} {\bibfnamefont {K.}~\bibnamefont {Takeuchi}},\ and\ \bibinfo {author} {\bibfnamefont {M.}~\bibnamefont {Akahira}},\ }\bibfield  {title} {\bibinfo {title} {Statistical theory of singular models},\ }\href@noop {} {\bibfield  {journal} {\bibinfo  {journal} {Iwanami, Tokyo}\ } (\bibinfo {year} {2004})}\BibitemShut {NoStop}%
\bibitem [{\citenamefont {Suzuki}(2023)}]{suzuki2023WBayes}%
  \BibitemOpen
  \bibfield  {author} {\bibinfo {author} {\bibfnamefont {J.}~\bibnamefont {Suzuki}},\ }\bibfield  {title} {\bibinfo {title} {{WAIC} and {WBIC} with python stan},\ }\href@noop {} {\bibfield  {journal} {\bibinfo  {journal} {Springer}\ } (\bibinfo {year} {2023})}\BibitemShut {NoStop}%
\bibitem [{\citenamefont {Amari}\ and\ \citenamefont {Nagaoka}(2000)}]{amari2000methods}%
  \BibitemOpen
  \bibfield  {author} {\bibinfo {author} {\bibfnamefont {S.-i.}\ \bibnamefont {Amari}}\ and\ \bibinfo {author} {\bibfnamefont {H.}~\bibnamefont {Nagaoka}},\ }\bibfield  {title} {\bibinfo {title} {Methods of information geometry},\ }\href@noop {} {\bibfield  {journal} {\bibinfo  {journal} {American Mathematical Soc.}\ }\textbf {\bibinfo {volume} {191}} (\bibinfo {year} {2000})}\BibitemShut {NoStop}%
\bibitem [{\citenamefont {Savage}(1972)}]{savage1972foundations}%
  \BibitemOpen
  \bibfield  {author} {\bibinfo {author} {\bibfnamefont {L.~J.}\ \bibnamefont {Savage}},\ }\href@noop {} {\emph {\bibinfo {title} {The foundations of statistics}}}\ (\bibinfo  {publisher} {Courier Corporation},\ \bibinfo {year} {1972})\BibitemShut {NoStop}%
\bibitem [{\citenamefont {Watanabe}(2022)}]{watanabe2022allmodels}%
  \BibitemOpen
  \bibfield  {author} {\bibinfo {author} {\bibfnamefont {S.}~\bibnamefont {Watanabe}},\ }\bibfield  {title} {\bibinfo {title} {Chapter 9 - mathematical theory of bayesian statistics where all models are wrong},\ }in\ \href {https://doi.org/https://doi.org/10.1016/bs.host.2022.06.001} {\emph {\bibinfo {booktitle} {Advancements in Bayesian Methods and Implementation}}},\ \bibinfo {series} {Handbook of Statistics}, Vol.~\bibinfo {volume} {47},\ \bibinfo {editor} {edited by\ \bibinfo {editor} {\bibfnamefont {A.~S.}\ \bibnamefont {{Srinivasa Rao}}}, \bibinfo {editor} {\bibfnamefont {G.~A.}\ \bibnamefont {Young}},\ and\ \bibinfo {editor} {\bibfnamefont {C.}~\bibnamefont {Rao}}}\ (\bibinfo  {publisher} {Elsevier},\ \bibinfo {year} {2022})\ pp.\ \bibinfo {pages} {209--238}\BibitemShut {NoStop}%
\bibitem [{\citenamefont {Box}(1976)}]{box1976science}%
  \BibitemOpen
  \bibfield  {author} {\bibinfo {author} {\bibfnamefont {G.~E.}\ \bibnamefont {Box}},\ }\bibfield  {title} {\bibinfo {title} {Science and statistics},\ }\href@noop {} {\bibfield  {journal} {\bibinfo  {journal} {Journal of the American Statistical Association}\ }\textbf {\bibinfo {volume} {71}},\ \bibinfo {pages} {791} (\bibinfo {year} {1976})}\BibitemShut {NoStop}%
\bibitem [{\citenamefont {Boyle}(1967)}]{boyle1967new}%
  \BibitemOpen
  \bibfield  {author} {\bibinfo {author} {\bibfnamefont {R.}~\bibnamefont {Boyle}},\ }\bibfield  {title} {\bibinfo {title} {New experiments physico-mechanicall, touching the spring of the air, and its effect},\ }\href@noop {} {\bibfield  {journal} {\bibinfo  {journal} {H. Hall}\ } (\bibinfo {year} {1967})}\BibitemShut {NoStop}%
\bibitem [{\citenamefont {Von~Neumann}(1947)}]{von1947mathematician}%
  \BibitemOpen
  \bibfield  {author} {\bibinfo {author} {\bibfnamefont {J.}~\bibnamefont {Von~Neumann}},\ }\bibfield  {title} {\bibinfo {title} {The mathematician},\ }\href@noop {} {\bibfield  {journal} {\bibinfo  {journal} {The works of the mind}\ }\textbf {\bibinfo {volume} {1}},\ \bibinfo {pages} {180} (\bibinfo {year} {1947})}\BibitemShut {NoStop}%
\bibitem [{\citenamefont {Bernstein}\ and\ \citenamefont {Gelfand}(1969)}]{bernstein1969meromorphic}%
  \BibitemOpen
  \bibfield  {author} {\bibinfo {author} {\bibfnamefont {I.}~\bibnamefont {Bernstein}}\ and\ \bibinfo {author} {\bibfnamefont {S.}~\bibnamefont {Gelfand}},\ }\bibfield  {title} {\bibinfo {title} {Meromorphic property of the functions p},\ }\href@noop {} {\bibfield  {journal} {\bibinfo  {journal} {Funct. Anal. Appl}\ }\textbf {\bibinfo {volume} {3}},\ \bibinfo {pages} {68} (\bibinfo {year} {1969})}\BibitemShut {NoStop}%
\bibitem [{\citenamefont {Kashiwara}(1976)}]{kashiwara1976b}%
  \BibitemOpen
  \bibfield  {author} {\bibinfo {author} {\bibfnamefont {M.}~\bibnamefont {Kashiwara}},\ }\bibfield  {title} {\bibinfo {title} {B-functions and holonomic systems},\ }\href@noop {} {\bibfield  {journal} {\bibinfo  {journal} {Inventiones mathematicae}\ }\textbf {\bibinfo {volume} {38}},\ \bibinfo {pages} {33} (\bibinfo {year} {1976})}\BibitemShut {NoStop}%
\bibitem [{\citenamefont {Bernstein}(1972)}]{bernstein1972analytic}%
  \BibitemOpen
  \bibfield  {author} {\bibinfo {author} {\bibfnamefont {J.}~\bibnamefont {Bernstein}},\ }\bibfield  {title} {\bibinfo {title} {The analytic continuation of generalized functions with respect to a parameter},\ }\href@noop {} {\bibfield  {journal} {\bibinfo  {journal} {Funktsional'nyi Analiz i ego Prilozheniya}\ }\textbf {\bibinfo {volume} {6}},\ \bibinfo {pages} {26} (\bibinfo {year} {1972})}\BibitemShut {NoStop}%
\bibitem [{\citenamefont {Sato}\ and\ \citenamefont {Shintani}(1974)}]{sato1974zeta}%
  \BibitemOpen
  \bibfield  {author} {\bibinfo {author} {\bibfnamefont {M.}~\bibnamefont {Sato}}\ and\ \bibinfo {author} {\bibfnamefont {T.}~\bibnamefont {Shintani}},\ }\bibfield  {title} {\bibinfo {title} {On zeta functions associated with prehomogeneous vector spaces},\ }\href@noop {} {\bibfield  {journal} {\bibinfo  {journal} {Annals of Mathematics}\ }\textbf {\bibinfo {volume} {100}},\ \bibinfo {pages} {131} (\bibinfo {year} {1974})}\BibitemShut {NoStop}%
\bibitem [{\citenamefont {Hayashi}(2002)}]{hayashi2002Two}%
  \BibitemOpen
  \bibfield  {author} {\bibinfo {author} {\bibfnamefont {M.}~\bibnamefont {Hayashi}},\ }\bibfield  {title} {\bibinfo {title} {Two quantum analogues of {{Fisher}} information from a large deviation viewpoint of quantum estimation},\ }\href {https://doi.org/10.1088/0305-4470/35/36/302} {\bibfield  {journal} {\bibinfo  {journal} {Journal of Physics A: Mathematical and General}\ }\textbf {\bibinfo {volume} {35}},\ \bibinfo {pages} {7689} (\bibinfo {year} {2002})}\BibitemShut {NoStop}%
\bibitem [{\citenamefont {Haber}(2023)}]{haber2023Notes}%
  \BibitemOpen
  \bibfield  {author} {\bibinfo {author} {\bibfnamefont {H.~E.}\ \bibnamefont {Haber}},\ }\href@noop {} {\bibinfo {title} {Notes on the {{Matrix Exponential}} and {{Logarithm}}}} (\bibinfo {year} {2023})\BibitemShut {NoStop}%
\bibitem [{\citenamefont {Leorato}(2017)}]{leorato2017Note}%
  \BibitemOpen
  \bibfield  {author} {\bibinfo {author} {\bibfnamefont {S.}~\bibnamefont {Leorato}},\ }\bibfield  {title} {\bibinfo {title} {A note on {{H{\"o}lder}}'s inequality for matrix-valued measures},\ }\href {https://doi.org/10.7153/mia-2017-20-76} {\bibfield  {journal} {\bibinfo  {journal} {Mathematical Inequalities \& Applications}\ ,\ \bibinfo {pages} {1183}} (\bibinfo {year} {2017})}\BibitemShut {NoStop}%
\bibitem [{\citenamefont {Iba}\ and\ \citenamefont {Yano}(2023)}]{iba2023Posterior}%
  \BibitemOpen
  \bibfield  {author} {\bibinfo {author} {\bibfnamefont {Y.}~\bibnamefont {Iba}}\ and\ \bibinfo {author} {\bibfnamefont {K.}~\bibnamefont {Yano}},\ }\bibfield  {title} {\bibinfo {title} {Posterior {{Covariance Information Criterion}} for {{Weighted Inference}}},\ }\href {https://doi.org/10.1162/neco_a_01592} {\bibfield  {journal} {\bibinfo  {journal} {Neural Computation}\ }\textbf {\bibinfo {volume} {35}},\ \bibinfo {pages} {1340} (\bibinfo {year} {2023})}\BibitemShut {NoStop}%
\bibitem [{\citenamefont {Iba}\ and\ \citenamefont {Yano}(2022)}]{iba2022Posterior}%
  \BibitemOpen
  \bibfield  {author} {\bibinfo {author} {\bibfnamefont {Y.}~\bibnamefont {Iba}}\ and\ \bibinfo {author} {\bibfnamefont {K.}~\bibnamefont {Yano}},\ }\href {https://doi.org/10.48550/arXiv.2206.05887} {\bibinfo {title} {Posterior covariance information criterion for arbitrary loss functions}} (\bibinfo {year} {2022}),\ \Eprint {https://arxiv.org/abs/2206.05887} {arXiv:2206.05887} \BibitemShut {NoStop}%
\end{thebibliography}%

\newpage
\numberwithin{equation}{section}
\appendix
\section*{Glossary}
\label{app:glossary}

\renewcommand{\descriptionlabel}[1]{\hspace{\labelsep}{#1}}
\begin{description}[labelwidth=1.5cm, leftmargin=1.7cm]
    \item[$a(x,u)$] Factor of the standard form of $f(x,g(u))$ (Definition \ref{defn:invariants in singular learning theory}).
    \item[$a^Q(x,u)$] Factor of the standard form of 
        $f^Q(\hat{\rho}_x,g(u))$ (Eq. \eqref{eq:intro_a^Q}).
        \item[\rm{AIC}] Akaike Information Criterion (Definition \ref{defn:aic and waic})
        \[\mathrm{AIC} \coloneqq \frac{1}{n} \sum_{i=1}^{n} \{ - \log p(X_i|\hat{\theta}(X^n)) \} + \frac{d}{n}.\] 
    \item[$C^Q$]
    (Definition \ref{defn:quantum analog of invariants})
        \[C^Q(\xi_n) \coloneqq \mathbb{E}_X\left[ \mathrm{Cov}_{\theta}\left[ \sqrt{t}a(X,u), \sqrt{t^Q}a^Q(\hat{\rho}_X,u) \right]\right].\]
    \item[$C_n^Q$] Posterior covariance of a classical log-likelihood and its quantum analog with a classical snapshot \[C_n^Q\coloneqq \frac{1}{n} \sum_{i=1}^{n} \mathrm{Cov}_{\theta}\left[ \log p (X_i|\theta), \Tr(\hat{\rho}_{X_i} \log \sigma(\theta))\right].\]
    \item[$\mathrm{Cov}_{\theta}$] Covariance with respect to the posterior distribution, i.e., posterior covariance.
    \item[$D$] Dimension of a Hilbert space (proofs in the higher order scalings; Lemmas \ref{lem:higher_order_scaling_regular} and \ref{lem:higher_order_scaling_singular}).
    \item[$d$] Dimension of the parameter set $\Theta$.
    \item[$\mathbb{E}_X$] Expectation by $q(X)$.
    \item[$\mathbb{E}_{X^n}$] Expectation over the sets of $n$ i.i.d. training samples by $q(X^n) = \prod_{i=1}^n q(X_i)$. 
    \item[$\mathbb{E}_\theta$] Posterior mean (Definition \ref{defn:posterior mean for classical})
  \[\mathbb{E}_\theta[s(\theta)] \coloneqq \int_\Theta s(\theta) \, p(\theta|X^n) d\theta,\]
  or generalized posterior mean for matrices (Eq. \eqref{eq:quantum posterior mean})    \[\mathbb{E}_\theta[\log \sigma(\theta)] \coloneqq \int_{\Theta} \log \sigma(\theta) p(\theta|X^n) d\theta.\]
  \item[$F(\theta, \theta_0^Q)$] (Eq. \eqref{eq:large F})
  \[    F(\theta, \theta_0^Q) \coloneqq \log \sigma(\theta) - \log \sigma(\theta_0^Q).\]
        \item[$f(x,\theta)$] Log-likelihood ratio function (Appendix \ref{subsec:Empirical process with log-likelihood})
    \[f(x,\theta) \coloneqq \log \frac{p(x|\theta_0)}{p(x|\theta)}.\]
        \item[$f^Q(\hat{\rho}_x,\theta)$] Quantum log-likelihood ratio function (Appendix \ref{defn:matrix I and J_quantum})
    \[f^Q(\hat{\rho}_x,\theta) \coloneqq \Tr(\hat{\rho}_x \{ \log \sigma(\theta_0^Q) - \log \sigma(\theta) \}).\]
    \item[$G_n$] Generalization loss (Eq. \eqref{eq:def of Gn and Tn})
    \[G_n \coloneqq - \mathbb{E}_X[ \log p(X|X^n) ].\]
    \item[$\mathcal{G}_n(\alpha)$] Cumulant generating function of the generalization loss (Definition \eqref{defn:classical_cumulant})
    \[\mathcal{G}_n(\alpha) \coloneqq \mathbb{E}_X[\log \mathbb{E}_\theta[p(X|\theta)^{\alpha}]].\]
    \item[$G_n^Q$] Quantum generalization loss (Eq. \eqref{eq:G_n^Q and T_n^Q})
    \[    G_n^{Q} \coloneqq - \Tr(\rho \log \sigma_B).\]
        \item[$g$] Log resolution of the pair $(\Theta,\Theta_0)$ (Theorem \ref{thm:resolution_Watanabe} and Eq. \eqref{eq:resolution of K and K^Q})
        \[g:\widetilde{\Theta} \to \Theta,\quad g'(u)=b(u)u_1^{h_1}\cdots u_d^{h_d}.\]
        \item[$I$] Classical Fisher information matrix (Definition \ref{defn:matrix I and J})
    \[
        I(\theta) \coloneqq \mathbb{E}_X\left[ \left(\frac{\partial \log p(X|\theta)}{\partial \theta}\right) \left(\frac{\partial \log p(X|\theta)}{\partial \theta}\right)^T  \right],\quad I \coloneqq I(\theta_0).\]
        \item[$I^Q$] (Definition \ref{defn:matrix I and J_quantum})  \[I^Q(\theta) \coloneqq \Tr\left( \rho \left(\frac{\partial \log \sigma(\theta)}{\partial \theta}\right) \left(\frac{\partial \log \sigma(\theta)}{\partial \theta}\right)^T \right), \quad I^Q \coloneqq I^Q(\theta_0).\]
        \item[$J$]   Hessian of the KL divergence (Definition \ref{defn:matrix I and J}) 
    \[        J(\theta) \coloneqq \mathbb{E}_X \left[ - \frac{\partial^2 \log p(X|\theta)}{\partial \theta^2} \right], \quad J \coloneqq J(\theta_0).\]
        \item[$J^Q$] 
        Hessian of the quantum relative entropy (Definition \ref{defn:matrix I and J_quantum})
    \[J^Q(\theta) \coloneqq -  \Tr( \rho \frac{\partial^2 \log \sigma(\theta)}{\partial \theta^2}), \quad J^Q = J^Q(\theta_0^Q). \]
    \item[$K$] Average log loss function (KL divergence) (Eq. \eqref{eq:average log loss function})
    \[K(\theta)\coloneqq \mathrm{KL}(p(x|\theta_0) || p(x|\theta)).\]
    \item[$K^Q$] Average quantum log loss function (quantum relative entropy) (Eq. \eqref{eq:average quantum log loss function})
\[K^Q(\theta) \coloneqq D(\sigma(\theta_0^Q) || \sigma(\theta)).\]
\item[$k, k^Q$] Simple normal crossing representations (Eq. \eqref{eq:resolution of K and K^Q})
  \[        K(g(u))=u_1^{2k_1}\cdots u_d^{2k_d}=:u^{2k},\quad K^Q(g(u))=r(u)u_1^{2k^Q_1}\cdots u_d^{2k^Q_d}=:r(u)u^{2k^Q}.\]
\item[$n$] Number of data samples.
\item[$p(x|\theta)$] Parametrized model or $p(x|\theta) = \Tr(\Pi_x\sigma(\theta))$.
\item[\rm{QWAIC}] Quantum Widely Applicable Information Criterion (Definition \ref{def:QWAIC})
\[    \mathrm{QWAIC} \coloneqq T_n^{Q} + \frac{1}{n} \sum_{i=1}^{n} \mathrm{Cov}_{\theta}\left[ \log p (X_i|\theta), \Tr(\hat{\rho}_{X_i} \log \sigma(\theta)) \right].\]
   \item[$R_1^Q, R_2^Q$]  (Corollary \ref{cor:expansion of GnQ and TnQ for regular cases})
    \[ R_1^Q = \frac{n}{2} \Delta_n^T J^Q \Delta_n,\quad R_2^Q = \frac{n}{2}(\Delta_n^{QT} J^Q \Delta_n + \Delta_n^T J^Q \Delta_n^Q).\] 
\item[$s^Q$] Cumulant generating function of the quantum log-likelihood (Definition \ref{def:quantum cumulant})
\[               s^Q(\hat{\rho}, \alpha) \coloneqq \Tr(\hat{\rho} \log \Phi(\alpha)). \]
\item[$s^{Q(0)}$] Cumulant generating function of the quantum log-likelihood ratio  (Definition \ref{def:quantum cumulant})
        \[s^{Q(0)}(\hat{\rho}, \alpha) \coloneqq \Tr(\hat{\rho} \log \Phi^{(0)}(\alpha)).\]
\item[$T_n$] Training loss (Eq. \eqref{eq:def of Gn and Tn})
\[T_n \coloneqq - \frac{1}{n} \sum_{i=1}^{n} \log p(X_i|X^n).\]
    \item[$\mathcal{T}_n(\alpha)$] Cumulant generating functions of the training losse (Definition \eqref{defn:classical_cumulant})
    \[\mathcal{T}_n(\alpha) \coloneqq \frac{1}{n}\sum_{i=1}^n\log \mathbb{E}_\theta[p(X_i|\theta)^{\alpha}.\]
\item[$T_n^Q$] Quantum training loss (Eq. \eqref{eq:G_n^Q and T_n^Q})
\[T_n^{Q} \coloneqq - \frac{1}{n} \sum_{i=1}^{n} \Tr(\hat{\rho}_{X_i} \log \sigma_B).\]
\item[$t$] State density parameter (Definition \ref{defn:invariants in singular learning theory})
\[t\coloneqq d\cdot u^{2k} \in \mathbb{R}.\]
\item[$u$] Local parameter of $\widetilde{\Theta}$ (Theorem \ref{thm:resolution_Watanabe} and Eq. \eqref{eq:resolution of K and K^Q}).
\item[$q(x)$] True distribution or $q(x) = 
 \Tr(\Pi_x\rho)$.
       \item[$V(\xi)$] Functional variance (Definition \ref{defn:invariants in singular learning theory})
       \[V(\xi) \coloneqq \mathbb{E}_X[\mathbb{V}_{\theta}[\sqrt{t}a(X,u)]].\]
    \item[$\mathbb{V}_{\theta}$] Posterior variance (Definition \ref{defn:posterior mean for classical})
    \begin{align}
    \mathbb{V}_{\theta}[s(\theta)] &\coloneqq \int_\Theta s(\theta)^2 p(\theta|X^n) d\theta - \left( \int_\Theta s(\theta) p(\theta|X^n) d\theta \right)^2
    \end{align}
    or generalized posterior variance for matrices (Eq. \eqref{eq:quantum posterior variance})
    \[\mathbb{V}_{\theta}[\log \sigma(\theta)] \coloneqq \mathbb{E}_\theta[(\log \sigma(\theta))^2] - \mathbb{E}_\theta[\log \sigma(\theta)]^2.\]
        \item[\rm{WAIC}] Widely Applicable Information Criterion (Definition \ref{defn:aic and waic})
    \[    \mathrm{WAIC} \coloneqq \frac{1}{n} \sum_{i=1}^{n} \{ - \log \mathbb{E}_\theta[p(X_i|\theta)]\} + \frac{1}{n} \sum_{i=1}^{n} \{ \mathbb{E}_\theta[(\log p(X_i|\theta))^2] - \mathbb{E}_\theta[\log p(X_i|\theta)]^2 \}.\]
\item[$X_{\alpha}, x_{\alpha}$] Data samples.    
\item[$\alpha$] Variable in the cumulant generating functions (Definition \ref{def:quantum cumulant}).
\item[$\Delta_n$]  (Eq. \eqref{eq:Delta_n})
\[    \Delta_n \coloneqq \frac{1}{\sqrt{n}} J^{-1} \nabla \eta_n(\theta_0). \] 
\item[$\Delta_n^Q$] (Eq. \eqref{eq:def_Delta_n^Q}) \[ \Delta_n^Q \coloneqq \frac{1}{\sqrt{n}} {J^Q}^{-1} \nabla \eta_n^Q(\theta_0^Q).\]
\item[$\eta_n$] Renormalized empirical process for regular cases (Eq. \eqref{eq:eta_n})
    \[\eta_n(\theta) \coloneqq \frac{1}{\sqrt{n}} \sum_{i=1}^{n} \left\{ \mathbb{E}_X\left[ f(X,\theta) \right] - f(X_i,\theta) \right\}.\]
    \item[$\eta_n^Q$] Renormalized empirical process for regular cases (Eq. \eqref{eq:def_eta_n^Q})
\[ \eta_n^Q(\theta) \coloneqq \frac{1}{\sqrt{n}}\sum_{i=1}^{n} \left\{ \mathbb{E}_X\left[ f^Q(\hat{\rho}_X,\theta) \right] - f^Q(\hat{\rho}_{X_i},\theta) \right\}.\]
\item[$\Theta$] Set of parameters.
\item[$\Theta_0, \Theta_0^Q$] Set of optimal parameters (Eq. \eqref{eq:Theta_0})
\[\Theta_0\coloneqq \{\theta\in\Theta\mid K(\theta) = 0\},\quad \Theta_0^Q\coloneqq \{\theta\in\Theta\mid K^Q(\theta) = 0\}.\]
\item[$\theta_0, \theta_0^Q$] Fixed optimal parameter (Eqs. \eqref{eq:fixed_optimal_para} and \eqref{eq:fixed_optimal_para2})
    \[\theta_0 = \theta_0^Q \in \Theta_0 \cap \Theta_0^Q.\]
    \item[$\kappa$] Vanishing order, that is an exponent of $u$ (Theorem \ref{thm:resolution_Watanabe}).
    \item[$\lambda$]  Learning coefficient or RLCT (Definition \ref{defn:invariants in singular learning theory})
       \begin{align}
           \lambda &\coloneqq \min_i\left\{\frac{2k_i + 1}{h_i} \right\}\\
           & =\frac{d}{2}\quad \mathrm{(if\ Definition\ \ref{def:regular for classical}\ holds}).
       \end{align}
       \item[$\lambda^Q$] (Corollary \ref{cor:expansion of GnQ and TnQ for regular cases}) 
       \[ \lambda^Q \coloneqq \frac{1}{2} \Tr(J^Q J^{-1}).\]
       \item[$\nu$] Singular fluctuation (Definition \ref{defn:invariants in singular learning theory})
       \begin{align}
           \nu &\coloneqq \frac{1}{2}\mathbb{E}_{\xi}[V(\xi)]\\
           &= \frac{1}{2}\Tr(IJ^{-1})\quad \mathrm{(if\ Definition\ \ref{def:regular for classical}\ holds)}.
       \end{align}
       \item[$\nu^Q$] (Corollary \ref{cor:expansion of GnQ and TnQ for regular cases} and Theorem \ref{thm:q_expansion formulas for the expectations for singular cases})
      \begin{align}
           \nu^Q &\coloneqq \frac{n}{2} \mathbb{E}_{X^n} \left[ \Tr(\rho \mathbb{V}_{\theta}[\log \sigma(\theta)]) \right] \\
            &= \frac{1}{2} \Tr(I^Q J^{-1})\quad (\mathrm{if\ Assumptions\ \ref{ass:R1}\ and\ \ref{ass:R2}\ hold}).
       \end{align}
       \item[$\nu'^Q$] (Proposition \ref{thm:q_expectations for regular cases})
       \[\nu'^Q = \frac{1}{2} \Tr(J^Q J^{-1} I J^{-1}).\]
    \item[$\xi_n$] Renormalized empirical process for singular cases (Eq. \eqref{eq:renormalized empirical process})
    \[  \xi_n(u) \coloneqq \frac{1}{\sqrt{n}} \sum_{i=1}^{n} \{ u^k - a(X_i,u) \}.\]
    \item[$\xi_n^Q$] Renormalized empirical process for singular cases (Eq. \eqref{eq:xi_n^Q})
\[\xi_n^Q \coloneqq \frac{1}{\sqrt{n}} \sum_{i=1}^{n} \left\{ u^{k^Q} - a^Q(X_i, u) \right\}.\]
\item[$\Pi$] Tomographic complete measurement.
\item[$\rho$] True quantum state.
\item[$\hat{\rho}$] Classical shadow 
\[\hat{\rho}^n \coloneqq \{ \hat{\rho}_{x_1}, ..., \hat{\rho}_{x_n} \}\] 
where $\hat{\rho}_{x_\alpha}$ is a classical snapshot, corresponding to a measurement outcome $x_\alpha$.
\item[$\sigma(\theta)$] Parametrized quantum model.
    \item[$\sigma_B$] Bayesian mean 
    \[\sigma_B\coloneqq \int_\Theta \sigma(\theta) p(\theta|X^n) d\theta.\]
\item[$\Phi, \Phi^{(0)}$]  (Definition \ref{def:quantum cumulant})
\begin{align}
    \Phi(\alpha) &\coloneqq \int_\Theta \sigma(\theta)^\alpha p(\theta|X^n) d\theta \\
        \Phi^{(0)}(\alpha) &\coloneqq \int_\Theta \left( \int_{0}^{1} \sigma(\theta_0^Q)^{-(1-r)\alpha} \sigma(\theta)^\alpha \sigma(\theta_0^Q)^{-r\alpha} dr \right) p(\theta|X^n) d\theta.
\end{align}        

\item[$\chi^Q$] (Theorems \ref{thm:q_expectations for regular cases} and \ref{thm:q_expansion formulas for the expectations for singular cases})
\begin{align}
    \chi^Q &\coloneqq \frac{1}{2} \mathbb{E}_{X^n}\left[ \mathbb{E}_\theta\left[ \sqrt{t^Q} \xi_n^Q \right] \right]\\
&=\frac{n}{4}(\mathbb{E}_{X^n}[\Delta_n^{QT} J^Q \Delta_n] + \mathbb{E}_{X^n}[\Delta_n^T J^Q \Delta_n^Q])\quad (\mathrm{if\ Assumptions\ \ref{ass:R1}\ and\ \ref{ass:R2}\ hold}).
\end{align}
\end{description}

\section{Assumptions and Fundamental conditions}
\label{app:Fundamental conditions}
In classical learning theory, regularity conditions play a crucial role in establishing theoretical results. However, singular learning theory, as developed by Watanabe, circumvents these regularity conditions. To align with this framework, we introduce the necessary assumptions that underpin our work throughout this paper.
\subsection{Assumptions in classical learning theory}
Watanabe \cite{watanabe2009algebraic,watanabe2018mathematical} developed singular learning theory under a certain assumption, called the \textit{relatively finite variance}, which generalizes the regularity condition implicitly assumed in classical statistics. 
In this subsection, let us explain the framework within which his work was completed \cite[Chapter 3]{watanabe2018mathematical}.
\begin{defn}[{\cite[Definition 5]{watanabe2018mathematical}}]
\label{def:regular for classical}
    A pair of the probability distributions $(q(x),p(x|\theta))$ is said to be \textit{classiclly regular} if the following conditions are satisfied:
\begin{enumerate}
    \item $\Theta_0$ consists of a single element $\theta_0$,
    \item the Hessian matrix $\nabla^2 \mathrm{KL}(q\|p(\cdot|\theta_0))$ is positive definite, and 
    \item there is an open neighborhood of $\theta_0$ in $\Theta$.
\end{enumerate}
Otherwise, we call the model \textit{classically singular}.
\end{defn}

\begin{rem}
    Note that the definition of regularity varies slightly in the literature.
    It being cross-disciplinary, to the best of the authors' knowledge, the situation is as follows.
    This problem can be attributed to the difference between the estimation theory of parameters and statistical inference.
    In the former region, singular models are defined by the fact that the associated Fisher information matrix degenerates \cite[Section 12.2.6]{amari2016information}.
    This is because the realizability (Definition \ref{defn:realizability}) is implicitly assumed in such situations, which forces that the Fisher information matrix and the Hessian matrix of the KL divergence coincide at the optimal parameter.
    The latter theory, which we mainly consider in this paper, works on both the estimation and the model selection, generalizing the former.
    It means that the Fisher information matrix and the Hessian matrix of the KL divergence do not coincide in general, and hence, we have to generalize the notion of ``singular model", as we denoted.
    In this paper, we formulate our theory according to Watanabe's singular learning theory and his notion  \cite{watanabe2009algebraic,watanabe2018mathematical}.
    Finally, the definition of singular sometimes requires that $\Theta_0$ contains singularities as a manifold \cite{fukumizu2004statistical}.
    It is desirable to treat these from a unified perspective.
\end{rem}

One application of Definition \ref{def:regular for classical} is the construction of AIC.
However, as noted in \cite[Preface, Chapter 7]{watanabe2009algebraic}, this assumption is often not fulfilled in practice.
To address this issue, the following concepts were introduced as generalizations of Definition \ref{def:regular for classical}.
\begin{defn}[{\cite[Definition 6]{watanabe2018mathematical}}]
\label{def:classical homogeneous}
A probability distribution $p(x|\theta)$ is said to be \textit{homogeneous} if $p(x|\theta_0) = p(x|\theta_1)$ for  any $\theta_0,\theta_1 \in \Theta_0$.
\end{defn}
Note that the homogeneous is also referred to as \textit{essentially unique} in \cite{watanabe2018mathematical}.
This condition is used to ensure the well-definedness of the average log loss function $K(\theta)$; see also Proposition \ref{lem:homogeneous}.
However, it is insufficient for analyzing the asymptotic behavior of the generalization loss because it provides limited information about the growth of $p(x|\theta)$.
Hence, we introduce the following assumption.
\begin{defn}[{\cite[Definition 7]{watanabe2018mathematical}}]
\label{def:classical relatively finite variance}
A probability distribution $p(x|\theta)$ is said to \textit{have relatively finite variance} if there exists a constant $c>0$ so that the log-likelihood ratio function $f(x,\theta)$, defined in Definition \ref{defn:matrix I and J} (1), satisfies
\[\mathbb{E}_X[f(X, \theta)^2] \leq c \mathbb{E}_X[f(X, \theta)]\]
    holds for any $\theta \in \Theta$.
\end{defn}
This assumption is currently the weakest one necessary for developing singular learning theory. It generalizes Definition \ref{def:regular for classical} as shown in the following proposition.

\begin{prop}[{\cite[Summary]{watanabe2018mathematical}}]
\label{prop:relationship}
For a pair of probability functions $(q(x), p(x|\theta))$,  the following relationships hold:
\begin{enumerate}
    \item The classical regularity or realization condition implies the relatively finite variance.
    \item The relatively finite variance implies the homogeneous condition.
\end{enumerate}
\end{prop}
\begin{figure}[t]
    \centering
    \includegraphics[width=0.65\textwidth]{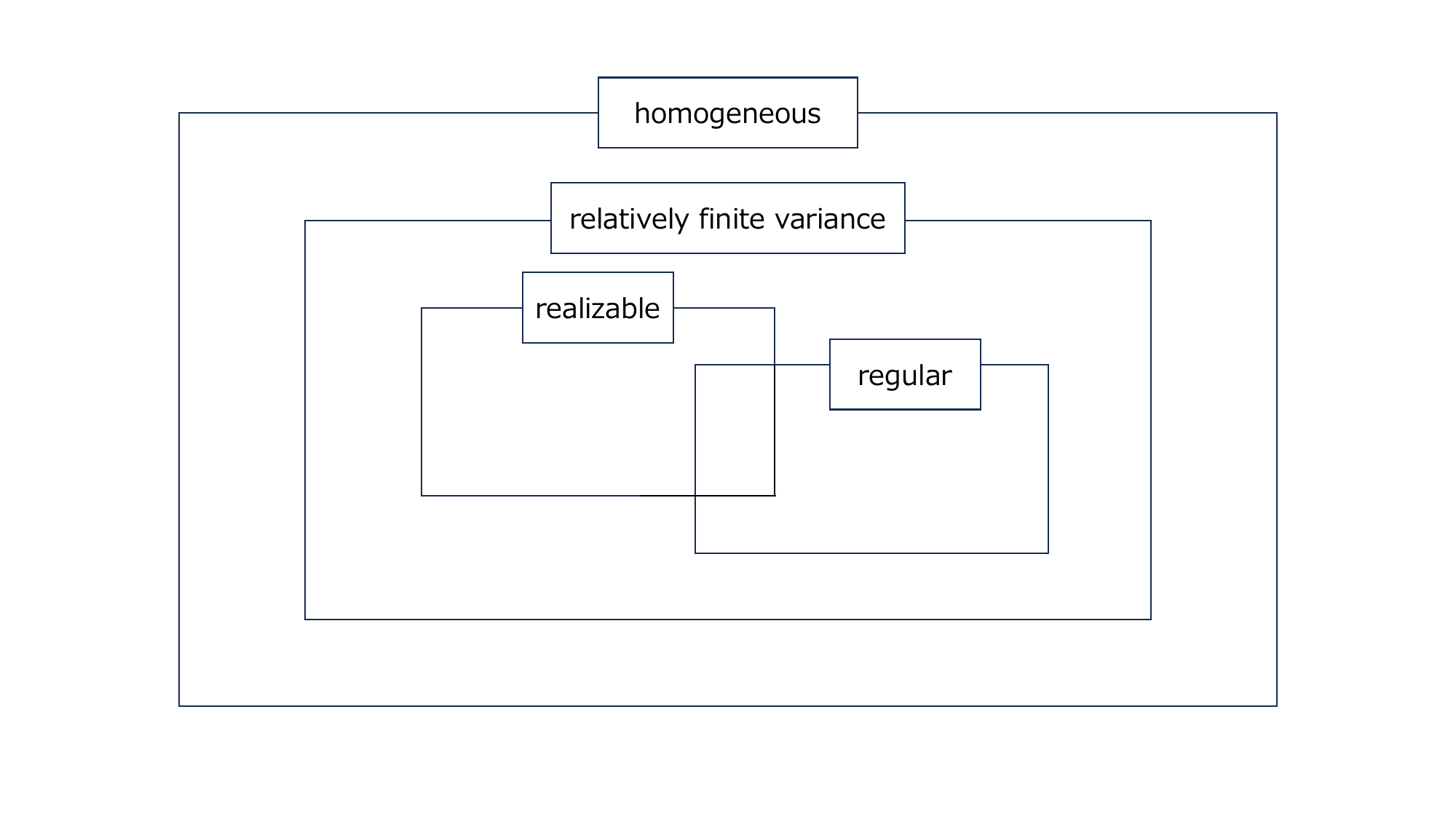}
    \caption{Relationship between the conditions of the model \cite[Figure 1.6]{suzuki2023WBayes}.}
    \label{fig:relation_1}
\end{figure}

We summarize the relationship in Figure \ref{fig:relation_1}. 
As noted above, Wayanabe's theory is developed under Definition \ref{def:classical relatively finite variance}. This assumption is essential for ensuring the consistency of desingularization and the expectation of the log-likelihood ratio functions. In particular, it allows us to obtain the standard forms of $K(\theta)$ and $f(x,\theta)$ \cite[Definition 14]{watanabe2018mathematical}. Moreover, this assumption ensures that $K(\theta)$ is well-defined; see Lemma \ref{lem:coincidence of Theta}. By drawing on this framework, we advance the theory of singular quantum state estimation. We will elaborate on this in the next subsection.

\subsection{Fundamental conditions}
\label{subsection:Fundamental conditions}
In previous work \cite{yano2023Quantuma}, the authors assumed the regularity condition for quantum models to develop their theoretical framework.
Our current study aims to generalize the Bayesian estimation framework for quantum state estimation from regular cases to singular situations. To achieve this, we do not require the full extent of the regularity condition previously assumed. Instead, it suffices to assume the regularity of the Hessian of the quantum relative entropy, which is a weaker condition than the full quantum regularity condition. For the sake of clarity and to streamline our exposition, we will adopt this weaker assumption throughout our work.

\begin{defn}
\label{def:quantum_regularity}
        A pair of the quantum states $(\rho,\sigma(\theta))$ is said to be \textit{quantum regular} if the following conditions are satisfied:
\begin{enumerate}
    \item $\Theta_0^Q$ consists of a single element $\theta_0^Q$,
    \item the Hessian matrix $\nabla^2 D(\rho\||\sigma(\theta_0^Q))$ is positive definite, and 
    \item there is an open neighborhood of $\theta_0^Q$ in $\Theta$.
\end{enumerate}
\end{defn}

To extend our analysis to singular situations, we introduce conditions analogous to those in classical learning theory. Before presenting these fundamental conditions, we first define the Banach space $L^s(q)$, which will be instrumental in our development.

\begin{defn}[{\cite[Section 5.2]{watanabe2009algebraic}}]
For $s\geq 1$, we denote by $L^s(q)$ the set consisting of all measurable functions $f$ from $\mathbb{R}^d$ to $\mathbb{C}$ satisfying
    \begin{align}
        \int |f(x)|^s q(x) dx < \infty, 
    \end{align}
    where $q(x)$ is a probability density function on $\mathbb{R}^d$.
\end{defn}

With this definition in place, we introduce the \textit{fundamental conditions}, which serve as the minimal framework required to apply singular learning theory to quantum state estimation. These conditions are analogous to \cite[Fundamental conditions I, II]{watanabe2009algebraic} and will be assumed throughout this paper.

\begin{fcwithname}{I}
\label{ass:fundamental condition}
\begin{enumerate}
    \item Quantum states $\rho$ and $\sigma(\theta)$ are full-rank.
    \item The function $f^Q(\hat{\rho},\theta)$ is an analytic function of $\theta \in \Theta$ which takes values in $L^2(q)$.
\end{enumerate}
\end{fcwithname}

\begin{fcwithname}{II}
\label{ass:fundamental conditionII}
 The set $\Theta$ is a semi-analytic compact subset of $\mathbb{R}^d$.
\end{fcwithname}

Fundamental condition \ref{ass:fundamental condition} (1) is necessary to define the quantum relative entropy between the true state $\rho$ and the parameterized model $\sigma(\theta)$.
We assume Fundamental conditions \ref{ass:fundamental condition} and \ref{ass:fundamental conditionII} throughout this paper to ensure the validity of our theoretical development.

We conclude this section by proposing a conjecture in quantum information theory, motivated by the observation that classical systems can exhibit greater the complexity of singularities than quantum systems when measurements are introduced.
\begin{conj}
    \label{conj:qregular to cregular}
    Let $(\rho, \sigma(\theta))$ be a quantum model and $(q(x), p(x|\theta))$ be the associated classical model.
    If $(q(x), p(x|\theta))$ is regular (Definition \ref{def:regular for classical}), then
 the  $(\rho, \sigma(\theta))$ is also regular (Definition \ref{def:quantum_regularity}).
\end{conj}
Intuitively, this conjecture seems reasonable because projection increases the complexity of the singularities.
We shall prove it in realizable situations (Definition \ref{defn:realizability}) in Proposition \ref{prop:proof of the first conjecture}.

\subsection{Assumptions in quantum state estimation}
In the following discussion, for a pair of quantum states $(\rho, \sigma(\theta))$, we consider the associated classical probability distributions derived from a tomographically complete measurement $\{\Pi_x\}$
\[q(x)\coloneqq\Tr(\Pi_x\rho),\quad  p(x|\theta)\coloneqq\Tr(\Pi_x\sigma(\theta)).\]
We refer to this situation by saying that the pair of classical distributions $(q(x), p(x|\theta))$ is associated with $(\rho,\sigma(\theta))$.
This association allows us to connect quantum estimation problems with classical statistical frameworks.

We now combine Definitions \ref{def:regular for classical} and \ref{def:quantum_regularity} into a single assumption for use in this paper.
\begin{asswithname}{R1}
\label{ass:R1}
       A pair of quantum states $(\rho,\sigma(\theta))$ and the associated pair $(q(x), p(x|\theta))$ satisfy Definition \ref{def:regular for classical} and Definition \ref{def:quantum_regularity} respectively.
\end{asswithname}
This assumption was implicitly made in the definition of QAIC \cite{yano2023Quantuma}. Our goal is to develop an information criterion that remains valid even in singular cases where Assumption \ref{ass:R1} does not hold. To achieve this, we propose a framework that relaxes this assumption, as described below.
For the definition of the log-likelihood ratio functions $f(x,\theta)$ and $f^Q(\hat{\rho},\theta)$, refer to Definition \ref{defn:matrix I and J} (1) and 
 Definition \ref{defn:matrix I and J_quantum} (1) respectively.
 These functions play a central role in analyzing the statistical properties of the models under consideration.
 We introduce the following assumption, which imposes a weaker condition than Assumption \ref{ass:R1} and allows us to handle singular cases.
\begin{asswithname}{S1}
\label{ass:S1}
       For a pair of quantum states $(\rho,\sigma(\theta))$, the log-likelihood ratio functions $f(x,\theta)$ and  $f^Q(\hat{\rho},\theta)$ have relatively finite variances.
 Specifically, there exists a constant $c>0$ such that, for all $\theta \in \Theta$, the following inequalities hold:
    \begin{equation}
        \mathbb{E}_{X}[f^Q(\hat{\rho}_X,\theta)^2] \leq c \mathbb{E}_{X}[f^Q(\hat{\rho}_X,\theta)],\quad     \mathbb{E}_X[f(X, \theta)^2] \leq c \mathbb{E}_X[f(X, \theta)].
    \end{equation}
\end{asswithname}
Assumption \ref{ass:S1} ensures that the variances of the log-likelihood ratio functions are controlled relative to their expectations, which is a crucial condition for extending singular learning theory to quantum models. By adopting this weaker assumption, we can analyze models that may not satisfy the regularity condition, thereby broadening the applicability of our theoretical framework.

In this paper, when we refer to$(\rho,\sigma(\theta))$ \textit{regular}, we mean that Assumption \ref{ass:R1} holds.
Now, let us introduce a condition used in regular cases regarding the compatibility of the quantum relative entropy and the KL divergence.

\begin{asswithname}{R2}
\label{ass:R2}
    The intersection $\Theta_0^Q\cap\Theta_0$ is nonempty.
In particular, if the model is regular, then this is equivalent to saying that $\Theta_0^Q = \Theta_0$, consisting of one point.
\end{asswithname}
We introduce this condition here for the first time. 
As we will see later, it is a necessary assumption for calculating the posterior distribution integral.
Assumption \ref{ass:R2} implies that we can choose an element 
\begin{align}
\label{eq:fixed_optimal_para}
    \theta_0=\theta_0^Q\in \Theta_0^Q\cap\Theta_0.
\end{align}
Hence, in the present paper, we will denote this common element as $\theta_0$.
Note that in particular, if moreover Assumption \ref{ass:R1} is satisfied, then we obtain $\Theta_0 = \Theta_0^Q = \{\theta_0\}$.

We will demonstrate how the introduced assumptions are necessary for our purposes, particularly for the well-definedness of the average log loss functions $K$ and $K^Q$. For a pair of quantum states, we define the notions of homogeneity and relatively finite variance in a similar way.
\begin{lem}
\label{lem:homogeneous}
Let us take any $\theta_1^Q\in\Theta_0^Q$ and $\theta_1\in\Theta_0$.
    \begin{enumerate}
        \item   Assumption \ref{ass:S1} implies that    
            \[\sigma(\theta_0^Q)=\sigma(\theta_1^Q),\quad p(x|\theta_0)=p(x|\theta_1).\]
        \item Assumptions \ref{ass:S1} and \ref{ass:R2} imply that $\Theta_0^Q = \Theta_0 \neq \phi$.
        \item Conversely, assuming $\Theta_0^Q = \Theta_0 \neq \phi$, the homogeneous condition for quantum models, that is $\sigma(\theta_0^Q)=\sigma(\theta_1^Q)$ is equivalent to the one for the associated pair $(q(x), p(x|\theta))$, that is $p(x|\theta_0)=p(x|\theta_1)$.
    \end{enumerate}
\end{lem}
\begin{proof}
   The latter statement of (1) follows from \cite[Lemma 3]{watanabe2018mathematical}.
   By our assumption on the supports of $\rho$ and $\sigma(\theta)$, the former part also follows in the same way.
     
   (2) Taking a measurement $\Pi_x$ on the both sides of $\sigma(\theta_0^Q)=\sigma(\theta_1^Q)$, it follows $p(x|\theta_0) = p(x|\theta_1^Q)$, which forces $\theta_1^Q\in\Theta_0$.
   Conversely, since $\{\Pi_x\}$ is a tomographically complete measurement, the equality $p(x|\theta_0) = p(x|\theta_1)$ also implies that $\sigma(\theta_0^Q) = \sigma(\theta_1)$.
   It concludes that $\theta_1 \in \Theta_0^Q$.

   Item (3) follows from a converse discussion of the proof of (2).
\end{proof}

We can verify that the proof strategy of  Proposition \ref{prop:relationship} also applies to quantum models, considering matrix computations.
\begin{lem}
\label{lem:quantum relationship}
    All claims in Proposition \ref{prop:relationship} hold for a pair of quantum models $(\rho,\sigma(\theta))$.
    Hence, the relationships described in Figure \ref{fig:relation_1} are also valid for quantum models. 
\end{lem}

Given that the above assumptions may be counterintuitive, we introduce a more commonly used, weaker assumption and demonstrate that it satisfies the necessary conditions.
\begin{defn}
\label{defn:realizability}
        For a given pair $(\sigma(\theta), \rho)$, the quantum state $\rho$ is said to be \textit{realizable} by the parametric quantum state $\{\sigma(\theta)|\theta \in \Theta\}$ if there exists $\theta \in \Theta$ with $\sigma(\theta) = \rho$.
\end{defn}
This condition is often used to analyze the Fisher information matrix and losses as a reasonable assumption \cite{watanabe2018mathematical}, which also simplifies our theory as follows.

\begin{lem}[{\cite[Appendix C]{yano2023Quantuma}}]
\label{lem:coincidence of Theta}
If $\sigma(\theta)$ realizes $\rho$, then $\Theta_0^Q = \Theta_0$.
In particular, Assumption \ref{ass:R2} holds.
\end{lem}
\begin{proof}
The realizability assumption on the quantum models implies that $q(x)$ is also realized by $p(x|\theta)$.
Hence, we can rewrite the parameter spaces as 
\begin{align*}
\Theta_0^Q&=\{\theta_0^Q\in\Theta^Q\mid D(\rho||\sigma(\theta_0^Q))=0\}\\
\Theta_0&=\{\theta_0\in\Theta\mid \mathrm{KL}(q||p(\cdot|\theta_0)=0\}.
\end{align*} 
Now, since $\{\Pi_x\}$ is a tomographically complete measurement, more strongly we have $\mathrm{KL}(\Tr(\Pi_x\rho) || \Tr(\Pi_x \sigma (\theta))) = 0$ if and only if $D(\rho||\sigma(\theta))=0$ for $\theta\in\Theta$.
It concludes the proof.
\end{proof}

Based on this observation, we can prove that certain assertions as follows.

\begin{prop}
\label{prop:assumption_relation}
    The realizability condition (Definition \ref{defn:realizability}) implies  Assumptions \ref{ass:S1} and \ref{ass:R2}.
\end{prop}
\begin{proof}
    The assertion follows from Lemmas \ref{lem:quantum relationship} and \ref{lem:coincidence of Theta}.
\end{proof}
Although the regularity condition does not appear in the main theorems of this paper, it indicates a certain class for the model we are dealing with.

\begin{prop}
\label{prop:proof of the first conjecture}
    If there is a parameter $\theta_0$ with $\rho = \sigma(\theta_0)$, then Conjecture \ref{conj:qregular to cregular} holds.
\end{prop}
\begin{proof}
    Since the realizability condition holds, Eqs. \eqref{eq:average log loss function} and \eqref{eq:average quantum log loss function} can be rewritten as
    \begin{align}
    \label{eq:proof of conj of models}
    K(\theta) = \mathrm{KL}(q(x)\|p(x|\theta)),\quad 
    K^{Q}(\theta) = D(\rho\|\sigma(\theta)).
\end{align}
By the data processing inequality in information theory, combined with Eq. \eqref{eq:proof of conj of models}, we get
\begin{align}
\label{eq:data processing inequality}
    K(\theta) \leq K^Q(\theta);
\end{align}
see also Ineq. \eqref{ineq:K less than K^Q} for the complex parameter case.
Moreover, the regularity condition for $(q(x),p(x|\theta))$ implies that the optimal parameter set consists of the unique point: $\Theta_0=\{\theta_0\}$.
It follows from Lemma \ref{lem:coincidence of Theta} that $\Theta_0^Q=\{\theta_0\}$.
Considering the Taylor expansion around $\theta_0$, Ineq. \eqref{eq:data processing inequality} shows that the vanishing order of $K^Q(\theta)$ is less than or equal to the vanishing order of $K(\theta)$ at $\theta_0$.
NHote that the latter is 2 because $K(\theta)$ can be approximated by a quadratic form around $\theta_0$ from the regularity condition for $(q(x),p(x|\theta))$.
Since $K(\theta_0) = 0$ and $K(\theta)\geq 0$, the former is a positive even integer.
Hence, the vanishing order of $K^Q(\theta)$ at $\theta_0$ is 2, which implies the regularity for $(\rho,\sigma(\theta))$.
\end{proof}

We illustrate the logical relationship between the conditions discussed so far in Figure \ref{fig:relation_2}. Here, Lemmas \ref{lem:homogeneous} and \ref{lem:coincidence of Theta} indicate that the homogeneous and realizability conditions are equivalent between quantum and classical situations, and hence we denote by ``homogeneous" and ``realizable" respectively; compare to Figure \ref{fig:relation_1}.

\begin{figure}[t]
    \centering    \includegraphics[width=0.65\textwidth]{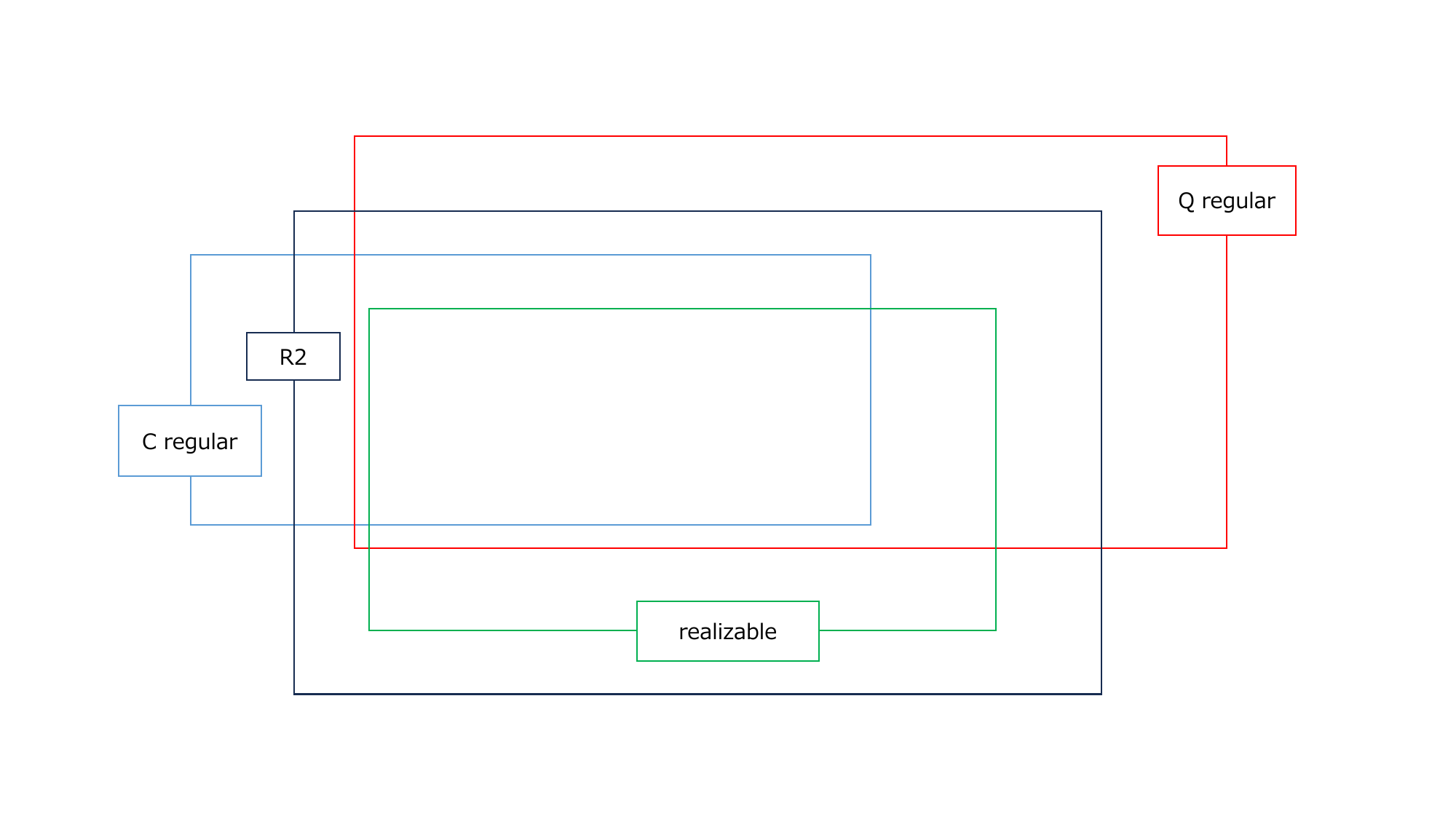}
    \caption{Relationship between the properties of quantum and classical models surrounding the realizability condition.}
    \label{fig:relation_2}
\end{figure}
We conclude this subsection with a final technical assumption, which generalizes the regularity of quantum models. As we will see in Section \ref{subsec:analytic_calculation}, reasonable examples satisfy this assumption.
\begin{asswithname}{S2}
\label{ass:S2}
In the expressions in Eq. \eqref{eq:resolution of K and K^Q}, the equality 
    $k=k^Q$
    holds.
    In particular, the learning coefficients associated with $K$ and $K^Q$ coincide.
\end{asswithname}
Intuitively, this is a generalization in a different direction from the relatively finite variance of the regularity condition. In the regular situation, $K$ (resp. $K^Q$) can be approximated around the optimal parameters in quadratic forms with the (resp. quantum) Fisher information matrix as coefficients, but in the singular situation, we get the normal crossing representation as a generalization of this. Assumption \ref{ass:S2} means the equality of the orders in this representation, or in other words, the equality of the vanishing orders of the derivatives of the average log loss functions.

We assume this condition only in certain proofs for singular cases, in addition to Fundamental conditions \ref{ass:fundamental condition} and \ref{ass:fundamental conditionII}.
Roughly speaking, it means that the behavior of the two average log loss functions $K(\theta)$ and $K^Q(\theta)$ near the zero locus are not significantly different.
\begin{rem}
\label{rem: k and k^Q}
    If we consider complex parameters for the models, then Assumption \ref{ass:S2} follows from Liouville's theorem (Proposition \ref{prop:Liouville and k=k^Q}).  We can observe that this assumption holds in reasonable circumstances (Section \ref{subsec:analytic_calculation}). However, since we are considering real parameters, the existence of $C^{\omega}$-class functions bounded on $\mathbb{R}^d$ makes the assumption nontrivial. Nonetheless, it is conceivable that Assumption \ref{ass:S1} shows it (Conjecture \ref{conj:k=k^Q}).
\end{rem}

\begin{prop}
\label{prop:ass_relation}
\begin{enumerate}
    \item Assumption \ref{ass:R1} implies Assumption \ref{ass:S1}.
    \item Assumptions \ref{ass:R1} and \ref{ass:R2} imply Assumption \ref{ass:S2}.
    \item Assumption \ref{ass:S2} implies Assumption \ref{ass:R2}.
    In particular, under Assumption \ref{ass:R1}, Assumption \ref{ass:R2} is equivalent to Assumption \ref{ass:S2}.
\end{enumerate}
\end{prop}
\begin{proof}
    Item (1) follows from Proposition \ref{prop:relationship} and Lemma \ref{lem:quantum relationship}.
    
    (2) It follows from Assumption \ref{ass:R2} that $\Theta_0$ and $\Theta_0^Q$ consist of only one points.
    Combined with Assumption \ref{ass:R2}, two optimal sets $\Theta_0$ and $\Theta_0^Q$ must coincide.
    Furthermore, the Hessian of the quantum relative entropy and KL divergence are non-degenerate, which forces $K(\theta)$ and $K^Q(\theta)$ can be approximated by quadratic forms; see also Example \ref{ex:primitive blowup example}.
    Hence, the two learning coefficients are $d/2$, and we obtain $k=k^Q$.
    
    (3)     If Assumption \ref{ass:R2} does not hold, then we can take an element $\theta\in\Theta_0^Q\setminus \Theta_0$ or $\theta\in\Theta_0\setminus \Theta_0^Q$.
    Taking a simultaneous log resolution, in the former case, we can show that $k^Q >0$ while $k=0$ around $\theta$.
    The same applies in the other case.

\end{proof}
Notably, Assumptions \ref{ass:R1} and \ref{ass:R2} imply Assumptions \ref{ass:S1} and \ref{ass:S2}.
Hence, the former is referred to as \textit{regular cases} and the latter as \textit{singular cases}, an analog of classical singular learning theory.
  Assumption \ref{ass:S2} can be seen as a coincidence of the optimal parameter set with multiplicity.
  In Figure \ref{fig:relation_3}, we present a logical relationship between R1, R2, S1, S2 assumptions.
  This can be regarded as a special case of Figure \ref{fig:relation_2}.
  \begin{figure}[t]
    \centering
    \includegraphics[width=0.65\textwidth]{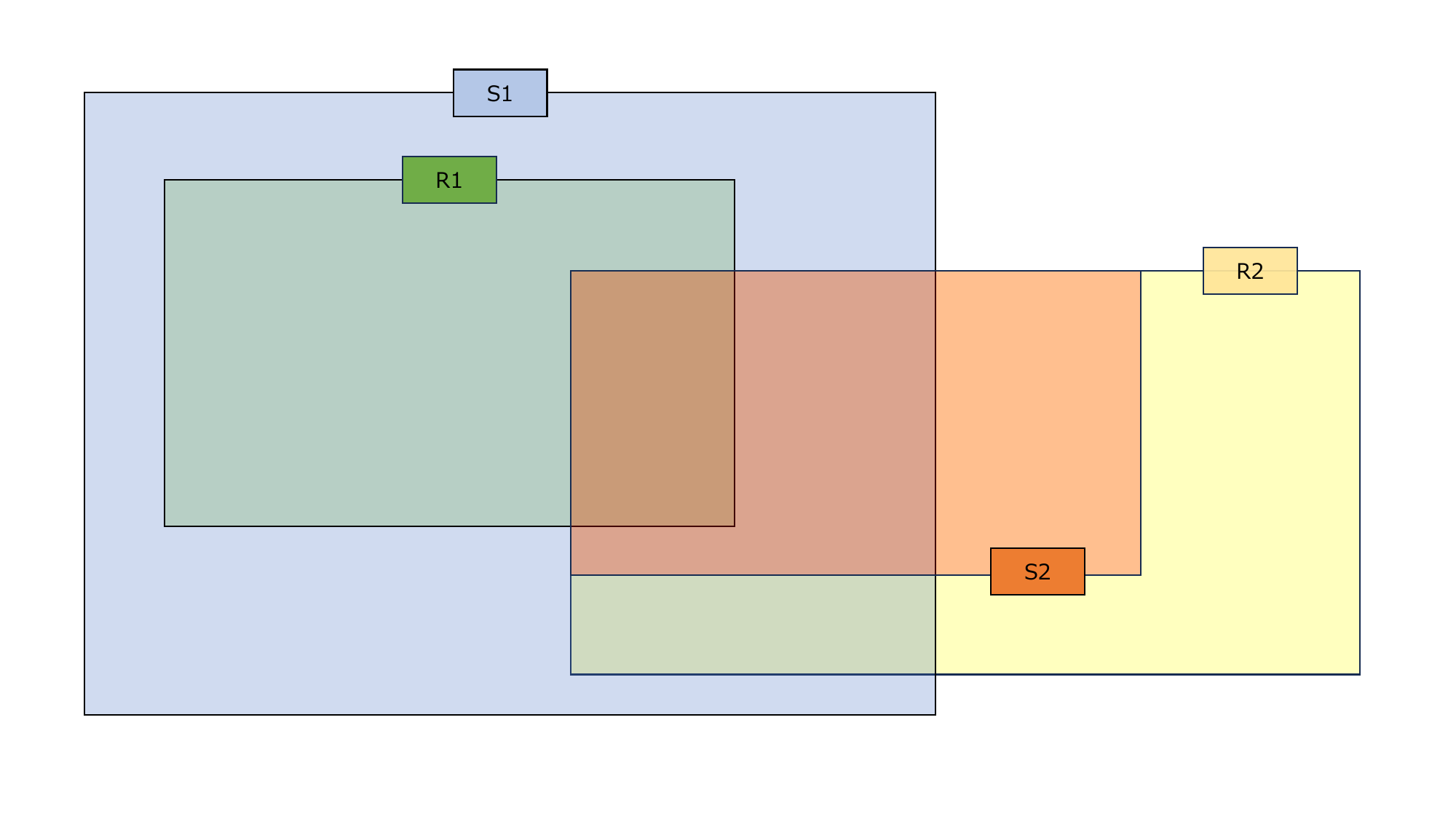}
    \caption{Relationship between the conditions of the quantum and classical models.}
    \label{fig:relation_3}
\end{figure}

From Proposition \ref{prop:ass_relation} (3), we can take and put 
\begin{align}
\label{eq:fixed_optimal_para2}
\theta_0=\theta_0^Q\in\Theta_0 \cap \Theta_0^Q
\end{align}
even in singular cases.
We will use this notation throughout the paper.

\section{Mathematical tools}
\label{app:Mathematical tools}

To consider the statistical behavior of quantum singular models, we redefine and analyze quantum state estimation using the methods of singular learning theory. Two main mathematical challenges emerge when dealing with quantum singular models.
The first one arises from the singularities associated with the analytic functions $K(\theta)$ in Eq. \eqref{eq:average log loss function} and $K^Q(\theta)$ in Eq. \eqref{eq:average quantum log loss function}. 
Approximating these functions using quadratic forms around the singularities is impossible in general, i.e., the Hessian matrix of $K(\theta)$ and $K^Q(\theta)$ would be degenerate. 
The second one is that empirical processes defined in regular situations are not properly defined on the optimal parameter set in singular cases. 
Consequently, the posterior distribution does not converge to a Gaussian distribution, 
rendering the integration with posterior distribution in Bayesian estimation of the quantum state more challenging.

To address these issues, in this section, we tackle the former problem by introducing normal crossing representations, a generalization of the quadratic approximations, through the resolution of singularities in algebraic geometry. 
Notably, we simultaneously deal with singularities arising from a quantum model $(\rho,\sigma(\theta))$ and the associated classical model $(q(x), p(x|\theta))$. 
For the latter consideration, we further introduce the theory of empirical processes. 
In singular learning theory, an empirical process of log-likelihood ratio function is redefined after the resolution of singularities to capture its behavior in the vicinity of the optimal parameter set.
Furthermore, in our case, another empirical process is also introduced to study the asymptotic behavior of the quantum training loss $T_n^Q$.
It is natural to consider the $\Theta$ as a parameter space for $\sigma(\theta)$, while this parameter space is unsuitable for analyzing singular models in Bayesian estimation. 
We explain in this paper that the space $\widetilde{\Theta}$, after being blown up, is the appropriate defining space of singular parametric models to describe the generalization loss.

\subsection{Algebraic geometry: log resolution of singularities}
\label{app:algebraic geometry}
In this subsection, we recall the basic notions of algebraic geometry, particularly the resolution of singularities and real log canonical thresholds, which are fundamental concepts in singular learning theory.
Before introducing the real log canonical thresholds, let us briefly recall the results on the log resolution of singularities, originally proved by Hironaka \cite{hironaka1964resolutionI,hironaka1964resolutionII}. 
Hironaka's groundbreaking work established that any algebraic variety over a field of characteristic zero admits a resolution of singularities.
His proof applies to both real and complex analytic spaces, as detailed in \cite{hironaka1964resolutionI} for real spaces and \cite{hironaka1964resolutionII} for complex spaces.

We present the theorem relevant to our discussion:
\begin{thm}
\label{thm:resolution_original}
    Let $X$ be a real analytic space and $X_{\mathrm{sm}}$ (resp. $X_{\mathrm{sing}}$) be its smooth (resp. singular) part.
    Then, a log resolution of singularities exists.
    In other words, there exists a proper bimeromorphic morphism $g:\widetilde{X}\to X$ so that:
    \begin{enumerate}
        \item The morphism $g$ is isomorphic over $X_{\mathrm{sm}}$.
        \item The inverse image $g^{-1}(X_{\mathrm{sing}})$ of the singular locus is a simple normal crossing divisor.
    \end{enumerate}
\end{thm}
For our purposes, we refer to an alternative version of this theorem that is more directly applicable to our context, as Atiyah wrote \cite{atiyah1970resolution}.
\begin{thm}[{\cite[Theorem 2.3]{watanabe2009algebraic}}]
\label{thm:resolution_Watanabe}
    Let $f(x)$ be a non-constant real analytic function from a neighborhood of the origin in $\mathbb{R}^d$ to $\mathbb{R}^1$, which satisfies $f(0)=0$.
    Then, there exists an open neighborhood $W\subset \mathbb{R}^d$ of the origin, $d$-dimensional manifold, and a real analytic map $g:U\to W$ satisfies the following conditions.
    \begin{enumerate}
        \item The map $g$ is proper.
        \item The map $g$ induces a real analytic isomorphism between $U\setminus U_0$ and $W\setminus W_0$
        where 
            \[W_0\coloneqq\{x\in W\mid f(x)=0\},\quad U_0\coloneqq g^{-1}(U_0).\]
        \item Around any $p\in U_0$, there is a local coordinate $u=(u_1,\cdots,u_d)$ of $U$ so that 
        \begin{align}
        f(g(u))&=Su_1^{\kappa_1}\cdots u_d^{\kappa_d},\\
        g'(u)&=b(u)u_1^{h_1}\cdots u_d^{h_d}
        \end{align}
        where $S$ is a constant, $b(u)$ is a real analytic function with $b(p)\neq 0$ and $\kappa_1,\cdots,\kappa_d, h_1,\cdots, h_d$ are nonnegative integers.
        Moreover, if $f(x)\geq 0$ for any $x$, then the integers $\kappa_1,\cdots, \kappa_d$ are even.
    \end{enumerate}
\end{thm}
For a variety $X$ and a subvariety $Y\subset X$ over $\mathbb{R}$, the \textit{real log canonical threshold (RLCT)} associated with $Y$ is defined as 
\[\mathrm{RLCT}\coloneqq\min_i\left\{\frac{h_i+1}{\kappa_i}\right\}\]
where $g:\widetilde{X}\to X$ is a proper birational morphism such that the exceptional locus is $Y$ and $g^{-1}(Y)$ is a simple normal crossing divisor.
The RLCT plays a significant role in singular learning theory as it quantifies the complexity of singularities around the optimal parameters.

Calculating resolutions explicitly can be challenging in general. However, we illustrate these concepts with a fundamental example.
\begin{ex}
\label{ex:primitive blowup example}
Let us consider the blowup of $\mathbb{A}_{\mathbb{R}}^d = \Spec \mathbb{R}[X_1,\cdots, X_d]$ at the origin $o\coloneqq (0,\cdots,0)$.
The blowup is $g:\mathrm{Bl}_o(\mathbb{A}_{\mathbb{R}}^d)\to\mathbb{A}_{\mathbb{R}}^d$ where \[\mathrm{Bl}_o(\mathbb{A}_{\mathbb{R}}^d)\coloneqq\Proj \mathbb{R}[X_1,\cdots,X_d,S_1,\cdots,S_d]/(S_iX_j - S_jX_i)_{i,j}\subset \mathbb{R}^d \times \mathbb{P}^{d-1}.\]
Let us take the affine open subset 
\begin{align}
    U_{S_i}\coloneqq (S_i\neq 0) &= \Spec\mathbb{R}[x_1,\cdots,x_d,s_1,\cdots,s_{i-1},s_{i+1},\cdots,s_d]/(x_j - x_i s_j)_j\\
    &\cong \Spec\mathbb{R}[x_i,s_1,\cdots,s_{i-1},s_{i+1},\cdots,s_d]
\end{align}
where $x_j\coloneqq X_j/S_i$ and $s_j\coloneqq S_j/S_i$.
Then, the restriction of $g$ to this affine open subset is described as
\[g\vert_{U_{S_i}}: (x_i,s_1,\cdots,s_{i-1},s_{i+1},\cdots,s_d) \mapsto (s_1x_i,\cdots, s_{i-1}x_i, x_i, s_{i+1}x_i, \cdots, s_dx_i),\]
which is isomorphism on $U_{S_i}\setminus g^{-1}(o)$.
Of course, the exceptional divisor is isomorphic to $\mathbb{P}^{d-1}$, which is, in particular, a simple normal crossing divisor.
Hence, this example realizes one concrete form of Theorem \ref{thm:resolution_Watanabe}.
Using this description, we now compute the normal form of the function $f(X_1,X_2) = X_1^2 + X_2^2$ as Theorem \ref{thm:resolution_Watanabe} (3).
On $U_{S_1}$, the function becomes $f(g(x_1,s_2)) = x_1^2(s_2^2 + 1)$.
This implies that $2k=(2,0)$ on this affine open subset.
Computing the differentials, we also have $h=(1,0)$.
A similar computation works for $U_{S_2}$, and hence $\mathrm{RLCT} = 1$.
The observation here can be generalized, asserting the coincidence of the learning coefficient $\lambda$ and the number $d/2$ in regular cases in singular learning theory; see also Remark \ref{rem:regular_RLCT}.
\end{ex}

Here, we explain how to utilize the resolution of singularities within our framework.
We now apply Theorem \ref{thm:resolution_Watanabe} to the real analytic functions $K$ and $K^Q$.
In this situation, we can use the simultaneous resolution of singularities \cite[Theorem 2.8]{watanabe2009algebraic}.
Since both the quantum relative entropy and KL divergence are non-negative functions, it produces the expressions as
\begin{align}\label{eq:resolution of K and K^Q}
    K(g(u))=u_1^{2k_1}\cdots u_d^{2k_d}=:u^{2k},\quad  K^Q(g(u))=r(u)u_1^{2k^Q_1}\cdots u_d^{2k^Q_d}=:r(u)u^{2k^Q}.
\end{align}
Here, $k_1,\cdots,k_d$ are non-negative integers and $r(u)$ be a real-analytic function with $r(0)\neq 0$.
Note that the exponents $\kappa_i$ in Theorem \ref{thm:resolution_Watanabe} (3) correspond to $2k_i$ and $ 2k_i^Q$ in our notation.
It is suggested that the function $r(u)$ reflects the ratio of the quantum and classical Fisher information matrices.
Throughout this paper, we denote by $g:\widetilde{\Theta} \to \Theta$ a log resolution associated with $\Theta_0 = \Theta_0^Q$, and $u_1,\cdots, u_d$ are local coordinates on $\widetilde{\Theta}$.
\begin{rem}
Here, we comment on the surrounding topics in algebraic geometry related to our setting in quantum state estimation.
\begin{enumerate}
    \item     As a more general situation, even when $\Theta_0^Q \neq \Theta_0$, we can also obtain Eq. \eqref{eq:resolution of K and K^Q}.
    This is a direct consequence of the simultaneous resolution of singularities.
    \item     As we introduce in Definition \ref{defn:invariants in singular learning theory}, we define the learning coefficient $\lambda$ of a model as the RLCT for $\Theta_0$.
    It is known that the smaller $\lambda$ is, the worse the singularity in algebraic geometry.
    Theorem \ref{mainthm:quantum expansion singular} tells us this situation corresponds to a smaller generalization loss.
    This suggests that the more complex singularity of a model, the better the performance. 
\end{enumerate}
\end{rem}

Now, we discuss the reasonableness of Assumption \ref{ass:S2}.
In Watanabe's original fundamental condition \cite[Definition 6.1]{watanabe2009algebraic}, he assumed the extendability of the log-likelihood ratio function $f$ to $\mathbb{C}^d$.
Although we do not make this assumption in our Fundamental Conditions \ref{ass:fundamental condition} and \ref{ass:fundamental conditionII}, we will show below that an analogous assumption implies Assumption \ref{ass:S2}.
If the realizability condition (Definition \ref{defn:realizability}) holds, then it is known that 
\begin{align}
\label{ineq:K less than K^Q}
    K(\theta) \leq K^Q(\theta)
\end{align}
in the quantum information theory, combined with the fundamental condition. 
\begin{prop}
\label{prop:Liouville and k=k^Q}
    Assume that the realizability condition (Definition \ref{defn:realizability}) holds.
    If there exist two holomorphic functions
    \[K_{\mathbb{C}}, K^Q_{\mathbb{C}}:\mathbb{C}^d \to \mathbb{C},\]
    whose restrictions to $\mathbb{R}^d$ coincide with $K$, $K^Q$, and Ineq. \eqref{ineq:K less than K^Q} still holds as in the form
    \[|K_{\mathbb{C}}(\theta)| \leq |K_{\mathbb{C}}^Q(\theta)|\]
    for all $\theta\in\mathbb{C}^d$, then Assumption \ref{ass:S2} is satisfied.
    More strongly, there exists a constant $c$ so that $K(\theta) = c\cdot K^Q(\theta)$.
\end{prop}
\begin{proof}
The assumption on the inequality between $K_{\mathbb{C}}$ and $K_{\mathbb{C}}^Q$, the function \[F(\theta)\coloneqq \frac{K_{\mathbb{C}}(\theta)}{K_{\mathbb{C}}^Q(\theta)}\]
is a bounded entire function on $\mathbb{C}^d$.
Hence, applying Liouville's theorem for multivariable to $F$, we can deduce that $F$ is a constant function.
Note that a constant multiple does not affect the computation of $k$ in the presentation Eq. \eqref{eq:resolution of K and K^Q}, which concludes that $k=k^Q$.
\end{proof}

Though in our case, Proposition \ref{prop:Liouville and k=k^Q} does not hold in general, the fundamental conditions restrict us to only work on the functions behaving naturally on the parameter set.
Moreover, the existence of the nowhere-vanishing functions does not affect the computation of the geometrical quantities $k$ and $k^Q$.
Hence, in our context, we propose the following conjecture, connecting quantum information theory and algebraic geometry.

\begin{conj}
\label{conj:k=k^Q}
The average log loss functions $K$ and $K^Q$ behave similarly under the fundamental conditions.
In other words,
    \[k=k^Q\]
    holds on any affine open subsets.
\end{conj}

\subsection{Empirical process with log-likelihood} 
\label{subsec:Empirical process with log-likelihood}
Empirical processes provide a framework for analyzing the random behavior of empirical functions based on observed data, enabling investigations into the asymptotic behavior of estimators of parameters (or posterior distribution) and training loss.
In singular learning theory, it is also a crucial step to consider the empirical process of a log-likelihood ratio (Definition \ref{defn:matrix I and J} (1)) on the parameter space of $u$ after the resolution of singularities, ensuring that it converges to a Gaussian process on $u$ in distribution even in a neighborhood of $K(g(u))=0$.
For our purpose, we first review the empirical process with classical log-likelihood to investigate the behavior of the posterior distribution in this subsection and extend it to its quantum analog with classical shadow for the quantum training loss in the next subsection.

\begin{defn}
\label{defn:matrix I and J}
   Let us define the basic notion in classical statistics.
   \begin{enumerate}
       \item The \textit{log-likelihood ratio function} is defined by 
            \[f(x,\theta) \coloneqq \log \frac{p(x|\theta_0)}{p(x|\theta)}.\]
        \item Matrices $I(\theta)$ and $J(\theta)$ are defined by 
            \begin{align}
                I(\theta) \coloneqq \mathbb{E}_X\left[ \left(\frac{\partial \log p(X|\theta)}{\partial \theta}\right) \left(\frac{\partial \log p(X|\theta)}{\partial \theta}\right)^T  \right], \quad
                J(\theta) \coloneqq - \mathbb{E}_X \left[ \frac{\partial^2 \log p(X|\theta)}{\partial \theta^2} \right],
            \end{align}
            where $J(\theta)$ is the Hessian matrix of the KL divergence.
            We denote by $I \coloneqq I(\theta_0)$ and $J \coloneqq J(\theta_0)$.
            In the realizable case, where $p(x|\theta_0) = q(x), \forall x$ holds, $I(\theta)$ is equivalent to the Fisher information matrix, and the relation $I = J$ holds.
   \end{enumerate}
\end{defn}
The empirical process of the log-likelihood ratio function and its relevant quantity is defined as
\begin{align}
    \eta_n(\theta) &\coloneqq \frac{1}{\sqrt{n}} \sum_{i=1}^{n} \left\{ \mathbb{E}_X\left[ f(X,\theta) \right] - f(X_i,\theta) \right\}, \label{eq:eta_n}\\
    \Delta_n &\coloneqq \frac{1}{\sqrt{n}} J^{-1} \nabla \eta_n(\theta_0).
    \label{eq:Delta_n}
\end{align}
The vector $\Delta_n$ assumes the inverse of $J$ and is only valid in regular cases. 
Importantly, we can see that the mean and covariance of $\eta_n$ do not depend on $n$:
\begin{align}
    \mathbb{E}_{X^n}[\eta_n(\theta)] &= 0, \\
    \mathbb{E}_{X^n}[\eta_n(\theta) \eta_n(\theta')] 
    &= \mathbb{E}_X[f(X,\theta)f(X,\theta')] - \mathbb{E}_X[f(X,\theta)]\mathbb{E}_X[f(X,\theta')],
\end{align}
for $\theta,\theta' \in \Theta$.
Similarly, we can see that 
\begin{align}
    &\mathbb{E}_{X^n}[\nabla\eta_n(\theta)] = \bm{0}, \\
    &\mathbb{E}_{X^n}[\nabla\eta_n(\theta) \nabla\eta_n(\theta')^T] = \mathbb{E}_X[\nabla \log p(X|\theta)\nabla \log p(X|\theta')^T] - \mathbb{E}_X[\nabla \log p(X|\theta)]\mathbb{E}_X[\nabla \log p(X|\theta')^T].
\end{align}
In particular, when $\theta = \theta' = \theta_0$, 
\begin{align}\label{eq:cov_nabla_eta}
    \mathbb{E}_{X^n}[\nabla\eta_n(\theta_0) \nabla\eta_n(\theta_0)^T] = \mathbb{E}_X[\nabla \log p(X|\theta_0)\nabla \log p(X|\theta_0)^T] = I,
\end{align}
where we use $\mathbb{E}_X[\nabla \log p(X|\theta)] = \bm{0}$.

In singular cases, we consider the renormalized empirical process:
\begin{equation}
\label{eq:renormalized empirical process}
    \xi_n(u) \coloneqq \frac{1}{\sqrt{n}} \sum_{i=1}^{n} \{ u^k - a(X_i,u) \}, 
\end{equation}
where $u$ and $k$ are appeared in Eq. \eqref{eq:resolution of K and K^Q}, and $a(\cdot,\cdot)$ is appeared in Definition \ref{defn:invariants in singular learning theory} (3).
A similar calculation involved in $\eta_n$ implies that
\begin{align}
    \mathbb{E}_{X^n}[\xi_n(u)] &= 0, 
    \label{eq:mean_xi} \\
    \mathbb{E}_{X^n}[\xi_n(u) \xi_n(u')] 
    &= \mathbb{E}_X\left[ a(X,u)a(X,u') \right] - u^k u'^k, 
    \label{eq:cov_xi}
\end{align}
for $u,u' \in \widetilde{\Theta}$.

The empirical processes $\eta_n(\theta)$ and $\xi_n(u)$ introduced above can be shown to converge in law to Gaussian processes, as stated in the following:
\begin{prop}[{\cite[Theorem 6.2]{watanabe2009algebraic}}]
\label{prop:classical convergence law}
    Assume $f(\cdot,\theta)$ defined in Definition \ref{defn:matrix I and J} (1) is an $L^2(q)$-valued analytic function.
    When $\Theta$ is compact, the empirical processes $\eta_1, \eta_2, ...$ and $\xi_1, \xi_2, ...$ converge in law to the Gaussian process $\eta$ and $\xi$, respectively.
\end{prop}

Note that we write the expectation of joint probability distribution over $X_1,...,X_n$ as $\mathbb{E}_{X^n}[\cdot]$ for each $n$, but we write the expectation concerning $X_1,...,X_n$ for $n \to \infty$, i.e., the Gaussian process, as $\mathbb{E}_\xi[\cdot]$.

\subsection{Empirical process with classical shadow}
\label{subsec:Empirical process with classical shadow}
First, we introduce the following as an analog of Appendix \ref{subsec:Empirical process with log-likelihood}.
\begin{defn}
\label{defn:matrix I and J_quantum}
   Let us define some quantities related to quantum information theory.
   \begin{enumerate}
        \item The \textit{quantum log-likelihood ratio function} is defined by 
            \begin{equation}
            \label{eq:q_log_likelihood_ratio}
                f^Q(\hat{\rho}_x,\theta) \coloneqq \Tr(\hat{\rho}_x \{ \log \sigma(\theta_0^Q) - \log \sigma(\theta) \}),
            \end{equation}
            with a classical snapshot $\hat{\rho}_x$.
        \item  Matrices $I^Q(\theta)$ and $J^Q(\theta)$ are defined by
            \begin{equation}
                I^Q(\theta) \coloneqq \Tr\left( \rho \left(\frac{\partial \log \sigma(\theta)}{\partial \theta}\right) \left(\frac{\partial \log \sigma(\theta)}{\partial \theta}\right)^T \right),\quad    
                J^Q(\theta) \coloneqq - \Tr\left( \rho \frac{\partial^2 \log \sigma(\theta)}{\partial \theta^2} \right),
            \end{equation}
            where $J^Q(\theta)$ is the Hessian matrix of the quantum relative entropy.
            We denote by $I^Q \coloneqq I(\theta_0^Q)$ and $J^Q \coloneqq J(\theta_0^Q)$.
            In the realizable case, where $\sigma(\theta_0^Q) = \rho$ holds, $J^Q$ is equivalent to the Bogoliubov Fisher information matrix \cite{amari2000methods}.
   \end{enumerate}
\end{defn}
Similar to the previous subsection, we also define the following empirical process and its relevant quantity:
\begin{align}
    \eta_n^Q(\theta) &\coloneqq \frac{1}{\sqrt{n}}\sum_{i=1}^{n} \left\{ \mathbb{E}_X\left[ f^Q(\hat{\rho}_X,\theta) \right] - f^Q(\hat{\rho}_{X_i},\theta) \right\},
    \label{eq:def_eta_n^Q} \\
    \Delta_n^Q &\coloneqq \frac{1}{\sqrt{n}} {J^Q}^{-1} \nabla \eta_n^Q(\theta_0^Q). \label{eq:def_Delta_n^Q}
\end{align}
Note that $\mathbb{E}_X\left[ f^Q(\hat{\rho}_X,\theta) \right] = \Tr(\rho \{ \log \sigma(\theta_0^Q) - \log \sigma(\theta) \})$.
The definition of $\Delta_n^Q$ requires the existence of the inverse of $J^Q$ and is only valid in quantum regular cases.

In singular cases, we also consider the quantum analog of the renormalized empirical process
\begin{align}\label{eq:xi_n^Q}
    \xi_n^Q(u) \coloneqq \frac{1}{\sqrt{n}} \sum_{i=1}^{n} \left\{ r(u) u^{k^Q} - a^Q(\hat{\rho}_{X_i}, u) \right\}, 
\end{align}
where $u$, $k^Q$, and $r(u)$ are appeared in Eq. \eqref{eq:resolution of K and K^Q} and $a^Q(\cdot,\cdot)$ is appeared in Eq. \eqref{eq:intro_a^Q}.

The empirical processes $\eta_n^Q(\theta)$ and $\xi_n^Q(u)$ converge in law to Gaussian processes.
\begin{prop}
\label{prop:quantum convergence law}
    Assume $f^Q(\cdot,\theta)$ defined in Definition \ref{defn:matrix I and J_quantum} (1) is an $L^2(q)$-valued analytic function.
    When $\Theta$ is compact, the empirical processes $\eta_1^Q, \eta_2^Q, ...$ and $\xi_1^Q, \xi_2^Q, ...$ converge in law to the Gaussian process $\eta^Q$ and $\xi^Q$, respectively.
\end{prop}
\begin{proof}
    We follow the proof of \cite[Theorem 6.2]{watanabe2009algebraic} for the case of $f(\cdot,\theta)$.
    Then, from \cite[Theorem 5.9]{watanabe2009algebraic}, we immediately obtain this proposition for $\eta_n^Q$ and $\xi_n^Q$.
    Here, note that we used Fundamental conditions \ref{ass:fundamental condition} and \ref{ass:fundamental conditionII}.
\end{proof}

\section{Review of singular learning theory}
\label{section:Revisiting singular learning theory and WAIC}
To prove our main results for quantum singular models, we utilize singular learning theory established by Watanabe \cite{watanabe2009algebraic,watanabe2018mathematical}.
In this section, we briefly recall the setting and results of singular learning theory; see also the recently published textbook \cite{suzuki2023WBayes}.

Classically, Bayesian statistics \cite{savage1972foundations} provides a posterior distribution that reflects prior knowledge, accounting for the uncertainty of statistical models. It is useful for analyzing complex systems such as deep learning \cite{lecun2015deep} while minimizing generalization loss is a generally difficult problem when the model is too complex, and singularities arise in the parameter space. This leads to the practical issue that the probability distribution representing all of nature is too large, meaning any model chosen by humans is wrong in some sense \cite{watanabe2022allmodels}. It is commonly referred to as ``all models are wrong" \cite{box1976science}.
However, the statement continues; ``but some are useful".
As Boyle's Law offers \cite{boyle1967new} or von Neumann said \cite{von1947mathematician}, we need to find a model that provides a good approximation to a true probability distribution.
We introduce singular learning theory, which provides the mathematical foundation for Bayesian statistics for such complex models.

Only in this section, we assume \cite[Fundamental conditions I, II]{watanabe2009algebraic}, the classical analogs of Fundamental conditions \ref{ass:fundamental condition} and \ref{ass:fundamental conditionII}, but omit the detailed descriptions for brevity.

\subsection{Introductory concepts and regular learning theory}
\label{subsection:Preliminaries on WAIC in terms of $G_n$ and $T_n$}
In this subsection, let us review the basic theorem of Bayesian statistics and the learning theory for regular models.
Specifically, we focus on the asymptotic behavior of the generalization and training losses \eqref{eq:def of Gn and Tn}, which shed light on many aspects of Bayesian statistics, including the effective number of parameters.

First, we define the posterior mean and variance.
\begin{defn}
\label{defn:posterior mean for classical}
    For a function $s:\Theta\to \mathbb{R}$ and given samples $x_1,\cdots,x_n$, let us define the \textit{posterior mean} and \textit{posterior variance}
    as
    \begin{align}
        \E[s(\theta)] &\coloneq \int_{\Theta} s(\theta) p(\theta|x_1,\cdots,x_n) d\theta,\\
        \mathbb{V}_{\theta}[s(\theta)] &\coloneqq \int_\Theta s(\theta)^2 p(\theta|x_1,\cdots,x_n) d\theta - \left( \int_\Theta s(\theta) p(\theta|x_1,\cdots,x_n) d\theta \right)^2.
    \end{align}
\end{defn} 

Recall that the generalization and training losses are defined as
\begin{align}
    G_n \coloneqq - \mathbb{E}_X[ \log \mathbb{E}_\theta[p(X|\theta)]], \quad T_n \coloneqq - \frac{1}{n} \sum_{i=1}^{n}\log \mathbb{E}_\theta[p(x_i|\theta)].
\end{align}
An important observation is that cumulant generating functions are valuable tools for studying the asymptotic behaviors of $G_n$ and $T_n$.
\begin{defn}[{\cite[Definition 8]{watanabe2018mathematical}}]
\label{defn:classical_cumulant}
    For a real number $\alpha\in\mathbb{R}$, the cumulant generating functions of the generalization and training losses are respectively defined by
    \[\mathcal{G}_n(\alpha) \coloneqq \mathbb{E}_X[\log \mathbb{E}_\theta[p(X|\theta)^{\alpha}]],\quad \mathcal{T}_n(\alpha) \coloneqq \frac{1}{n}\sum_{i=1}^n\log \mathbb{E}_\theta[p(X_i|\theta)^{\alpha}]. \]
\end{defn}

Using these cumulant generating functions, we can derive the following fundamental result, which serves as a starting point for asymptotic learning theory.
Assuming certain assumptions on the cumulants, the following proposition holds:
\begin{prop}[Basic theorem of Bayesian statistics; informal statement of {\cite[Theorem 3]{watanabe2018mathematical}}]
\label{prop:generalization loss $G_n$ and the training loss $T_n$ for general cases}
    The generalization loss $G_n$ and the training loss $T_n$ can be expanded as 
    \begin{align}
        G_n &= - \mathcal{G}_n(1) = - \mathcal{G}_n'(0) - \frac{1}{2}\mathcal{G}_n''(0) + o_p\left(\frac{1}{n}\right), \\
        T_n &= - \mathcal{T}_n(1) = - \mathcal{T}_n'(0) - \frac{1}{2}\mathcal{T}_n''(0) + o_p\left(\frac{1}{n}\right).
    \end{align}
    Furthermore, the first and second cumulants of the generalization and training losses at $\alpha=0$ are provided below:
    \begin{align}
        \mathcal{G}_n'(0) = \mathbb{E}_X[\mathbb{E}_\theta[\log p(X|\theta)]], &\quad
        \mathcal{T}_n'(0) = \frac{1}{n} \sum_{i=1}^{n} \mathbb{E}_\theta[\log p(X_i|\theta)], \\
        \mathcal{G}_n''(0) = \mathbb{E}_X[\mathbb{V}_{\theta}[\log p(X|\theta)]], &\quad
        \mathcal{T}_n''(0) = \frac{1}{n} \sum_{i=1}^{n} \mathbb{V}_{\theta}[\log p(X_i|\theta)].
    \end{align}
\end{prop}
Building upon this basic theorem, Watanabe \cite{watanabe2018mathematical} investigated the asymptotic behaviors in both regular and singular cases, leading to a natural generalization of the previous work on model selection criteria like AIC. 
Let us first introduce the descriptions for regular cases where the Hessian of the KL divergence does not degenerate.
In such situations, famous useful properties in classical statics hold, such as the asymptotic normality, which implies that the posterior distribution can be approximated by a normal distribution.

With the quantities introduced in Appendix \ref{subsec:Empirical process with log-likelihood}, the mean and variance of parameters with respect to the posterior distribution are described as follows:
\begin{lem}[{\cite[Lemma 14]{watanabe2018mathematical}}]
\label{lem:gb_lem14}
    When the regular condition (Definition \ref{def:regular for classical}) is satisfied, the following equations hold.
    \begin{align}
        \mathbb{E}_\theta[\theta - \theta_0] &= \Delta_n + o_p\left( \frac{1}{\sqrt{n}} \right), \\
        \mathbb{E}_\theta[(\theta - \theta_0)(\theta - \theta_0)^T] &= \frac{1}{n} J^{-1} + \Delta_n \Delta_n^T + o_p\left( \frac{1}{n} \right), \\
        \mathbb{E}_{X^n}\left[ \mathbb{E}_\theta[(\theta - \theta_0)(\theta - \theta_0)^T] \right] &= \frac{1}{n} \left( J^{-1} + J^{-1} I J^{-1} \right) + o\left( \frac{1}{n} \right).
    \end{align}
\end{lem}
\begin{rem}
    If a pair $(q(x),p(x|\theta))$ is classically regular, then the asymptotic normality of the posterior distribution holds.
    This is a typical property resulting from the regular condition \cite[Chapter 4]{watanabe2018mathematical}. 
\end{rem}

Having discerned the asymptotic behavior of the parameters, we perform a second-order Taylor expansion of the equations in Proposition \ref{prop:generalization loss $G_n$ and the training loss $T_n$ for general cases} around the optimal parameter $\theta_0$ to obtain more detailed descriptions.
\begin{prop}[{\cite[Theorems 6 and 7]{watanabe2018mathematical}}]
\label{prop:generalization loss of G_n and T_n}
     Let us suppose that the regular condition (Definition \ref{def:regular for classical}) is satisfied.
     \begin{enumerate}
         \item The generalization loss $G_n$ and training loss $T_n$ can be expanded as follows:
        \begin{align}
            G_n &= \mathbb{E}_X[- \log p(X|\theta_0)] + \frac{d}{2n} + \frac{1}{2} \Delta_n^T J \Delta_n - \frac{1}{2n} \Tr(IJ^{-1}) + o_p\left(\frac{1}{n}\right), \\
            T_n &= - \frac{1}{n} \sum_{i=1}^{n} \log p(x_i|\theta_0) + \frac{d}{2n} - \frac{1}{2} \Delta_n^T J \Delta_n - \frac{1}{2n} \Tr(IJ^{-1}) + o_p\left(\frac{1}{n}\right).
        \end{align}
    \item The expected values of $G_n$ and $T_n$ with respect to given samples are
    \begin{eqnarray*}
        \mathbb{E}_{X^n}[G_n] &=& \mathbb{E}_X[- \log p(X|\theta_0)] + \frac{\lambda}{n} + o\left(\frac{1}{n}\right), \\
        \mathbb{E}_{X^n}[T_n] &=& \mathbb{E}_X[- \log p(X|\theta_0)] + \frac{\lambda - 2 \nu}{n} + o\left(\frac{1}{n}\right),
        \label{eq:exp_T}
    \end{eqnarray*}
    where
    \begin{equation}
        \lambda = \frac{d}{2}, \quad
        \nu = \frac{1}{2} \Tr(IJ^{-1}).
        \label{eq:lambda_and_nu}
    \end{equation}
    \end{enumerate}
\end{prop}

\begin{rem}
\label{rem:regular_RLCT}
    If the regularity (Definition \ref{def:regular for classical}) and realizability (Definition \ref{defn:realizability}) conditions are satisfied, then $\lambda = \nu = d/2$.
    However, as we see below, in singular cases, these two quantities do not coincide and are not equal to $d/2$ in general.
    This reflects the complexity of the singularities contained in $\Theta_0$, a key insight in singular learning theory.    
\end{rem}
These formulas imply that the generalization and training losses can be expressed by the dimension of parameters $d$, the Fisher information matrix $I$, and the Hessian of the KL divergence $J$ in classically regular cases.

In learning theory and the analysis of machine learning models, constructing criteria for model selection is of significant importance \cite{anderson2004model}.
The expressions in Proposition \ref{prop:generalization loss of G_n and T_n} lead to deriving an information criterion WAIC.

\begin{defn}[{\cite[Definitions 3]{watanabe2018mathematical}}]
\label{defn:aic and waic}
The \textit{Widely Applicable Information Criterion} (WAIC) is defined as
    \begin{equation}
        \mathrm{WAIC} \coloneqq T_n + \frac{1}{n} \sum_{i=1}^{n} \mathbb{V}_{\theta}[\log p(x_i|\theta)].
        \label{def:waic}
    \end{equation}
\end{defn}

One can prove that WAIC is an asymptotically unbiased estimator for the generalization loss $G_n$ in regular cases.
\begin{cor}
\label{cor:waic}
    Assuming that the statistical model satisfies the regularity condition (Definition \ref{def:regular for classical}), WAIC is an asymptotically unbiased estimator for $G_n$.
    More precisely, 
    \begin{equation} 
      \mathbb{E}_{X^n}[G_n] = \mathbb{E}_{X^n}[\mathrm{WAIC}] + o\left(\frac{1}{n}\right).
    \end{equation}
\end{cor}

\begin{rem}
    Under the regularity and realizability condition, the information criterion AIC \cite{akaike1973Information,akaike1974New} can also be derived from the KL divergence 
    \begin{equation}
        \mathrm{AIC} \coloneqq -\frac{1}{n}\sum_{i=1}^n \log p(x_i|\hat{\theta}) + \frac{d}{n}.
        \label{def:aic}
    \end{equation}
    It is an asymptotically unbiased estimator for the generalization loss:
    \begin{equation}
        \mathbb{E}_{X^n}[- \mathbb{E}_X[ \log p(X|\hat{\theta})]] = \mathbb{E}_{X^n}[\mathrm{AIC}] + o\left(\frac{1}{n}\right),
    \end{equation}
    Note that AIC is defined for the maximum likelihood estimator $\hat{\theta}$, and so is the generalization loss.
\end{rem}

For applications of these criteria, see Section \ref{subsection:singular learning theory}.
In Corollary \ref{cor:waic}, we assumed that the pair $(q(x),p(x|\theta))$ is classically regular.
However, one can generalize the results without assuming the regularity condition.
In the next subsection, we will see how this has been done using mathematical tools such as algebraic geometry (Appendix \ref{app:algebraic geometry}).

\subsection{Singular learning theory}
For singular models, the optimal parameter set $\Theta_0$ contains multiple (singular) points in general.
In such cases, the Hessian of the KL divergence degenerates, which hinders the formulation of a learning theory in the same way as in regular cases.
One solution is to evaluate the complexity of singularities using algebraic geometrical methods and develop a new theory that accounts for this complexity.
To analyze the complexity of such a singularity, recall that singular learning theory introduces the following average log loss function (Eq. \eqref{eq:average log loss function})
    \[K(\theta)\coloneqq \mathrm{KL}(p(x|\theta_0) || p(x|\theta)).\]

The key insight of singular learning theory is that the space in which the function $K(\theta)$ should be properly defined is not the original parameter space $\Theta$, but rather a space with milder singularities.
Specifically, the resolution theorems (Theorems \ref{thm:resolution_original} and \ref{thm:resolution_Watanabe}) ensure that there exists an analytic map $g$ from a manifold to the parameter space, i.e., $\theta = g(u)$, such that the average log loss function $K(g(u))$ has a normal crossing form at a neighborhood of $K(g(u))=0$.
By virtue of this theorem, one can derive that the volume of the almost optimal set $\{\theta | K(\theta)<t \}$ as $t \to +0$ is in proportion to $t^\lambda (- \log t)^{m-1}$ with the RLCT $\lambda$ and a multiplicity $m$, both of which are well-known invariants in algebraic geometry.
Moreover, the empirical process $\eta_n(\theta)$, which characterizes the posterior distribution, can be transformed into a well-defined function $\xi_n(u)$ of $u$ at a neighborhood of $K(g(u))=0$ that converges in law to a Gaussian process as discussed in Appendix \ref{subsec:Empirical process with log-likelihood}.

Singular learning theory builds upon these observations to study the asymptotic behavior of the posterior distribution.
Note that the marginal likelihood, or the partition function, is proportional to \[\int \exp(-\sum_{i=1}^{n} f(x_i,\theta)) \pi(\theta) d\theta.\]
We can redefine this marginal likelihood by the integral with respect to $u$ using the renormalized empirical process $\xi_n(u)$.
One of the key insights is that, as is well known in statistical physics, the partition function is the Laplace transform of a state density, which corresponds to the volume of the optimal parameter set in the learning theory.
Therefore, the marginal likelihood can also be described by the volume of the optimal set, $t^\lambda (- \log t)^{m-1}$.
Consequently, one can successfully represent the posterior distribution with the RLCT $\lambda$ and a multiplicity $m$, which express the complexity of singularities.
Especially, as will be shown in subsequent discussions, the RLCT helps in understanding the singular dimension of a statistical model and the accuracy of Bayesian statistics.

\begin{rem}
The pair of the invariants $(\lambda, m)$ has at least three aspects: the description of the marginal likelihood, as introduced above, the invariants in algebraic geometry, and finally, the connection with the zeta function.
Although this paper does not address it in depth, it is an interesting multidisciplinary issue, and hence, it is addressed here.

In the Amsterdam Congress of 1954, Gel'fand proposed the problem of the meromorphic continuation of the zeta function associated with the complex power of holomorphic functions.
Nowadays, this zeta function is called the local zeta function and is treated over local fields \cite{atiyah1970resolution}, including non-archimedian fields.
We will only introduce it here mainly over $\mathbb{R}$ (or $\mathbb{C}$).
The above problem was solved by Atiyah \cite{atiyah1970resolution} and Bernstein-Gel'fand  \cite{bernstein1969meromorphic} by using Hironaka's resolution theorem (Theorem \ref{thm:resolution_original}).
This solution method connects the pair $(\lambda, m)$ and the zeta function of the learning model, associated with $K(\theta)$ as follows.
In singular learning theory, the zeta function is defined as
\begin{align}
\label{eq:zeta function for learning theory}
\zeta(z) \coloneqq \int K(\theta)^z\pi(\theta) d\theta,    
\end{align}
which is a special version of the local zeta function.
Substituting a normal crossing presentation (Eq. \eqref{eq:resolution of K and K^Q}) into Eq. \eqref{eq:zeta function for learning theory}, the Laurent series expansion of the zeta function has the form
\[\zeta(z) = \sum_{k=1}^{\infty} \frac{S_k(\pi)}{(z + \lambda_k)^{m_k}}\]
where $\{S_k\}$ is a set of Schwartz distributions.
Here, the sequence 
\[\cdots < -\lambda_2 < -\lambda_1 < 0\]
defines poles of the zeta functions with multiplicities $m_k$.
Now, the learning coefficient, or the RLCT, $\lambda$ coincides with $\lambda_1$, and the multiplicity $m$ is $m_1$ in the above definitions.
Note that there is another solution to the above problem.
It is the method using algebraic analysis, particularly $b$-functions developed by many mathematicians like Kashiwara \cite{kashiwara1976b}, Bernstein \cite{bernstein1972analytic}, Sato and Shintani \cite{sato1974zeta}.
See also \cite[Chapter 4]{watanabe2009algebraic} for history.

This shows that singular learning theory can be explained in terms of algebraic geometry and zeta functions. For regular models, the Fisher information matrix has played a central role. However, the non-degeneracy of the Fisher information matrix does not hold under birational transformations. In the analysis of singular models, the invariants introduced above characterize the learning generalization error $G_n$ as a generalization of the fact that the Gaussian distribution is characterized by its mean and variance.
\end{rem}

Before studying the asymptotic behaviors of $G_n$ and $T_n$, we summarize several important quantities introduced in the singular theory.
\begin{defn}
\label{defn:invariants in singular learning theory}
    Let $K(\theta)$ be the average log loss function with respect to a model $(q(x), p(x|\theta))$.
    \begin{enumerate}
        \item The \textit{learning coefficient} $\lambda$ is defined as the real log canonical threshold associated to $\Theta_0$.
        In other words, 
        \[\lambda \coloneqq \min_i\left\{\frac{2k_i + 1}{h_i} \right\},\]
        where  $k_i$ and $h_i$ are defined through a log resolution of singularities Eq. \eqref{eq:resolution of K and K^Q}.
        \item The \textit{renormalized log likelihood function} $a(x,u)$ is defined from the standard form $f(x,g(u)) = u^k \cdot a(x,u)$ in \cite[Definition 14]{watanabe2018mathematical}.
        \item The \textit{renormalized empirical process} $\xi_n(u)$ is introduced in Eq. \eqref{eq:renormalized empirical process}.
        \item Let $t \coloneqq n \cdot u^{2k} \in \mathbb{R}$. The origin of this is the state density formula and the Mellin transformation. In this paper, we treat it simply as a variable.
        \item Let us define the \textit{functional variance} and \textit{singular fluctuation} as
        \begin{align*}
            V(\xi_n) &\coloneqq \mathbb{E}_X[\mathbb{V}_{\theta}[\sqrt{t}a(X,u)]], \\
            \nu &\coloneqq \frac{1}{2}\mathbb{E}_{\xi}[V(\xi)].
        \end{align*}
        Note that $V(\xi_n)$ is a functional of $\xi_n$ because $\mathbb{V}_{\theta}[\cdot]$ depends on $\xi_n$.
    \end{enumerate}
\end{defn}
It is worth reminding that, contrary to regular cases, the posterior distribution cannot be asymptotically approximated by a normal distribution. 
However, it can be represented with parameters $(t,u)$ after taking the resolution of singularities.
\begin{defn}[{\cite[Definition 19]{watanabe2018mathematical}}]
\label{defn:renormalized_posterior_dist}
    For a function $z(t,u)$, the renormalized posterior distribution is defined by
    \begin{equation}
        \mathbb{E}_\theta[z(t,u)] = \frac{\int du \, D(u) \int_0^\infty dt \, z(t,u) t^{\lambda-1} \exp(-t+\sqrt{t}\xi_n(u))}{\int du \, D(u) \int_0^\infty dt \, t^{\lambda-1} \exp(-t+\sqrt{t}\xi_n(u))}
        \label{eq:renormalized_posterior_dist}
    \end{equation}
    where $D(u)$ is defined in Eq. (5.27) of Ref. \cite{watanabe2018mathematical}.
\end{defn}
As noted in \cite[Chapter 6]{watanabe2018mathematical}, the posterior distribution for many models and priors is represented as a finite mixture of locally standard forms.
That is, there exists a division of a parameter set such that the average log loss function can be expressed by a normal crossing representation in each local parameter set.
For such general posterior distribution, we consider a set of local coordinates $\{U_j\}$ in the resolution theorem, and only the $U_j$ that maximize $t^{\lambda_j} (- \log t)^{m_j-1}$ are referred to as essential local coordinates (ELC).
Thus, $D(u)$ in Definition \ref{defn:renormalized_posterior_dist} is replaced by $\sum_{j \in \mathrm{ELC}} D_j(u)$ for general posterior distribution.

Having presented several key notions and quantities, the asymptotic expansions of $G_n$ and $T_n$ can be obtained by representing the cumulants in Proposition \ref{prop:generalization loss $G_n$ and the training loss $T_n$ for general cases} using the standard form and some integral calculations with the renormalized posterior distribution
This generalizes the result of Proposition \ref{prop:generalization loss of G_n and T_n} to singular cases.
\begin{prop}[{\cite[Theorems 14, 15]{watanabe2018mathematical}}]
\label{prop:singular for Gn and Tn}
Even if a pair ($q(x), p(x|\theta))$ is singular, under Definition \ref{def:classical relatively finite variance}, we obtain the following results.
    \begin{enumerate}
        \item The generalization and training losses are expanded as
    \begin{eqnarray*}
        G_n &=& \mathbb{E}_X \left[ - \log p(X|\theta_0) \right] + \frac{1}{n}\left( \lambda + \frac{1}{2}\mathbb{E}_\theta[\sqrt{t} \xi_n(u)] - \frac{1}{2} V(\xi_n) \right) + o_p\left(\frac{1}{n}\right), \\
        T_n &=& \frac{1}{n} \sum_{i=1}^n \{ - \log p(x_i|\theta_0) \} + \frac{1}{n}\left( \lambda - \frac{1}{2}\mathbb{E}_\theta[\sqrt{t} \xi_n(u)] - \frac{1}{2} V(\xi_n) \right) + o_p\left(\frac{1}{n}\right).
    \end{eqnarray*}
    \item 
   Their expectations are given by
    \begin{align*}
        \mathbb{E}_{X^n}[G_n] &= \mathbb{E}_X[- \log p(X|\theta_0)] + \frac{\lambda}{n} + o\left(\frac{1}{n}\right), \\
        \mathbb{E}_{X^n}[T_n] &= \mathbb{E}_X[- \log p(X|\theta_0)] + \frac{\lambda - 2 \nu}{n} + o\left(\frac{1}{n}\right).
    \end{align*}
    \end{enumerate}
\end{prop}

Now, though we omit the details of the proof here, let us state that WAIC is a generalization of AIC in the following sense.
\begin{thm}[{\cite[Theorem 16]{watanabe2018mathematical}}]
\label{thm:WAIC is unbiased estimator}
    Even if a pair ($q(x), p(x|\theta))$ is singular, under Definition \ref{def:classical relatively finite variance}, WAIC has a property as an asympto unbiased estimator:
    \[  \mathbb{E}_{X^n}[G_n] = \mathbb{E}_{X^n}[\mathrm{WAIC}] + o\left(\frac{1}{n}\right).\]
\end{thm}

Note that to develop singular learning theory, including Theorem \ref{thm:WAIC is unbiased estimator}, we have to prepare a reasonable situation for analysis \cite[Fundamental conditions I, II]{watanabe2009algebraic}, which will be reformulated in Fundamental conditions \ref{ass:fundamental condition} and \ref{ass:fundamental conditionII} for the usage of our setting.
This condition is intended to make the setting explicit and is considered to be fulfilled in most situations, but caution is required in the following example.

\begin{ex}
\label{ex:normal mixture}
    One of the most famous examples of singular models is a normal mixture \cite[Section 7.3]{watanabe2009algebraic}, \cite[Section 2.5]{watanabe2018mathematical}.
    This is a model in which the two Gaussian distributions are parameterized in suitable proportions.
    For $i=1,2$, let 
    \[f_i(x|\mu_i) \coloneqq \frac{1}{\sqrt{2\pi}}\exp\left(-\frac{(x-\mu_i)^2}{2}\right)\]
    be the Gaussian distribution with mean $\mu_i$ and variance $1$.
    We now consider the normal mixture 
    \begin{align}
    \label{eq:normal mixture first}
        p(x|a,\mu_1,\mu_2) \coloneqq (1-a)f_1(x|\mu_1) + af_2(x|\mu_2)
    \end{align}
    for $0\leq a \leq 1$.
    For simplicity, putting $\mu_1 = 0$, $\mu_2 = b$, 
    \[\Theta \coloneqq \{(a,b) \in \mathbb{R}^2 \mid 0\leq a \leq 1\} \subset \mathbb{R}^2,\]
    we assume the true distribution is $q(x) = p(x|\alpha, 0, \beta)$ for some $\alpha, \beta \in \Theta$ with $\alpha\beta = 0$.
    Then, by \cite[Example 14]{suzuki2023WBayes}, we have
    \[K(a,b) = \mathrm{KL}(q||p) = \int_{\mathbb{R}} \log \left(\frac{1 + \alpha(\exp(\beta x - \beta^2/2) - 1)}{1 + a(\exp(bx - b^2/2) - 1)}\right)q(x)dx,\]
which implies that 
\[\Theta_0 = \{(a,b) \in \mathbb{R}^2\mid ab = 0\}.\]
Hence, the model $(q(x), p(x|a,b))$ is realizable but singular.
We note the parametrization of this model for further discussion.
The parameter $(0,b)$ is contained in $\Theta_0$ but located in the boundary of $\Theta$.
This means that we cannot discuss the differential of the average log loss function $K(a,b)$ there.
Hence, we cannot apply singular learning theory a priori.
Though it is known that a similar approach can be developed in this case through individual discussions \cite[Section 7.3]{watanabe2009algebraic}, we would like to present a simpler way around this.
It is to reparametrize the model Eq. \eqref{eq:normal mixture first} as
\begin{align}
    p(x|\theta,\mu_1,\mu_2) \coloneqq \cos^2(\theta) f_1(x|\mu_1) + \sin^2(\theta) f_2(x|\mu_2)
\end{align}
for $0\leq \theta < 2\pi$.
Then we can avoid the above difficulty.
We also use this reparametrization method in the quantum setting; see \cite[Remark 6.1 (8)]{watanabe2009algebraic}.
\end{ex}

\section{Proof of main results}
\label{app:Proof of main results}
In this section, we prove the main theorems concerning the asymptotic behaviors of $G_n^Q$ and $T_n^Q$ as well as the asymptotic unbiasedenss of our proposed information criterion QWAIC.

First, we consider generalizations of the posterior mean and variance to matrices, extending Definition \ref{defn:posterior mean for classical} 
\begin{align}
    \mathbb{E}_\theta[\log \sigma(\theta)] &\coloneqq \int_{\Theta} \log \sigma(\theta) p(\theta|x^n) d\theta, \label{eq:quantum posterior mean}\\
    \mathbb{V}_{\theta}[\log \sigma(\theta)] &\coloneqq \mathbb{E}_\theta[(\log \sigma(\theta))^2] - \mathbb{E}_\theta[\log \sigma(\theta)]^2, \label{eq:quantum posterior variance}
\end{align}
to the quantum setting.
Here, recall that $\sigma(\theta)$ is a quantum statistical model parameterized by $\theta$, and $p(\theta|x^n)$ is the posterior distribution.

\subsection{Basic Bayes theorem}
To study the asymptotic behaviors of $G_n^Q$ and $T_n^Q$, we introduce the \textit{cumulant generating function of the quantum log-likelihood}.
\begin{defn}
\label{def:quantum cumulant}
    Let $\alpha \in \mathbb{R}$. 
    The \textit{cumulant generating function of the quantum log-likelihood} is defined by 
    \begin{equation}
        s^Q(\hat{\rho}, \alpha) \coloneqq \Tr(\hat{\rho} \log \Phi(\alpha)), \quad 
        \Phi(\alpha) \coloneqq \int_\Theta \sigma(\theta)^\alpha p(\theta|x^n) d\theta,
    \end{equation}
    where $\hat{\rho}$ is the classical snapshot of a true quantum state $\rho$.
    
    Similarly, the \textit{cumulant generating function of the quantum log-likelihood ratio} is defined by 
    \begin{equation}
        s^{Q(0)}(\hat{\rho}, \alpha) \coloneqq \Tr(\hat{\rho} \log \Phi^{(0)}(\alpha)), 
    \end{equation}
    where
    \begin{equation}
        \Phi^{(0)}(\alpha) \coloneqq \int_\Theta \left( \int_{0}^{1} \sigma(\theta_0^Q)^{-(1-r)\alpha} \sigma(\theta)^\alpha \sigma(\theta_0^Q)^{-r\alpha} dr \right) p(\theta|x^n) d\theta.
    \end{equation}
\end{defn}
These definitions differ from their classical counterparts due to the non-commutative nature of quantum statistical models. 
The cumulant generating function in the classical case is originally defined as $s(x,\alpha) = \log \mathbb{E}_\theta[p(x|\theta)^\alpha]$ for a measurement outcome $x$ (see Definition \ref{defn:classical_cumulant}).
In our definition for the quantum case, we consider $s^Q(\hat{\rho}_x,\alpha) = \Tr(\hat{\rho}_x \log \mathbb{E}_\theta[\sigma(\theta)^\alpha])$ as its counterpart, where we first raise $\sigma(\theta)$ to the $\alpha$-th power and lastly take the trace after applying the snapshot corresponding to a measurement outcome $x$. 
Note that $s^Q(x,\alpha)$ can be viewed as a quantum generalization of $s(x,\alpha)$ if both $\sigma(\theta)$ and $\hat{\rho}_x$ were classical, i.e., they are diagonal and non-negative matrices with unit trace (although in practice $\hat{\rho}_x$ cannot be a classical state).
Additionally, the cumulant generating function of the log-likelihood ratio function plays a crucial role when studying the asymptotic behavior around the optimal parameter $\theta_0$ \cite[Section 3.3]{watanabe2018mathematical}.
Since the ``ratio'' is not uniquely determined between density operators, we use the form of $\int_0^1 A^{-(1-r)} B A^{-r} dr$ for density operators $A$ and $B$, inspired by the Bogoliubov inner product \cite{amari2000methods}.
We remark that the Bogoliubov Fisher metric can be characterized as the limit of the quantum relative entropy \cite{hayashi2002Two}.
Before proceeding, let us study the first and second cumulants.
\begin{lem}
\label{lem:derivatives of quantum cumulant}
    The first and second derivatives of of $s^Q(\hat{\rho}, \alpha)$ and $s^{Q(0)}(\hat{\rho}, \alpha)$ at $\alpha = 0$ are described using the posterior mean and variance as follows:
    \begin{align}
        &\left. \frac{\partial s^Q(\hat{\rho}, \alpha)}{\partial \alpha} \right|_{\alpha=0} 
        = \Tr( \hat{\rho} \mathbb{E}_\theta[\log \sigma(\theta)] ), \quad
        \left. \frac{\partial^2 s^Q(\hat{\rho}, \alpha)}{\partial \alpha^2} \right|_{\alpha=0} 
        =\Tr(\hat{\rho} \mathbb{V}_{\theta}[\log \sigma(\theta)]),\\
        &\left. \frac{\partial s^{Q(0)}(\hat{\rho}, \alpha)}{\partial \alpha} \right|_{\alpha=0} 
        = \Tr( \hat{\rho} \left\{ \mathbb{E}_\theta[\log \sigma(\theta)] - \log \sigma(\theta_0^Q) \right\} ), \quad
        \left. \frac{\partial^2 s^{Q(0)}(\hat{\rho}, \alpha)}{\partial \alpha^2} \right|_{\alpha=0} 
        =\Tr(\hat{\rho} \mathbb{V}_{\theta}[\log \sigma(\theta)]).
    \end{align}
\end{lem}
\begin{proof} 
    The claim follows from straightforward calculations using matrix calculus and properties of the trace.
    First, recall the relation
    \begin{align*}
        \frac{\partial}{\partial \alpha} \sigma(\theta)^\alpha &= \frac{\partial}{\partial \alpha} \exp( \alpha \log \sigma(\theta) ) 
        = \exp( \alpha \log \sigma(\theta) ) \frac{\partial}{\partial \alpha}( \alpha \log \sigma(\theta) ) 
        = \sigma(\theta)^\alpha \log \sigma(\theta)
    \end{align*}
    for $\sigma(\theta)^\alpha = \exp( \alpha \log \sigma(\theta) )$,
    following from the fact that $\exp(Ax)$ and $A$ are commutative for all $A \in \mathbb{C}^{n \times n}$ and $x\in\mathbb{C}$ in the second equation.
    Then, the first derivative of $\Phi(\alpha)$ is 
    \begin{align}
        \frac{\partial \Phi(\alpha)}{\partial \alpha} = \int_\Theta \sigma(\theta)^\alpha \log \sigma(\theta) p(\theta|x^n) d\theta, \quad 
        \left. \frac{\partial \Phi(\alpha)}{\partial \alpha} \right|_{\alpha=0} = \mathbb{E}_\theta[\log \sigma(\theta)],
    \end{align}
    and the first derivative of $\Phi^{(0)}(\alpha)$ and its value are given by 
    \begin{align*}
        \frac{\partial \Phi^{(0)}(\alpha)}{\partial \alpha} &= \int_\Theta \left[ \int_{0}^{1} \left\{ \left(-(1-r) \sigma(\theta_0^Q)^{-(1-r)\alpha} \log \sigma(\theta_0^Q) \sigma(\theta)^\alpha \sigma(\theta_0^Q)^{-r\alpha}\right) \right.\right. \\
        &\quad \left.\left. + \left( \sigma(\theta_0^Q)^{-(1-r)\alpha} \sigma(\theta)^\alpha \log \sigma(\theta) \sigma(\theta_0^Q)^{-r\alpha} \right) \right.\right. \\
        &\quad \left.\left. + \left( -r \sigma(\theta_0^Q)^{-(1-r)\alpha} \sigma(\theta)^\alpha \sigma(\theta_0^Q)^{-r\alpha} \log \sigma(\theta_0^Q) \right) \right\} dr \right] p(\theta|x^n) d\theta, \\ 
        \left. \frac{\partial \Phi^{(0)}(\alpha)}{\partial \alpha} \right|_{\alpha=0} &= \int_\Theta \left[ \int_{0}^{1} \left\{ -(1-r) \log \sigma(\theta_0^Q) + \log \sigma(\theta) -r \log \sigma(\theta_0^Q) \right\} dr \right] p(\theta|x^n) d\theta \\
        &= \mathbb{E}_\theta[\log \sigma(\theta)] - \log \sigma(\theta_0^Q).
    \end{align*}
    Similarly, we compute the second derivatives of $\Phi(\alpha)$ and $\Phi^{(0)}(\alpha)$:
    \begin{align}
        \left. \frac{\partial^2 \Phi(\alpha)}{\partial \alpha^2} \right|_{\alpha=0} = \mathbb{E}_\theta[ (\log \sigma(\theta))^2 ], \quad 
        \left. \frac{\partial^2 \Phi^{(0)}(\alpha)}{\partial \alpha^2} \right|_{\alpha=0} = \mathbb{E}_\theta[ (\log \sigma(\theta) - \log \sigma(\theta_0^Q))^2 ], 
    \end{align}
    respectively.
    Before calculating the derivatives of $s^Q(\hat{\rho}, \alpha)$ and $s^{Q(0)}(\hat{\rho}, \alpha)$,
    we recall the basic resolvent formulas
    \begin{align}\label{eq:resolvent}
        (sI+Q)^{-1} - (sI+P)^{-1} &=  (sI+P)^{-1} (P-Q)  (sI+Q)^{-1},\\
        \log P - \log Q &= \int_{0}^{\infty} (sI+P)^{-1} (P-Q) (sI+Q)^{-1} ds,\label{eq:resolvent_int}
    \end{align}
    for all positive Hermitian matrices $P, Q > 0$; see \cite{haber2023Notes}.
    Then, the first derivative of $s^Q(\hat{\rho}, \alpha)$ is 
    \begin{align*}
        \frac{\partial s^Q(\hat{\rho}, \alpha)}{\partial \alpha} &= \Tr(\hat{\rho} \frac{\partial}{\partial \alpha} \log \Phi(\alpha)) \\ 
        &= \lim_{\epsilon \rightarrow 0} \Tr(\hat{\rho} \frac{\log \Phi(\alpha+\epsilon) - \log \Phi(\alpha)}{\epsilon}) \\
        &= \lim_{\epsilon \rightarrow 0} \Tr(\hat{\rho} \int_{0}^{\infty} (sI+\Phi(\alpha+\epsilon))^{-1} \frac{\Phi(\alpha+\epsilon)-\Phi(\alpha)}{\epsilon} (sI+\Phi(\alpha))^{-1} ds ) \\
        &= \Tr(\hat{\rho} \int_{0}^{\infty} (sI+\Phi(\alpha))^{-1} \frac{\partial \Phi(\alpha)}{\partial \alpha} (sI+\Phi(\alpha))^{-1} ds ). \\
    \end{align*}
    The third equation can be obtained using Eq. \eqref{eq:resolvent_int}.
    Plugging $\alpha=0$ into the above, the first cumulant is calculated as 
    \begin{align*}
        \left. \frac{\partial s^Q(\hat{\rho}, \alpha)}{\partial \alpha} \right|_{\alpha=0} &= \Tr(\hat{\rho} \int_{0}^{\infty} (sI+I)^{-1} \left.\frac{\partial \Phi(\alpha)}{\partial \alpha}\right|_{\alpha=0} (sI+I)^{-1} ds ) \\ 
        &= \Tr(\hat{\rho} \int_{0}^{\infty} \left(\frac{1}{s+1} I\right) \left.\frac{\partial \Phi(\alpha)}{\partial \alpha}\right|_{\alpha=0} \left(\frac{1}{s+1} I\right) ds ) \\
        &= \int_{0}^{\infty} \left(\frac{1}{s+1}\right)^2 ds \cdot \Tr( \hat{\rho} \left.\frac{\partial \Phi(\alpha)}{\partial \alpha}\right|_{\alpha=0} ) \\
        &= \Tr( \hat{\rho} \left.\frac{\partial \Phi(\alpha)}{\partial \alpha}\right|_{\alpha=0} )
        = \Tr( \hat{\rho} \mathbb{E}_\theta[\log \sigma(\theta)] ).
    \end{align*}
    A similar calculation allows us to write the second derivative of $s^Q(\hat{\rho}, \alpha)$: 
    \begin{align*}
        \frac{\partial^2 s^Q(\hat{\rho}, \alpha)}{\partial \alpha^2} 
        &= \frac{\partial}{\partial \alpha} \Tr(\hat{\rho} \int_{0}^{\infty} (sI+\Phi(\alpha))^{-1} \frac{\partial \Phi(\alpha)}{\partial \alpha} (sI+\Phi(\alpha))^{-1} ds ) \\
        &= \Tr(\hat{\rho} \int_{0}^{\infty} \frac{\partial}{\partial \alpha}(sI+\Phi(\alpha))^{-1} \frac{\partial \Phi(\alpha)}{\partial \alpha} (sI+\Phi(\alpha))^{-1} ds ) \nonumber\\
        &\quad + \Tr(\hat{\rho} \int_{0}^{\infty} (sI+\Phi(\alpha))^{-1} \frac{\partial^2 \Phi(\alpha)}{\partial \alpha^2} (sI+\Phi(\alpha))^{-1} ds ) \nonumber\\
        &\quad + \Tr(\hat{\rho} \int_{0}^{\infty} (sI+\Phi(\alpha))^{-1} \frac{\partial \Phi(\alpha)}{\partial \alpha} \frac{\partial}{\partial \alpha}(sI+\Phi(\alpha))^{-1} ds ),
    \end{align*}
    where the derivative $\frac{\partial}{\partial \alpha}(sI+\Phi(\alpha))^{-1}$ is calculated as 
    \begin{align}
        \frac{\partial}{\partial \alpha}(sI+\Phi(\alpha))^{-1}
        &= \lim_{\epsilon\rightarrow0} \frac{(sI+\Phi(\alpha+\epsilon))^{-1} - (sI+\Phi(\alpha))^{-1}}{\epsilon} \\
        &= \lim_{\epsilon\rightarrow0} (sI+\Phi(\alpha))^{-1} \frac{\Phi(\alpha) - \Phi(\alpha+\epsilon)}{\epsilon} (sI+\Phi(\alpha+\epsilon))^{-1} \\
        &= (sI+\Phi(\alpha))^{-1} \left(-\frac{\partial \Phi(\alpha)}{\partial \alpha}\right) (sI+\Phi(\alpha))^{-1},
        \label{eq:derivative_inv_sI_plus_sigma}
    \end{align}
    by using Eq. \eqref{eq:resolvent} in the second equation.
    This yields the second cumulant as
    \begin{align*}
        \left. \frac{\partial^2 s^Q(\hat{\rho}, \alpha)}{\partial \alpha^2} \right|_{\alpha=0} 
        &= \int_{0}^{\infty} \left(\frac{1}{s+1}\right)^2 ds \cdot \Tr(\hat{\rho} \left.\frac{\partial^2 \Phi(\alpha)}{\partial \alpha^2}\right|_{\alpha=0} )  - 2 \int_{0}^{\infty} \left(\frac{1}{s+1}\right)^3 ds \cdot \Tr(\hat{\rho} \left.\frac{\partial \Phi(\alpha)}{\partial \alpha}\right|_{\alpha=0} \left.\frac{\partial \Phi(\alpha)}{\partial \alpha}\right|_{\alpha=0} ) \\
        &= \Tr(\hat{\rho} \left.\frac{\partial^2 \Phi(\alpha)}{\partial \alpha^2}\right|_{\alpha=0} ) - \Tr(\hat{\rho} \left.\frac{\partial \Phi(\alpha)}{\partial \alpha}\right|_{\alpha=0} \left.\frac{\partial \Phi(\alpha)}{\partial \alpha}\right|_{\alpha=0}) \\
        &= \Tr(\hat{\rho} \left\{ \mathbb{E}_\theta[\log \sigma(\theta)^2] - \mathbb{E}_\theta[\log \sigma(\theta)]^2 \right\} ) 
        = \Tr(\hat{\rho} \mathbb{V}_{\theta}[\log \sigma(\theta)]).
    \end{align*}
    The same method applies to $s^{Q(0)}(\hat{\rho}, \alpha)$. It yields
    \begin{align}
        \left. \frac{\partial s^{Q(0)}(\hat{\rho}, \alpha)}{\partial \alpha} \right|_{\alpha=0} &= \Tr( \hat{\rho} \left.\frac{\partial \Phi^{(0)}(\alpha)}{\partial \alpha}\right|_{\alpha=0} )
        = \Tr( \hat{\rho} \left\{ \mathbb{E}_\theta[\log \sigma(\theta)] - \log \sigma(\theta_0^Q) \right\} ), \\
        \left. \frac{\partial^2 s^{Q(0)}(\hat{\rho}, \alpha)}{\partial \alpha^2} \right|_{\alpha=0} 
        &= \Tr(\hat{\rho} \left.\frac{\partial^2 \Phi^{(0)}(\alpha)}{\partial \alpha^2}\right|_{\alpha=0} ) - \Tr(\hat{\rho} \left.\frac{\partial \Phi^{(0)}(\alpha)}{\partial \alpha}\right|_{\alpha=0} \left.\frac{\partial \Phi^{(0)}(\alpha)}{\partial \alpha}\right|_{\alpha=0}) \\
        &= \Tr(\hat{\rho} \left[ \mathbb{E}_\theta[ \{\log \sigma(\theta) - \log \sigma(\theta_0^Q)\}^2 ] - \{ \mathbb{E}_\theta[\log \sigma(\theta)] - \log \sigma(\theta_0^Q) \}^2 \right] ) \\
        &= \Tr(\hat{\rho} \left\{ \mathbb{E}_\theta[\log \sigma(\theta)^2] - \mathbb{E}_\theta[\log \sigma(\theta)]^2 \right\} ) 
        = \Tr(\hat{\rho} \mathbb{V}_{\theta}[\log \sigma(\theta)]).
    \end{align}
    It concludes the proof.
\end{proof}
We now state the basic theorem of Bayesian statistics in quantum state estimation, which can be regarded as a quantum analog of Proposition \ref{prop:generalization loss $G_n$ and the training loss $T_n$ for general cases}.
\begin{thm}[Basic theorem]\label{thm:q_basic}
    Assume
    \begin{align}
        \left| \mathbb{E}_X[\partial_\alpha^{\ell} s^Q(\hat{\rho}_X,\alpha) |_{\alpha=0}] \right| &\leq O_p\left(\frac{1}{n^{\ell/2}}\right), \\
        \left| \frac{1}{n} \sum_{i=1}^{n} \partial_\alpha^{\ell} s^Q(\hat{\rho}_{X_i},\alpha) |_{\alpha=0} \right| &\leq O_p\left(\frac{1}{n^{\ell/2}}\right),
    \end{align}
    for $\ell \geq 3$.
    Then, the generalization loss $G_n^{Q}$ and training loss $T_n^{Q}$ can be expanded as
    \begin{align}
        G_n^Q &= - \Tr(\rho \mathbb{E}_\theta[\log \sigma(\theta)]) - \frac{1}{2} \Tr(\rho \mathbb{V}_{\theta}[\log \sigma(\theta)]) + o_p\left(\frac{1}{n}\right),
        \label{eq:G_n^Q_expanded} \\
        T_n^Q &= - \Tr( \left(\frac{1}{n}\sum_{i=1}^{n} \hat{\rho}_{x_i}\right) \mathbb{E}_\theta[\log \sigma(\theta)])
        - \frac{1}{2} \Tr(\left(\frac{1}{n}\sum_{i=1}^{n} \hat{\rho}_{x_i}\right) \mathbb{V}_{\theta}[\log \sigma(\theta)]) + o_p\left(\frac{1}{n}\right).
        \label{eq:T_n^Q_expanded}
    \end{align}
\end{thm}
\begin{proof} 
    The Taylor expansion of $s^Q(\hat{\rho},\alpha)$ around $\alpha = 0$ gives
    \begin{equation*}
        s^Q(\hat{\rho}, \alpha) = s^Q(\hat{\rho},0) + \partial_\alpha s^Q(\hat{\rho},0) \alpha + \frac{1}{2} \partial_{\alpha}^2 s^Q(\hat{\rho},0) \alpha^2 + \sum_{\ell=3}^{\infty} \frac{1}{\ell!} \partial_{\alpha}^{\ell} s^Q(\hat{\rho},0) \alpha^{\ell}.
    \end{equation*}
    Combined with Lemma \ref{lem:derivatives of quantum cumulant}, plugging $\alpha=1$ and taking the expectation (or empirical sum) of both sides of the above equation yields 
    \begin{align}
        - G_n^Q &= \mathbb{E}_X[s^Q(\hat{\rho}_X,1)] \\
        &= \mathbb{E}_X[s^Q(\hat{\rho}_X,0)] + \mathbb{E}_X[\partial_\alpha s^Q(\hat{\rho}_X,0)] + \frac{1}{2} \mathbb{E}_X[\partial_{\alpha}^2 s^Q(\hat{\rho}_X,0)] + \sum_{\ell=3}^{\infty} \frac{1}{\ell!} \mathbb{E}_X[\partial_{\alpha}^{\ell} s^Q(\hat{\rho}_X,0)] \\
        &= \Tr(\rho \mathbb{E}_\theta[\log \sigma(\theta)]) + \frac{1}{2} \Tr(\rho \mathbb{V}_{\theta}[\log \sigma(\theta)]) + \sum_{\ell=3}^{\infty} \frac{1}{\ell!} \mathbb{E}[\partial_{\alpha}^{\ell} s^Q(\hat{\rho}_X,0)], \\
        - T_n^Q &= \frac{1}{n} \sum_{i=1}^{n} s^Q(\hat{\rho}_{x_i},1) \\
        &= \frac{1}{n} \sum_{i=1}^{n}s^Q(\hat{\rho}_{x_i},0) + \frac{1}{n} \sum_{i=1}^{n} \partial_\alpha s^Q(\hat{\rho}_{x_i},0) + \frac{1}{2} \left(\frac{1}{n} \sum_{i=1}^{n} \partial_{\alpha}^2 s^Q(\hat{\rho}_{x_i},0)\right)+ \sum_{\ell=3}^{\infty} \frac{1}{\ell!} \left(\frac{1}{n} \sum_{i=1}^{n} \partial_{\alpha}^{\ell} s^Q(\hat{\rho}_{x_i},0)\right) \\
        &= \Tr(\left(\frac{1}{n}\sum_{i=1}^{n} \hat{\rho}_{x_i}\right) \mathbb{E}_\theta[\log \sigma(\theta)]) + \frac{1}{2} \Tr(\left(\frac{1}{n}\sum_{i=1}^{n} \hat{\rho}_{x_i}\right) \mathbb{V}_{\theta}[\log \sigma(\theta)])
        +  \sum_{\ell=3}^{\infty} \frac{1}{\ell!} \left(\frac{1}{n} \sum_{i=1}^{n} \partial_{\alpha}^{\ell} s^Q(\hat{\rho}_{x_i},0)\right),
    \end{align}
    respectively.
    Then, the assumption in this theorem ensures that the residual terms for $\ell \geq 3$ have an order $o(1/n)$, which completes the proof.
\end{proof}
Highlighting the significance of this theorem is crucial.
It is not a trivial consequence of the classical version (Proposition \ref{prop:generalization loss $G_n$ and the training loss $T_n$ for general cases}) because of the non-commutativity.
As a result of the appropriate definition of the cumulant generating function (Definition \ref{def:quantum cumulant}), it becomes apparent that the quantum generalization and training loss can also be represented by their first and second cumulants up to the order $o_p(1/n)$.
The assumptions in Theorem \ref{thm:q_basic} are verified in Lemma \ref{lem:higher_order_scaling_regular} for regular cases and in Lemma \ref{lem:higher_order_scaling_singular} for singular cases. 
A key calculation necessary for both lemmas involves bounding the higher-order derivatives.

Below we introduce a shorthand notation
\begin{equation}
\label{eq:large F}
    F(\theta, \theta_0^Q) \coloneqq \log \sigma(\theta) - \log \sigma(\theta_0^Q).
\end{equation}
\begin{lem}
\label{lem:quantum higher order}
    For $\ell \geq 3$, the following inequalities hold:
    \begin{align}
        \left| \mathbb{E}_X[\partial_\alpha^{\ell} s^Q(\hat{\rho}_X,\alpha) |_{\alpha=0}] \right| &\leq \left|\mathbb{E}_\theta\left[\Tr( \rho|F(\theta, \theta_0^Q)|^{\ell} )\right]\right|, 
        \label{eq:quantum_higher_order_G} \\
        \left| \frac{1}{n} \sum_{i=1}^{n} \partial_\alpha^{\ell} s^Q(\hat{\rho}_{X_i},\alpha) |_{\alpha=0} \right| &\leq \left|\mathbb{E}_\theta\left[\Tr( \left( \frac{1}{n} \sum_{i=1}^{n} \hat{\rho}_{X_i} \right) |F(\theta, \theta_0^Q)|^{\ell} )\right]\right|.
        \label{eq:quantum_higher_order_T}
    \end{align}
\end{lem}
\begin{proof}
     Below we illustrate the proof for $\ell = 3$; the same method can be applied for $\ell \geq 4$.
     We define the following functions temporarily for this proof:
    \begin{align}
        s_{\ell}^Q(\hat{\rho}, 0) &\coloneqq \Tr(\hat{\rho} \mathbb{E}_\theta[F(\theta, \theta_0^Q)^{\ell}]), \quad
        s_{(\ell_1,\ell_2,...)}^Q(\hat{\rho}, 0) \coloneqq \Tr(\hat{\rho} \mathbb{E}_\theta[F(\theta, \theta_0^Q)^{\ell_1}] \mathbb{E}_\theta[F(\theta, \theta_0^Q)^{\ell_2}] \cdots), \\
        s_{\ell}^{Q*}(\hat{\rho}, 0) &\coloneqq \Tr(\hat{\rho} \mathbb{E}_\theta[|F(\theta, \theta_0^Q)|^{\ell}]), \quad 
        s_{(\ell_1,\ell_2,...)}^{Q*}(\hat{\rho}, 0) \coloneqq \Tr(\hat{\rho} \mathbb{E}_\theta[|F(\theta, \theta_0^Q)|^{\ell_1}] \mathbb{E}_\theta[|F(\theta, \theta_0^Q)|^{\ell_2}] \cdots).
    \end{align}
    One can check that $F(\theta, \theta_0^Q)$ is Hermitian and $|F(\theta, \theta_0^Q)|$ is positive semi-definite in general.
    Using H\"older's inequality for a Hermitian matrix-valued random variable \cite[Remark 1]{leorato2017Note}, the following inequalities hold for $1<i<j<\infty$:
    \begin{align}
        \mathbb{E}_\theta[|F(\theta, \theta_0^Q)|^i] \leq (\mathbb{E}_\theta[|F(\theta, \theta_0^Q)|^{i+j}])^\frac{i}{i+j},\quad 
        \mathbb{E}_\theta[|F(\theta, \theta_0^Q)|^j] \leq (\mathbb{E}_\theta[|F(\theta, \theta_0^Q)|^{i+j}])^\frac{j}{i+j}. 
    \end{align}
    Then it follows that 
    \begin{align}
        \Tr(A \mathbb{E}_\theta[|F(\theta, \theta_0^Q)|^i] \mathbb{E}_\theta[|F(\theta, \theta_0^Q)|^j]) \leq \Tr(A \mathbb{E}_\theta[|F(\theta, \theta_0^Q)|^{i+j}])
        \label{eq:hoelder_matrix_valued_i_j}
    \end{align}
    for any positive semi-definite matrix $A$.

    In addition, recalling the calculation in Lemma \ref{lem:derivatives of quantum cumulant}, we can rewrite the $\ell$-th ($\ell \geq 1$) derivative of $s^Q(\hat{\rho},\alpha)$ as 
    \begin{equation}
        \left. \partial_\alpha^\ell s^Q(\hat{\rho},\alpha) \right|_{\alpha=0} = \Tr(\hat{\rho} \int_{0}^{\infty} \left. \partial_\alpha^{\ell-1} Y_1\right|_{\alpha=0} ds ),
    \end{equation}
    putting the shorthand notation
    \begin{equation}
        Y_\ell \coloneqq (sI+\Phi(\alpha))^{-1} (\partial_\alpha^\ell \Phi(\alpha)) (sI+\Phi(\alpha))^{-1}.
    \end{equation}
    Then, the following relation holds for $\ell \geq 1$:
    \begin{equation}
        \partial_\alpha Y_\ell = Y_{\ell+1} - (sI+\Phi(\alpha))^{-1} (\partial_\alpha \Phi(\alpha)) Y_\ell - Y_\ell (\partial_\alpha \Phi(\alpha)) (sI+\Phi(\alpha))^{-1}.
    \end{equation}
    Note that the same relation holds for the $\ell$-th derivative of $s^{Q(0)}(\hat{\rho},\alpha)$.
    Through a simple calculation based on the above equations, we can obtain the third derivative of $s^Q(\hat{\rho},\alpha)$ at $\alpha=0$ as
    \begin{align}
        \partial_\alpha^3 s^Q(\hat{\rho},0) = \partial_\alpha^3 s^{Q(0)}(\hat{\rho},0) = s_3^Q(\hat{\rho},0) - \frac{3}{2}\left\{s_{(1,2)}^Q(\hat{\rho},0) + s_{(2,1)}^Q(\hat{\rho},0)\right\} + 2 s_{(1,1,1)}^Q(\hat{\rho},0).
    \end{align}
    Then, to bound the third-order term of $G_n^Q$, we take the expectation of the above and apply the triangle inequality to obtain 
    \begin{align}
        |\mathbb{E}_X[\partial_\alpha^3 s^Q(\hat{\rho}_X,\alpha) |_{\alpha=0}]| &= |\partial_\alpha^3 s^Q(\rho,0)| \\
        &\leq |s_3^Q(\rho,0)| + \frac{3}{2}\left\{|s_{(1,2)}^Q(\rho,0)| + |s_{(2,1)}^Q(\rho,0)|\right\} + 2 |s_{(1,1,1)}^Q(\rho,0)| \\
        &\leq |s_3^{Q*}(\rho,0)| + \frac{3}{2}\left\{|s_{(1,2)}^{Q*}(\rho,0)| + |s_{(2,1)}^{Q*}(\rho,0)|\right\} + 2 |s_{(1,1,1)}^{Q*}(\rho,0)|
        \label{eq:partial_alpha^3_s^Q_2nd} \\
        &\leq 6 |s_3^{Q*}(\rho,0)| = 6 \left|\mathbb{E}_\theta\left[\Tr( \rho|F(\theta, \theta_0^Q)|^3 )\right]\right|.
        \label{eq:partial_alpha^3_s^Q_3rd}
    \end{align}
    The second inequality follows the fact that $\Tr(AB) \leq \Tr(A|B|)$ for a positive semi-definite matrix $A$ and a Hermitian matrix $B$.
    The last inequality follows from Eq. \eqref{eq:hoelder_matrix_valued_i_j}.
    To bound the third-order term of $T_n^Q$, we follow the same calculation as $G_n^Q$ except that $\partial_\alpha^3 s^Q(\rho,0)$ is replaced with $\partial_\alpha^3 s^Q((1/n)\sum_{i=1}^{n} \hat{\rho}_{x_i},0)$.
    We note that $(1/n) \sum_{i=1}^{n} \hat{\rho}_{x_i}$ is not generally positive semi-definite. 
    Thus, the above second and third inequalities do not directly hold for $T_n^Q$.
    However, in the asymptotic limit, that is, $n$ is much larger than the size of $\rho$, the matrix $\frac{1}{n}\sum_{i=1}^{n} \hat{\rho}_{x_i}$ is positive semi-definite, and thus the same approach can be applied.
    This completes the proof for $\ell=3$.
    The same method can be applied for the case of $\ell \geq 4$.
\end{proof}

We now ready to analyze the asymptotic behaviors of the generalization loss $G_n^Q$ and $T_n^Q$.
In the remainder of the paper, we will discuss two cases, regular and singular, separately.
\subsection{Regular cases}
In this subsection, we consider regular cases, where quantum models satisfy the regularity condition (Assumption \ref{ass:R1} and \ref{ass:R2}). 
This assumption allows us to utilize the benefits coming from the classical and quantum statistics in view of the Fisher information matrix.

First, we estimate the growth order of the cumulants, which proves the assumption made in Theorem \ref{thm:q_basic}.
Interestingly, the proof is derived from the properties of the classical shadow, a recently developed method in quantum state estimation \cite{huang2020Predicting}.
It is a newly emerged phenomenon in the quantum setting.
\begin{lem}[Higher order scaling for regular cases]
\label{lem:higher_order_scaling_regular}
    Suppose that Assumptions \ref{ass:R1} and \ref{ass:R2} are satisfied.
    For $\ell \geq 3$, the higher order cumulants satisfy:
    \begin{align}
        \left| \mathbb{E}_X[\partial_\alpha^{\ell} s^Q(\hat{\rho}_X,\alpha) |_{\alpha=0}] \right| &\leq O_p\left(\frac{1}{n^{\ell/2}}\right) \\
        \left| \frac{1}{n} \sum_{i=1}^{n} \partial_\alpha^{\ell} s^Q(\hat{\rho}_{X_i},\alpha) |_{\alpha=0} \right| &\leq O_p\left(\frac{1}{n^{\ell/2}}\right).
    \end{align} 
\end{lem}
\begin{proof}
    Let us write the Schmidt decomposition of $F(\theta, \theta_0^Q)$ as
    \begin{align}
        F(\theta, \theta_0^Q) = \log \sigma(\theta_0^Q) - \log \sigma(\theta) = \sum_{j=1}^{D} \lambda_j \ketbra*{v_j}{v_j}, \quad \lambda_j \in \mathbb{R},
        \label{eq:Schmidt_decomp_F}
    \end{align}
    with the orthonormal basis $\{v_1, ..., v_D\}$.
    Then, using
    \begin{equation}
        \hat{\rho}_{v_i} \coloneqq (D+1) \ketbra*{v_i}{v_i} - I_D, \quad i = 1, ..., D,
        \label{eq:snapshot_Haar}
    \end{equation}
    we obtain 
    \begin{equation}
        f^Q(\hat{\rho}_{v_i}, \theta) = \Tr(\hat{\rho}_{v_i} F(\theta, \theta_0^Q)) = (D+1) \lambda_i - \sum_{j=1}^D \lambda_j, \quad i = 1,...,D,
        \label{eq:eigendcomp_f^Q}
    \end{equation}
    ($\hat{\rho}_{v_i}$ is the temporary notation only used in this Lemma and Lemma \ref{lem:higher_order_scaling_singular}).
    It also follows from the mean-value theorem and the relation $\theta_0^Q = \theta_0$ that
    \begin{align}
        f^Q(\hat{\rho}_{v_i},\theta) = (\theta - \theta_0^Q) \nabla f^Q(\hat{\rho}_{v_i}, \theta_1^{(i)}) = (\theta - \theta_0) \nabla f^Q(\hat{\rho}_{v_i}, \theta_1^{(i)}), \quad i = 1,...,D,
        \label{eq:mvt_f_Q}
    \end{align}
    where $\theta_1^{(i)}$ is a point between $\theta$ and $\theta_0$.
    From Eqs. \eqref{eq:eigendcomp_f^Q} and \eqref{eq:mvt_f_Q}, we have
    \begin{equation}
    \label{eq:relation of lambda}
         (D+1) \lambda_i - \sum_{j=1}^D \lambda_j = (\theta - \theta_0) \nabla f^Q(\hat{\rho}_{v_i}, \theta_1^{(i)}), \quad i = 1,...,D.
    \end{equation}
    The system of the above $D$ linear equations allows us to obtain
    \begin{align}
        \lambda_i = \frac{\theta-\theta_0}{D+1} \left( \nabla f^Q(\hat{\rho}_{v_i}, \theta_1^{(i)}) + \sum_{j=1}^{D} \nabla f^Q(\hat{\rho}_{v_j}, \theta_1^{(j)}) \right), \quad i = 1, ..., D.
    \end{align}
    Note that the argument of $f^Q(\cdot,\theta)$ does not necessarily have to be $\hat{\rho}_{v_i}$ \eqref{eq:snapshot_Haar} (in fact, $\ketbra*{v_i}{v_i}$ induces a simpler solution for $\lambda_i$).
    Nonetheless, we will also utilize Eq. \eqref{eq:eigendcomp_f^Q} for singular cases, where we regard $\hat{\rho}$ as a classical snapshot of $\rho$ with Haar random unitaries.
    From these equations, the right-hand side of Eq. \eqref{eq:quantum_higher_order_G} is bounded as follows:
    \begin{align}
        \left|\mathbb{E}_\theta\left[\Tr( \rho|F(\theta, \theta_0^Q)|^{\ell} )\right]\right|
        &= \left| \mathbb{E}_\theta \left[ \sum_{i=1}^{D} |\lambda_i|^{\ell} \Tr(\rho \ketbra*{v_i}{v_i}) \right] \right| \\
        &\leq \left|\mathbb{E}_\theta\left[ \left| \frac{\theta-\theta_0}{D+1} \right|^{\ell} \left( \sum_{i=1}^{D} \left| \nabla f^Q(\hat{\rho}_{v_i}, \theta_1^{(i)}) + \sum_{j=1}^{D} \nabla f^Q(\hat{\rho}_{v_j}, \theta_1^{(j)}) \right|^{\ell} \Tr(\rho \ketbra*{v_i}{v_i}) \right) \right]\right| \\
        &\leq \left| \mathbb{E}_\theta\left[ \left| \frac{\theta-\theta_0}{D+1} \right|^{\ell} \right] \left( \sum_{i=1}^{D} \sup_{\theta_1} \left| \nabla f^Q(\hat{\rho}_{v_i}, \theta_1) + \sum_{j=1}^{D} \nabla f^Q(\hat{\rho}_{v_j}, \theta_1) \right|^{\ell} \Tr(\rho \ketbra*{v_i}{v_i}) \right) \right| \\ 
        &= O(n^{-\ell/2}).
    \end{align}
    The last equality follows from \cite[Lemma 12]{watanabe2018mathematical} for the evaluation of $\mathbb{E}_\theta[|\theta-\theta_0|^{\ell}]$.    
    A similar calculation yields bounding the right-hand side of Eq. \eqref{eq:quantum_higher_order_T} by $O(n^{-\ell/2})$.
    Plugging these bounds into Eqs. \eqref{eq:quantum_higher_order_G} and \eqref{eq:quantum_higher_order_T} completes the proof. 
\end{proof}
The above proof depends on the analysis of the properties of the analysis of the properties of operators on Hilbert spaces.
As stated in the proof, we considered the behaviors of the cumulants in a way that does not depend on the choice of classical snapshots.
While this is sufficient for our purposes, it would be interesting to discuss in more detail, i.e., whether the best measurement can be selected. Addressing this question would require adaptive estimation theory and could provide new insights, even in regular cases.

Next, we describe the asymptotic behavior of the losses in statistical inference for quantum regular models.
Similarly to Proposition \ref{prop:generalization loss of G_n and T_n}, it can be proven to utilize the second-order Taylor expansion, the posterior integration (Lemma \ref{lem:gb_lem14}), and the empirical process \eqref{eq:def_eta_n^Q}.
\begin{thm}\label{thm:q_regular expansion formula for G and T}
    Suppose that Assumptions \ref{ass:R1} and \ref{ass:R2} are satisfied.
    Then, the generalization loss $G_n^Q$ and training loss $T_n^Q$ can be expanded as follows:
    \begin{align}
        G_n^Q &= - \Tr(\rho \log \sigma(\theta_0)) + \frac{1}{2n} \Tr(J^Q J^{-1}) + \frac{1}{2} \Delta_n^T J^Q \Delta_n - \frac{1}{2n} \Tr(I^Q J^{-1}) + o_p\left(\frac{1}{n}\right),
        \label{eq:expn_G_n^Q}\\
        T_n^Q &= - \Tr(\left(\frac{1}{n}\sum_{i=1}^{n} \hat{\rho}_{x_i}\right) \log \sigma(\theta_0)) + \frac{1}{2n} \Tr(J^Q J^{-1}) + \frac{1}{2} \Delta_n^T J^Q \Delta_n - \frac{1}{2}(\Delta_n^{QT} J^Q \Delta_n + \Delta_n^T J^Q \Delta_n^Q) - \frac{1}{2n} \Tr(I^Q J^{-1}) + o_p\left(\frac{1}{n}\right)
        \label{eq:expn_T_n^Q}
    \end{align}
    where $I^Q$ is defined in Definition \ref{defn:matrix I and J_quantum}.
\end{thm}
\begin{proof}
    \noindent\textit{Expansion of $G_n^Q$: }
    We further expand Eq. \eqref{eq:G_n^Q_expanded}.
    First, the Taylor expansion of $\log \sigma(\theta)$ around $\theta_0$ gives
    \begin{align}
        \log \sigma(\theta)
        &= \log \sigma(\theta_0) + (\theta - \theta_0)^T \nabla \log \sigma(\theta_0) + o_p(\| \theta - \theta_0 \|)
        \label{eq:expn_log_sigma_1st}\\
        &= \log \sigma(\theta_0) + (\theta - \theta_0)^T \nabla \log \sigma(\theta_0) + \frac{1}{2} (\theta - \theta_0)^T \nabla^2 \log \sigma(\theta_0) (\theta - \theta_0) + o_p(\| \theta - \theta_0 \|^2)
        \label{eq:expn_log_sigma_2nd}
    \end{align}
    where the first (resp. second) equation is the expansion to the first (resp, second) order.
    We can use Eq. \eqref{eq:expn_log_sigma_2nd} to find the expansion of the first term of Eq. \eqref{eq:G_n^Q_expanded} as
    \begin{align}
        &- \Tr(\rho \mathbb{E}_\theta[\log \sigma(\theta)])\\ 
        &= - \Tr(\rho \log \sigma(\theta_0)) - \mathbb{E}_\theta[\theta - \theta_0]^T \Tr( \rho \nabla \log \sigma(\theta_0) ) - \frac{1}{2} \mathbb{E}_\theta[ (\theta - \theta_0)^T \Tr( \rho \nabla^2 \log \sigma(\theta_0) ) (\theta - \theta_0) ] + o_p\left(\frac{1}{n}\right) \\
        &= - \Tr(\rho \log \sigma(\theta_0)) - \frac{1}{2} \Tr( \left( \frac{J^{-1}}{n} + \Delta_n \Delta_n^T \right) J^Q ) + o_p\left(\frac{1}{n}\right) \\
        &= - \Tr(\rho \log \sigma(\theta_0)) + \frac{1}{2n} \Tr(J^Q J^{-1}) + \frac{1}{2} \Delta_n^T J^Q \Delta_n + o_p\left(\frac{1}{n}\right).
        \label{eq:reg_expn_G_n^Q_first}
    \end{align}
    The second equality follows the facts $\theta_0^Q = \theta_0$, $\nabla \Tr(\rho \log \sigma(\theta_0^Q)) = \bm0$, $\Tr( \rho \nabla^2 \log \sigma(\theta_0) ) = J^Q$, and Lemma \ref{lem:gb_lem14}. 
    Next, we work on the second term of Eq. \eqref{eq:G_n^Q_expanded}:
    \begin{equation*}
        \Tr(\rho \mathbb{V}_{\theta}[\log \sigma(\theta)]) =  
        \Tr(\rho \{ \mathbb{E}_\theta[(\log \sigma(\theta))^2] - \mathbb{E}_\theta[\log \sigma(\theta)]^2 \}).
    \end{equation*}
    We use Eq. \eqref{eq:expn_log_sigma_1st} to obtain
    \begin{align}
        \mathbb{E}_\theta[\log \sigma(\theta) - \log \sigma(\theta_0)] 
        &= \mathbb{E}_\theta\left[ (\theta - \theta_0)^T \left( \nabla \log \sigma(\theta_0) + o_p(1) \right) \right] \notag\\
        &= \mathbb{E}_\theta[ \theta - \theta_0 ]^T \left( \nabla \log \sigma(\theta_0) + o_p(1) \right)  \notag\\
        &= \left( \Delta_n + o_p\left(\frac{1}{\sqrt{n}}\right) \right)^T \left( \nabla \log \sigma(\theta_0) + o_p(1) \right) \\
        &= \Delta_n^T \nabla \log \sigma(\theta_0) + o_p\left(\frac{1}{\sqrt{n}}\right).
        \label{eq:post_log_sigma} 
    \end{align}
    The third equality follows from Lemma \ref{lem:gb_lem14}.
    Then, it is easy to check that 
    \begin{align}
        \left( \mathbb{E}_\theta[\log \sigma(\theta) - \log \sigma(\theta_0)] \right)^2 &= \left( \Delta_n^T \nabla \log \sigma(\theta_0) + o_p\left(\frac{1}{\sqrt{n}}\right) \right)^2 \notag\\
        &= \Tr( \Delta_n \Delta_n^T (\nabla \log \sigma(\theta_0)) (\nabla \log \sigma(\theta_0))^T ) + o_p\left(\frac{1}{n}\right).
        \label{eq:squared_post_log_sigma}
    \end{align}
    Similarly, Eq. \eqref{eq:expn_log_sigma_1st} allows us to obtain 
    \begin{align}
        \mathbb{E}_\theta\left[ \left( \log \sigma(\theta) - \log \sigma(\theta_0) \right)^2 \right]
        &= \mathbb{E}_\theta\left[ \left\{ (\theta - \theta_0)^T \left( \nabla \log \sigma(\theta_0) + o_p(1) \right) \right\}^2 \right] \notag \\
        &= \mathbb{E}_\theta\left[ (\theta - \theta_0)^T (\nabla \log \sigma(\theta_0)) (\nabla \log \sigma(\theta_0))^T (\theta - \theta_0) (1 + o_p(1)) \right] \notag\\
        &= \Tr( \left(\frac{J^{-1}}{n} + \Delta_n \Delta_n^T\right) (\nabla \log \sigma(\theta_0)) (\nabla \log \sigma(\theta_0))^T ) + o_p\left(\frac{1}{n}\right) 
        \label{eq:post_log_sigma_squared}
    \end{align}
    where the third equality follows from Lemma \ref{lem:gb_lem14}.
    Use a combination of Eqs. \eqref{eq:squared_post_log_sigma} and \eqref{eq:post_log_sigma_squared} to conclude 
    \begin{align}
        \mathbb{V}_{\theta}[\log \sigma(\theta)] &= \frac{1}{n} \Tr( J^{-1} (\nabla \log \sigma(\theta_0)) (\nabla \log \sigma(\theta_0))^T ) + o_p\left(\frac{1}{n}\right),\\
        \Tr( \rho \mathbb{V}_{\theta}[\log \sigma(\theta)] ) 
        &= \frac{1}{n} \Tr( \Tr( \rho (\nabla \log \sigma(\theta_0)) (\nabla \log \sigma(\theta_0))^T ) J^{-1} ) + o_p\left(\frac{1}{n}\right) \\
        &= \frac{1}{n} \Tr( I^Q J^{-1} ) + o_p\left(\frac{1}{n}\right).
        \label{eq:reg_expn_G_n^Q_second}
    \end{align}
    Combining Eqs. \eqref{eq:G_n^Q_expanded}, \eqref{eq:reg_expn_G_n^Q_first}, and \eqref{eq:reg_expn_G_n^Q_second} concludes the proof for $G_n^Q$.
    
    \noindent\textit{Expansion of $T_n^Q$: }
    We further expand Eq. \eqref{eq:T_n^Q_expanded}.
    We start the expansion with recalling the empirical process $\eta_n^Q(\theta)$ defined in Appendix \ref{subsec:Empirical process with classical shadow}:
    \begin{align}
        - \Tr( \left(\frac{1}{n}\sum_{i=1}^{n} \hat{\rho}_{x_i}\right) \log \sigma(\theta)) 
        &= - \Tr(\rho \log \sigma(\theta)) + \Tr(\rho \log \sigma(\theta_0^Q) ) - \Tr( \left(\frac{1}{n}\sum_{i=1}^{n} \hat{\rho}_{x_i}\right) \log \sigma(\theta_0^Q)) - \frac{1}{\sqrt{n}} \eta_n^Q(\theta) \notag\\
        &= - \Tr(\rho \log \sigma(\theta_0^Q)) - (\theta-\theta_0^Q)^T \nabla \Tr(\rho \log \sigma(\theta_0^Q)) + \frac{1}{2} (\theta-\theta_0^Q)^T J^Q(\theta_1) (\theta-\theta_0^Q)  \nonumber\\
        &\quad + \Tr(\rho \log \sigma(\theta_0^Q))- \Tr( \left(\frac{1}{n}\sum_{i=1}^{n} \hat{\rho}_{x_i}\right) \log \sigma(\theta_0^Q)) - \frac{1}{\sqrt{n}} \left\{\eta_n^Q(\theta_0^Q) + (\theta-\theta_0^Q)^T \nabla \eta_n^Q(\theta_2)\right\} \\
        &= - \Tr( \left(\frac{1}{n}\sum_{i=1}^{n} \hat{\rho}_{x_i}\right) \log \sigma(\theta_0^Q)) + \frac{1}{2} (\theta-\theta_0^Q)^T J^Q(\theta_1) (\theta-\theta_0^Q) - \frac{1}{\sqrt{n}} (\theta-\theta_0^Q)^T \nabla \eta_n^Q(\theta_2) \notag\\
        &= - \Tr( \left(\frac{1}{n}\sum_{i=1}^{n} \hat{\rho}_{x_i}\right) \log \sigma(\theta_0^Q)) + \frac{1}{2} \left\| J^Q(\theta_1)^{1/2} \left(\theta - \theta_0^Q - \frac{J^Q(\theta_1)^{-1} \nabla \eta_n^Q(\theta_2)}{\sqrt{n}}\right) \right\|^2 \nonumber\\ 
        &\quad - \frac{1}{2} \left\| J^Q(\theta_1)^{-1/2} \frac{\nabla \eta_n^Q(\theta_2)}{\sqrt{n}} \right\|^2.
        \label{eq:reg_expn_T_n^Q_first_mid1}
    \end{align}
    The second equality uses the mean-value theorem for $\Tr(\rho \log \sigma(\theta))$ to the second order and $\eta_n^Q(\theta)$ to the first order, where $\theta_1$ and $\theta_2$ are points between $\theta$ and $\theta_0^Q$, respectively.
    The third equality follows from the fact that $\nabla \Tr(\rho \log \sigma(\theta_0^Q)) = \bm 0$ and $\eta_n^Q(\theta_0^Q) = 0$.
    The last equality is obtained by completing the square.
    Since $\theta_1$ and $\theta_2$ converge $\theta_0$ as $n$ grows to $\infty$, the right hand side of Eq. \eqref{eq:reg_expn_T_n^Q_first_mid1} converges to 
    \begin{align}
        & - \Tr( \left(\frac{1}{n}\sum_{i=1}^{n} \hat{\rho}_{x_i}\right) \log \sigma(\theta_0^Q)) + \frac{1}{2} \left\| \sqrt{J^Q} (\theta - \theta_0^Q - \Delta_n^Q) \right\|^2 - \frac{1}{2} \left\| \sqrt{J^Q} \Delta_n^Q \right\|^2 + o_p\left(\frac{1}{n}\right) \\
        &= - \Tr( \left(\frac{1}{n}\sum_{i=1}^{n} \hat{\rho}_{x_i}\right) \log \sigma(\theta_0^Q)) + \frac{1}{2}  (\theta - \theta_0^Q - \Delta_n^Q)^T J^Q (\theta - \theta_0^Q - \Delta_n^Q)
        - \frac{1}{2} \Delta_n^{QT} J^Q \Delta_n^{Q} + o_p\left(\frac{1}{n}\right),
        \label{eq:reg_expn_T_n^Q_first_mid2}
    \end{align}
    with $\Delta_n^Q$ defined in Appendix \ref{subsec:Empirical process with classical shadow}.
    Then, evaluate the posterior mean of the above second term as 
    \begin{align}
        &\mathbb{E}_\theta[ (\theta - \theta_0^Q - \Delta_n^Q)^T J^Q (\theta - \theta_0^Q - \Delta_n^Q) ] \nonumber\\
        &= \mathbb{E}_\theta[ (\theta - \theta_0^Q)^T J^Q (\theta - \theta_0^Q) ] - \mathbb{E}_\theta[ \Delta_n^{QT} J^Q (\theta - \theta_0^Q) ] - \mathbb{E}_\theta[ (\theta - \theta_0^Q)^T J^Q \Delta_n^{Q} ] + \Delta_n^{QT} J^Q \Delta_n^{Q} \notag \\
        &= \left( \frac{1}{n} \Tr(J^Q J^{-1}) + \Delta_n^T J^Q \Delta_n + o_p\left(\frac{1}{n}\right) \right) - \Delta_n^{QT} J^Q \left( \Delta_n + o_p\left(\frac{1}{\sqrt{n}}\right) \right) 
        -\left( \Delta_n + o_p\left(\frac{1}{\sqrt{n}}\right) \right)^T J^Q  \Delta_n^{Q} + \Delta_n^{QT} J^Q \Delta_n^{Q} \\ 
        &= \frac{1}{n} \Tr(J^Q J^{-1}) + \Delta_n^T J^Q \Delta_n - \Delta_n^{QT} J^Q \Delta_n 
        - \Delta_n^T J^Q \Delta_n^{Q} + \Delta_n^{QT} J^Q \Delta_n^{Q} + o_p\left(\frac{1}{n}\right).
        \label{eq:expn_T_n^Q_1st_2nd}
    \end{align}
    The second equality uses the relation $\theta_0^Q=\theta_0$ and Lemma \ref{lem:gb_lem14}.
    It is also easy to check that
    \begin{align}
        \Tr(\left(\frac{1}{n}\sum_{i=1}^{n} \hat{\rho}_{x_i}\right) \mathbb{V}_{\theta}[\log \sigma(\theta)]) 
        &= \frac{1}{n}\sum_{i=1}^{n} \Tr(\hat{\rho}_{x_i} \mathbb{V}_{\theta}[\log \sigma(\theta)])\notag \\
        &= \frac{1}{n}\sum_{i=1}^{n} \left\{ \frac{1}{n} \Tr((I^Q + o_p(1)) J^{-1}) \right\} + o_p\left(\frac{1}{n}\right) \\
        &= \frac{1}{n} \Tr(I^Q J^{-1}) + o_p\left(\frac{1}{n}\right).
        \label{eq:expnT_n^Q_2nd}
    \end{align}
    The second equality follows from the law of large numbers for $\Tr(\hat{\rho}_{x_\alpha} \mathbb{V}_{\theta}[\log \sigma(\theta)])$ and Eq. \eqref{eq:reg_expn_G_n^Q_second}.
    Combining Eqs. \eqref{eq:T_n^Q_expanded}, \eqref{eq:reg_expn_T_n^Q_first_mid1}, \eqref{eq:reg_expn_T_n^Q_first_mid2}, \eqref{eq:expn_T_n^Q_1st_2nd}, and \eqref{eq:expnT_n^Q_2nd} concludes the proof for $T_n^Q$.
\end{proof}

As an application of this claim, we obtain the following simple representations. 
\begin{cor}
\label{cor:expansion of GnQ and TnQ for regular cases}
    Under Definition \ref{def:regular for classical}, the generalization loss $G_n^Q$ and training loss $T_n^Q$ are expanded as 
    \begin{align*}
        G_n^Q &=
        - \Tr(\rho \log \sigma(\theta_0)) + \frac{1}{n} (\lambda^Q + R_1^Q - \nu^Q) + o_p\left(\frac{1}{n}\right), \\
        T_n^Q &= - \Tr(\left(\frac{1}{n}\sum_{i=1}^{n} \hat{\rho}_{x_i}\right) \log \sigma(\theta_0)) + \frac{1}{n} (\lambda^Q + R_1^Q - R_2^Q - \nu^Q) + o_p\left(\frac{1}{n}\right),
    \end{align*}
    where 
    \begin{align}
        &\lambda^Q \coloneqq \frac{1}{2} \Tr(J^Q J^{-1}), 
        \quad \nu^Q \coloneqq \frac{1}{2} \Tr(I^Q J^{-1}),\\
            &R_1^Q \coloneqq \frac{n}{2} \Delta_n^T J^Q \Delta_n,
        \quad R_2^Q \coloneqq \frac{n}{2}(\Delta_n^{QT} J^Q \Delta_n + \Delta_n^T J^Q \Delta_n^Q).
    \end{align}
\end{cor}
Several distinctions can be observed when compared with Proposition \ref{prop:generalization loss of G_n and T_n} (1), which is the corresponding result in the classical case.
Let us first compare each term in the expansion of the generalization loss $G_n^Q$ and $G_n$.
At first, while the second term explicitly represents the number of parameters $d$ in the classical case, $\lambda^Q$ involves the Hessian of the quantum relative entropy and KL divergence, leading to the ratio of quantum (Bogoliubov) and classical Fisher information in the realizable case.
Secondly, in the quantum case, $J$ in the third term and $I$ in the forth term of the classical case are replaced with $J^Q$ and $I^Q$, respectively.
Thirdly, while the difference between the generalization and training loss up to the order $o_p(1/n)$ is $\Delta_n^T J \Delta_n /2$ in the classical case, it is $R_2^Q$ in the quantum case.
As we will see shortly, these differences account for the differences between the second term in WAIC and QWAIC.

Taking the expectations, we obtain the necessary presentations to define QWAIC.
\begin{thm}\label{thm:q_expectations for regular cases}
    Suppose that Assumptions \ref{ass:R1} and \ref{ass:R2} are satisfied.
    Then, the generalization loss $G_n^Q$ and training loss $T_n^Q$ can be expanded as follows:
    \begin{align}
        \mathbb{E}_{X^n}[G_n^Q] &= - \Tr(\rho \log \sigma(\theta_0)) + \frac{1}{n}\left( \lambda^Q + \nu'^Q - \nu^Q \right) + o\left(\frac{1}{n}\right),
        \\
        \mathbb{E}_{X^n}[T_n^Q] &= - \Tr(\rho \log \sigma(\theta_0)) + \frac{1}{n}\left( \lambda^Q + \nu'^Q - 2\chi^Q - \nu^Q \right) + o\left(\frac{1}{n}\right), 
    \end{align}
    where 
    \begin{align}
        \nu'^Q & \coloneqq \frac{1}{2} \Tr(J^Q J^{-1} I J^{-1}), \\
        \quad \chi^Q & \coloneqq \frac{n}{4}(\mathbb{E}_{X^n}[\Delta_n^{QT} J^Q \Delta_n] + \mathbb{E}_{X^n}[\Delta_n^T J^Q \Delta_n^Q]).
        \label{eq:lambda^Q_and_chi^Q}
    \end{align}
\end{thm}
\begin{proof}
    Lemma \ref{lem:gb_lem14} implies  
    \begin{align}
        \mathbb{E}_{X^n}[\Delta_n^T J^Q \Delta_n] = \frac{1}{n} \Tr(J^Q J^{-1} I J^{-1}).
    \end{align}
    From the law of large numbers,
    \begin{align*}
        \mathbb{E}_{X^n}\left[- \Tr(\left(\frac{1}{n}\sum_{i=1}^{n} \hat{\rho}_{x_i}\right) \log \sigma(\theta_0))\right] &= - \Tr(\rho \log \sigma(\theta_0)) + o\left(\frac{1}{n}\right).
    \end{align*}

    Taking the expectation $\mathbb{E}_{X^n}[\cdot]$ of both sides of Eqs. \eqref{eq:expn_G_n^Q} and \eqref{eq:expn_T_n^Q} and inserting the above relations completes the proof.
\end{proof}

We shall introduce a criterion for model selection based on the observations in WAIC (Theorem \ref{thm:WAIC is unbiased estimator}).
\begin{defn}
\label{def:QWAIC}
    Let us define \textit{QWAIC (Quantum Widely Applicable Information Criterion)} as
    \begin{align}
        \mathrm{QWAIC} &\coloneqq T_n^{Q} + C_n^Q
        \label{eq:defn_qwaic}
    \end{align}
    where $C_n^Q$ is the posterior covariance of classical and quantum log-likelihood, using a classical snapshot defined by 
    \begin{align}
        C_n^Q &\coloneqq \frac{1}{n} \sum_{i=1}^{n} \mathrm{Cov}_{\theta}\left[ \log p (x_i|\theta), \Tr(\hat{\rho}_{x_i} \log \sigma(\theta)) \right], 
        \label{eq:C_n^Q} \\
        \mathrm{Cov}_{\theta}\left[ \log p (x|\theta), \Tr(\hat{\rho}_x \log \sigma(\theta)) \right]
        &\coloneqq \mathbb{E}_\theta[(\log p(x|\theta) - \mathbb{E}_\theta[\log p(x|\theta)])(\Tr(\hat{\rho}_x \log \sigma(\theta)) - \mathbb{E}_\theta[\Tr(\hat{\rho}_x \log \sigma(\theta))])] \notag \\ 
        &= \mathbb{E}_\theta[f(x,\theta)f^Q(\hat{\rho}_x,\theta)] - \mathbb{E}_\theta[f(x,\theta)] \mathbb{E}_\theta[f^Q(\hat{\rho}_x,\theta)],
        \label{eq:cov_f_f^Q}
    \end{align}
    where we recall that $f(x,\theta)$ (resp. $f^Q(\hat{\rho}_x,\theta)$) is the log-likelihood ratio function (resp. quantum log-likelihood ratio function) defined in Definition \ref{defn:matrix I and J} (resp. \ref{defn:matrix I and J_quantum}).
\end{defn}
\begin{rem}
    For the purpose of extending the reach of WAIC, Refs. \cite{iba2023Posterior,iba2022Posterior} also utilize the posterior covariance to estimate the predictive risks for weighted likelihood and arbitrary loss functions.
    While they share the same idea of using posterior covariance as QWAIC, our proposal suggests that it can also be applied to quantum loss functions, underscoring that the idea can be applied to quantum state models $\sigma(\theta)$ that do not commute with the target state $\rho$, which is not obvious in classical statistics.
\end{rem}
To conclude this subsection, we show that QWAIC is an asymptotically unbiased estimator of the quantum generalization loss $G_n^Q$ for regular cases; see also Corollary \ref{cor:waic}.
\begin{thm}\label{thm:q_QWAIC for regular cases}
    Suppose that Assumptions \ref{ass:R1} and \ref{ass:R2} are satisfied.
    Then, the following equation holds:
    \begin{align*}
        \mathbb{E}_{X^n}[G_n^Q] = \mathbb{E}_{X^n}[\mathrm{QWAIC}] + o\left(\frac{1}{n}\right).
    \end{align*}
\end{thm}
\begin{proof}
    It suffices to show the asymptotic equation
    \begin{align}\label{eq:chi^Q_equal_cov}
        \frac{2 \chi^Q}{n} = \mathbb{E}_{X^n}\left[ C_n^Q \right] + o\left(\frac{1}{n}\right).
    \end{align}
    Using the Taylor expansion to the first-order 
    \begin{align}
        \log p(x|\theta) &= \log p(x|\theta_0) + (\theta - \theta_0)^T (\nabla \log p(x|\theta_0) + o_p(1)), \\
        \log \sigma(\theta) &= \log \sigma(\theta_0^Q) + (\theta - \theta_0^Q)^T (\nabla \log \sigma(\theta_0^Q) + o_p(1)),
    \end{align}
    the empirical sum of the first term in Eq. \eqref{eq:cov_f_f^Q} can be expanded as
    \begin{align}
        &\frac{1}{n} \sum_{i=1}^{n} \mathbb{E}_\theta[f(x_i,\theta)f^Q(\hat{\rho}_{x_i},\theta)] \nonumber\\
        &= \frac{1}{n} \sum_{i=1}^{n} \mathbb{E}_\theta\left[ \left\{(\theta - \theta_0)^T (\nabla \log p(x_i|\theta_0) + o_p(1)) \right\} \left\{ (\theta - \theta_0^Q)^T \left(\nabla \Tr(\hat{\rho}_{x_i} \log \sigma(\theta_0^Q)) + o_p(1)\right)\right\} \right] \\
        &= \frac{1}{n} \sum_{i=1}^{n} \mathbb{E}_\theta\left[ \Tr\left\{  (\theta - \theta_0^Q) (\theta - \theta_0)^T \left(\nabla \log p(x_i|\theta_0)\right) \left(\nabla \Tr(\hat{\rho}_{x_i} \log \sigma(\theta_0^Q)) \right)^T \right\} (1 + o_p(1)) \right] \\
        &= \frac{1}{n} \sum_{i=1}^{n} \Tr\left\{ \mathbb{E}_\theta\left[ (\theta - \theta_0^Q) (\theta - \theta_0)^T \right] \left(\nabla \log p(x_i|\theta_0)\right) \left(\nabla \Tr(\hat{\rho}_{x_i} \log \sigma(\theta_0^Q)) \right)^T \right\} (1 + o_p(1)) \\
        &= \Tr\left\{ \left( \frac{J^{-1}}{n} + \Delta_n \Delta_n^T \right) \left( \frac{1}{n} \sum_{i=1}^{n} \left(\nabla \log p(x_i|\theta_0)\right) \left(\nabla \Tr(\hat{\rho}_{x_i} \log \sigma(\theta_0^Q)) \right)^T \right) \right\} + o_p\left(\frac{1}{n}\right). 
    \end{align}
    The last equality follows from the relation $\theta_0^Q = \theta_0$ and Lemma \ref{lem:gb_lem14}.
    The expansion of the empirical sum of the second term in Eq. \eqref{eq:cov_f_f^Q} follows the same calculation yielding 
    \begin{align}
        &\frac{1}{n} \sum_{i=1}^{n} \mathbb{E}_\theta[f(x_i,\theta)] \mathbb{E}_\theta[f^Q(\hat{\rho}_{x_i},\theta)] \nonumber\\
        &= \frac{1}{n} \sum_{i=1}^{n} \mathbb{E}_\theta\left[ (\theta - \theta_0)^T (\nabla \log p(x_i|\theta_0) + o_p(1)) \right] \mathbb{E}_\theta\left[ (\theta - \theta_0^Q)^T \left(\nabla \Tr(\hat{\rho}_{x_i} \log \sigma(\theta_0^Q)) + o_p(1)\right) \right] \\
        &= \frac{1}{n} \sum_{i=1}^{n} \left( \Delta_n^T \nabla \log p(x_i|\theta_0) + o_p(1/\sqrt{n}) \right) \left( \Delta_n^T \nabla \Tr(\hat{\rho}_{x_i} \log \sigma(\theta_0^Q)) + o_p(1/\sqrt{n}) \right)\\ 
        &= \frac{1}{n} \sum_{i=1}^{n} \Tr\left\{ \Delta_n \Delta_n^T (\nabla \log p(x_i|\theta_0)) \left(\nabla \Tr(\hat{\rho}_{x_i} \log \sigma(\theta_0^Q))\right)^T \right\} + o_p\left(\frac{1}{n}\right) \\
        &= \Tr\left\{ \Delta_n \Delta_n^T \left( \frac{1}{n} \sum_{i=1}^{n} \left(\nabla \log p(x_i|\theta_0)\right) \left(\nabla \Tr(\hat{\rho}_{x_i} \log \sigma(\theta_0^Q)) \right)^T \right) \right\} + o_p\left(\frac{1}{n}\right).
    \end{align}
    Hence, plugging these expansions into Eq. \eqref{eq:C_n^Q} and then taking the expectation $\mathbb{E}_{X_n}[\cdot]$ yields
    \begin{align}
        \mathbb{E}_{X^n}\left[ C_n^Q \right]
        &= \mathbb{E}_{X^n}\left[ \frac{1}{n} \sum_{i=1}^{n} \mathrm{Cov}_{\theta}\left[ \log p (x_i|\theta), \Tr(\hat{\rho}_{x_i} \log \sigma(\theta)) \right] + o_p\left(\frac{1}{n}\right) \right] \\
        &= \mathbb{E}_{X^n}\left[ \Tr\left\{ \frac{J^{-1}}{n} \left( \frac{1}{n} \sum_{i=1}^{n} \left(\nabla \log p(x_i|\theta_0)\right) \left(\nabla \Tr(\hat{\rho}_{x_i} \log \sigma(\theta_0^Q)) \right)^T \right) \right\} \right] + o\left(\frac{1}{n}\right) \\
        &= \Tr\left\{ \frac{J^{-1}}{n} \mathbb{E}_X\left[ \left(\nabla \log p(X|\theta_0)\right) \left(\nabla \Tr(\hat{\rho}_X \log \sigma(\theta_0^Q)) \right)^T  \right] \right\} + o\left(\frac{1}{n}\right) \\
        &= \Tr\left\{ \frac{J^{-1}}{n} \mathbb{E}_{X^n}\left[ (\nabla \eta_n(\theta_0)) (\nabla \eta_n^Q(\theta_0^Q))^T \right] \right\} + o\left(\frac{1}{n}\right) \\
        &= \Tr\left\{ \frac{J^{-1}}{n} \left( n \mathbb{E}_{X^n}\left[ J \Delta_n J^Q \Delta_n^{QT} \right] \right) \right\} + o\left(\frac{1}{n}\right) \\
        &= \mathbb{E}_{X^n}\left[ \Tr( \Delta_n J^Q \Delta_n^{QT} ) \right] (1 + o(1)) \\
        &= \frac{1}{2} \mathbb{E}_{X^n}\left[\Delta_n^{QT} J^Q \Delta_n + \Delta_n^T J^Q \Delta_n^Q\right] + o\left(\frac{1}{n}\right). 
    \end{align}
    The third equality follows from the fact that $\mathbb{E}_{X^n}[(1/n)\sum_{i=1}^{n} Z(X_i)] = \mathbb{E}_X[Z(X)]$ for a random matrix $Z$.
    The fourth equality can be obtained by a similar calculation to Eq. \eqref{eq:cov_nabla_eta}.
    The fifth equality follows from the definition of $\Delta_n$ and $\Delta_n^Q$ defined in Appendices \ref{subsec:Empirical process with log-likelihood}, \ref{subsec:Empirical process with classical shadow}.
    The last equality follows from the symmetry of $J^Q$.
    Thus, Eq. \eqref{eq:chi^Q_equal_cov} is proven, completing the proof of Theorem \ref{thm:q_QWAIC for regular cases}.
\end{proof}
This theorem ensures that the difference between $G_n^Q$ and $T_n^Q$ is asymptotically complemented by a newly introduced random variable $C_n^Q$ up to $o(1/n)$.
While it is the posterior variance of the classical log-likelihood in the classical case, the posterior covariance of the classical log-likelihood $\log p(x|\theta)$ and its quantum analog with a snapshot $\Tr(\hat{\rho}_x \log \sigma(\theta))$ is the key to resolve the difference between the generalization and training loss.
This can be attributed to the fact that the parameters are estimated classically, whereas the loss function is a quantum information-theoretic quantity.

\begin{rem}
\label{rem:QAIC_LL}
    Let us compare the result with our previous study \cite{yano2023Quantuma} on deriving a quantum analog of AIC, what we call QAIC. 
    Both information criteria, QAIC and QWAIC, aim to estimate the generalization performance of quantum models.
    In particular, we focus on $\mathrm{QAIC}_\mathrm{LL}$ in the previous work:
    \begin{equation}
        \mathrm{QAIC}_\mathrm{LL} = - \frac{1}{n}\sum_{i=1}^{n} \log p(x_i|\hat{\theta}) + \frac{1}{2n}\left( d + \Tr(\hat{J}^Q \hat{I}^{-1})  \right),
    \end{equation}
    where $\hat{\theta}$ is the maximum likelihood estimator, and $\hat{J}^Q$ and $\hat{I}$ are consistent estimators of $J^Q$ and $I$, respectively.
    (Note that the value of $\mathrm{QAIC}_\mathrm{LL}$ initially introduced in the paper is divided by $1/(2n)$ such that it is also an asymptotically unbiased estimator of the quantum cross entropy.)
    Although the derivation was done in a similar manner using the Taylor expansion, there are two main differences in the setting: the parameter estimation (the maximum likelihood estimation in $\mathrm{QAIC}_\mathrm{LL}$ case and Bayesian estimation in QWAIC case) and the definition of the training loss (the classical log-likelihood function in $\mathrm{QAIC}_\mathrm{LL}$ case and its quantum analog with the classical shadow in QWAIC case).
    This leads to the difference in the first and second terms between $\mathrm{QAIC}_\mathrm{LL}$ and QWAIC.
    As seen in the relation between AIC and WAIC, the second term of $\mathrm{QAIC}_\mathrm{LL}$ is replaced with the posterior covariance term of QWAIC.
    This is a crucial improvement, especially for the quantum case, because estimating the quantum and classical Fisher information matrix is necessary even in the realizable case.
\end{rem}

\subsection{Singular cases}
Finally, we develop a theory to deal with quantum singular models.
To define WAIC, one had to assume an $L^2$ and finiteness property of the classical log-likelihood ratio function.
To analyze the singularities arising from a quantum models $(\rho, \sigma(\theta))$, we introduce the quantum average log loss function (Eq. \eqref{eq:average quantum log loss function})
    \[K^{Q}(\theta) \coloneqq D (\sigma(\theta_0^Q)\|\sigma(\theta)).\]

We begin by establishing a key lemma that connects the quantum log-likelihood ratio function to an analytic function in $L^2(q)$.
This leads to the concept of the \textit{standard form} in the literature.

\begin{lem}
\label{lem:intro_a^Q}
    Suppose that Assumption \ref{ass:S1} is satisfied.
    There exists an analytic function $a^Q(\hat{\rho}_x,u)$ with respect to $u$ that takes values in $L^2(q)$ so that 
    \begin{equation}\label{eq:intro_a^Q}
        f^Q(\hat{\rho}_x,g(u)) = u^{k^Q} a^Q(\hat{\rho}_x,u),
    \end{equation}
    for $f^Q(\hat{\rho}_x,\theta)$.
\end{lem}
\begin{proof}
This is a direct consequence of the relatively finite variance of $f^Q$.
The central concept of the proof is similar to the classical case \cite[Definition 14]{watanabe2018mathematical}, but we complete the proof considering the simultaneous resolution.
    Applying Hironaka's resolution theorem (Theorem \ref{thm:resolution_Watanabe}) to $K^Q(\theta) \geq 0$, we obtain 
    \begin{equation}
        K^Q(g(u)) = \mathbb{E}_{X}[f^Q(\hat{\rho}_X,g(u))] = r(u)u^{2k^Q}.
    \end{equation}
    From the above and Assumption \ref{ass:S1}, we obtain 
    \begin{equation}\label{eq:bound_var_f^Q}
        r(u) \geq c \mathbb{E}_{X}\left[\left(\frac{f^Q(\hat{\rho}_X,\theta)}{u^{k^Q}}\right)^2\right]
    \end{equation}
    for some $c>0$.
    Since $f^Q(\hat{\rho}_x,g(u))$ is an analytic function with respect to $u$, $f^Q(\hat{\rho}_x,g(u))$ can be written as 
    \begin{equation}
        f^Q(\hat{\rho}_x,g(u)) = u^{k^Q} a^Q(\hat{\rho}_x,u) + b^Q(\hat{\rho}_x,u),
    \end{equation}
    with some analytic functions $a^Q(\hat{\rho}_x,u)$ and $b^Q(\hat{\rho}_x,u)$.
    If $b^Q(\hat{\rho}_x,u) = 0$, then $b^Q(\hat{\rho}_x,u)/u^{k^Q}$ would not be bounded when $u^{k^Q}$ goes to $0$, which contradicts Eq. \eqref{eq:bound_var_f^Q} because $r(0)$ is bounded.
    It implies $b^Q(\hat{\rho}_x,u) = 0$.
\end{proof}

Next, we establish a lemma on the higher-order scaling for singular cases to guarantee that the assumption in Theorem \ref{thm:q_basic} holds even for singular models.
We use the standard form $a^Q(\hat{\rho}_x,u)$ introduced in the previous lemma for studying the right-hand side of the equations in Lemma \ref{lem:quantum higher order}.
\begin{lem}[Higher order scaling for singular cases]
\label{lem:higher_order_scaling_singular}
    Suppose that Assumptions \ref{ass:S1} and \ref{ass:S2} are satisfied. For $\ell \geq 3$, the higher order cumulants satisfy:
    \begin{align}
        \left| \mathbb{E}_X[\partial_\alpha^\ell s^Q(\hat{\rho}_X,\alpha) |_{\alpha=0}] \right| &\leq O_p\left(\frac{1}{n^{\ell/2}}\right), \\
        \left| \frac{1}{n} \sum_{i=1}^{n} \partial_\alpha^\ell s^Q(\hat{\rho}_{X_i},\alpha) |_{\alpha=0} \right| &\leq O_p\left(\frac{1}{n^{\ell/2}}\right).
    \end{align}
\end{lem}
\begin{proof}
    The idea of the proof is similar to that of Lemma \ref{lem:higher_order_scaling_regular} for regular cases, but we need to account for the singularities in $\Theta_0$. Instead of using the mean-value theorem, we employ the normal crossing representation (Eq. \eqref{eq:resolution of K and K^Q}).
    Noting that Lemma \ref{lem:intro_a^Q} holds for a classical snapshot $\hat{\rho}_{v_i}$ with Haar random unitaries (Eq. \eqref{eq:snapshot_Haar}), we can obtain
    \begin{equation}
        f^Q(\hat{\rho}_{v_i},g(u)) = u^{k^Q} a^Q(\hat{\rho}_{v_i},u),  \quad i = 1, ..., D,
        \label{eq:a^Q_Haar}
    \end{equation}
    for $\hat{\rho}_{v_i} = (D+1) \ketbra*{v_i}{v_i} - I_D$ with $\ket{v_i}$ appeared in the Schmidt decomposition of $F(\theta,\theta_0^Q)$ in Eq. \eqref{eq:Schmidt_decomp_F}.
    Then, using Eqs. \eqref{eq:eigendcomp_f^Q} and \eqref{eq:a^Q_Haar}, and the coincidence $k^Q = k$ (i.e. the similarity of the classical and quantum average log loss functions), we obtain
    \begin{align}
        (D+1) \lambda_i - \sum_{j=1}^{D} \lambda_j = u^{k^Q} a^Q(\hat{\rho}_{v_i}, u) = u^{k} a^Q(\hat{\rho}_{v_i}, u), \quad i = 1, ..., D.
    \end{align}
    The system of the above $D$ linear equations then implies
    \begin{align}
        \lambda_i = \frac{u^{k}}{D+1} \left( a^Q(\hat{\rho}_{v_i}, u) + \sum_{j=1}^{D} a^Q(\hat{\rho}_{v_j}, u) \right), \quad i = 1, ..., D.
    \end{align}
    This ensures that the right-hand side of Eq. \eqref{eq:quantum_higher_order_G} is bounded as
    \begin{align}
        \left|\mathbb{E}_\theta\left[\Tr( \rho|F(\theta, \theta_0^Q)|^\ell )\right]\right|
        &= \left| \mathbb{E}_\theta \left[ \sum_{i=1}^{D} |\lambda_i|^\ell \Tr(\rho \ketbra*{v_i}{v_i}) \right] \right| \\
        &\leq \left|\mathbb{E}_\theta\left[ \left| \frac{u^{k}}{D+1} \right|^\ell \left( \sum_{i=1}^{D} \left| a^Q(\hat{\rho}_{v_i}, u) + \sum_{j=1}^{D} a^Q(\hat{\rho}_{v_j}, u) \right|^\ell \Tr(\rho \ketbra*{v_i}{v_i}) \right) \right]\right| \\
        &\leq \left| \mathbb{E}_\theta\left[ \left| \frac{u^{k}}{D+1} \right|^\ell \right] \left( \sum_{i=1}^{D} \sup_u \left| a^Q(\hat{\rho}_{v_i}, u) + \sum_{j=1}^{D} a^Q(\hat{\rho}_{v_j}, u) \right|^\ell \Tr(\rho \ketbra*{v_i}{v_i}) \right) \right| \\
        &= O(n^{-\ell/2}).
    \end{align}
    The last equality can be obtained in the same way as \cite[Theorem 13]{watanabe2018mathematical} for the evaluation of $\mathbb{E}_\theta[|u^k|^\ell]$.
    A similar calculation yields bounding the right-hand side of Eq. \eqref{eq:quantum_higher_order_T} by $O(n^{-\ell/2})$.
    Plugging these bounds into Eqs. \eqref{eq:quantum_higher_order_G} and \eqref{eq:quantum_higher_order_T} completes the proof. 
\end{proof}

Now, we proceed to obtain the asymptotic behaviors of $G_n^Q$ and $T_n^Q$.
For notational convenience, we introduce a new formal variable $t^Q \coloneqq n u^{2 k^Q}$.
This is a variable that, in the classical case, was closely related to the density of states (Definition \ref{defn:invariants in singular learning theory} (4)), and similar concepts can be considered in the context of quantum state estimation. Still, here, we consider it to be just a variable for the proofs. 
Now, Theorem \ref{thm:q_basic} and Lemma \ref{lem:higher_order_scaling_singular} lead us to derive the following explicit expansion formulas.
The central idea relies on the Taylor expansion around the singular points; see also Proposition \ref{prop:singular for Gn and Tn} for comparison.
\begin{thm}
\label{thm:q_expansion formulas of G^Q and T^Q for singular cases}
    Suppose that Assumptions \ref{ass:S1} and \ref{ass:S2} are satisfied. Then, the generalization loss $G_n^Q$ and training loss $T_n^Q$ can be expanded as follows:
    \begin{align*}
        G_n^Q &= - \Tr(\rho \log \sigma(\theta_0)) + \frac{1}{n}\left( r_{CQ} \lambda + r_{CQ} \frac{1}{2}\mathbb{E}_\theta[\sqrt{t} \xi_n(u)] \right) - \frac{1}{2} \Tr(\rho \mathbb{V}_{\theta}[\log \sigma(\theta)]) + o_p\left(\frac{1}{n}\right),  \\
        T_n^Q &= - \Tr(\left(\frac{1}{n}\sum_{i=1}^{n} \hat{\rho}_{x_i}\right) \log \sigma(\theta_0)) + \frac{1}{n}\left( r_{CQ} \lambda + r_{CQ} \frac{1}{2}\mathbb{E}_\theta[\sqrt{t} \xi_n(u)]  - \mathbb{E}_\theta\left[ \sqrt{t^Q} \xi_n^Q \right] \right)
        - \frac{1}{2} \Tr(\rho \mathbb{V}_{\theta}[\log \sigma(\theta)]) + o_p\left(\frac{1}{n}\right),
    \end{align*}
    for a positive real number $r_{CQ}$.
\end{thm}
\begin{proof}
We see how each term of Eqs. \eqref{eq:G_n^Q_expanded} and \eqref{eq:T_n^Q_expanded} unfolds in singular cases.
    
    \noindent\textit{Expansion of $G_n^Q$: }
    First, we use the definition of $f^Q(\hat{\rho},\theta)$ and Lemma \ref{lem:homogeneous}, claiming $\Theta_0^Q=\Theta_0$, to obtain
    \begin{align}
        \Tr(\hat{\rho}_x \mathbb{E}_\theta[\log \sigma(\theta)]) &= \Tr(\hat{\rho}_x \log \sigma(\theta_0)) - \mathbb{E}_\theta[f^Q(\hat{\rho}_x,\theta)] \\ 
        &= \Tr(\hat{\rho}_x \log \sigma(\theta_0)) - \mathbb{E}_\theta[u^{k^Q} a^Q(\hat{\rho}_x, u)]. 
        \label{eq:sing_snapshot_posterior_mean}
    \end{align}
    The second equality follows from Lemma \ref{lem:intro_a^Q}.
    Next, taking the expectation $\mathbb{E}_X[\cdot]$ at both sides of the above equation yields
    \begin{align}
        \Tr(\rho \mathbb{E}_\theta[\log \sigma(\theta)]) &= \Tr(\rho \log \sigma(\theta_0)) - \mathbb{E}_\theta\left[ u^{k^Q} \mathbb{E}_{X}[a^Q(\hat{\rho}_X, u)] \right] \\
        &= \Tr(\rho \log \sigma(\theta_0)) - \mathbb{E}_\theta\left[ r(u) u^{2k^Q}\right],
    \end{align}
    where in the last equality, we use the relation
    \begin{equation}
        K^Q(g(u)) = \mathbb{E}_{X}[f^Q(\hat{\rho}_X,g(u))] = u^{k^Q} \mathbb{E}_{X}[a^Q(\hat{\rho}_X,u)] = r(u) u^{2k^Q}.
    \end{equation}
    Since $r(u) \neq 0$ on the whole of the affine open subset and $\Theta$ is compact,  we can evaluate as
    \begin{equation}
        r_1 \mathbb{E}_\theta[ u^{2k^Q}] \leq
        \mathbb{E}_\theta[r(u) u^{2k^Q}] \leq r_2 \mathbb{E}_\theta[ u^{2k^Q}], \quad r_1, r_2 \in \mathbb{R}_{>0},
    \end{equation}
    implying that there exists a positive real number $r_{CQ}$ ($r_1 \leq r_{CQ} \leq r_2$) such that 
    \begin{equation}
        \mathbb{E}_\theta[r(u) u^{2k^Q}] = r_{CQ} \mathbb{E}_\theta[ u^{2k^Q}].
    \end{equation}
    Then, use \cite[Theorem 12]{watanabe2018mathematical} and Assumption \ref{ass:S2} to obtain
    \begin{equation}
        r_{CQ} \mathbb{E}_\theta[ u^{2k^Q}] = r_{CQ} \mathbb{E}_\theta[u^{2k}] = \frac{r_{CQ}}{n} \mathbb{E}_\theta[t] = \frac{r_{CQ}}{n}\left( \lambda + \frac{1}{2} \mathbb{E}_\theta\left[ \sqrt{t} \xi_n(u) \right] \right),
    \end{equation}
    with the renormalized empirical process $\xi_n(u)$ defined in Eq. \eqref{eq:renormalized empirical process} and the variable $t \coloneqq n \cdot u^{2k}$ in Definition \ref{defn:invariants in singular learning theory}.

    \noindent\textit{Expansion of $T_n^Q$: }
    Taking the empirical mean at both sides of Eq. \eqref{eq:sing_snapshot_posterior_mean} yields
    \begin{align}
        \Tr( \left(\frac{1}{n}\sum_{i=1}^{n} \hat{\rho}_{x_i}\right) \mathbb{E}_\theta[\log \sigma(\theta)]) 
        &= \Tr(\left(\frac{1}{n}\sum_{i=1}^{n} \hat{\rho}_{x_i}\right) \log \sigma(\theta_0)) - \mathbb{E}_\theta\left[ u^{k^Q} \left(\frac{1}{n}\sum_{i=1}^{n} a^Q(\hat{\rho}_{x_i}, u)\right) \right] \\
        &= \Tr(\left(\frac{1}{n}\sum_{i=1}^{n} \hat{\rho}_{x_i}\right) \log \sigma(\theta_0)) - \left( \mathbb{E}_\theta\left[ r(u) u^{2k^Q} \right] - \mathbb{E}_\theta\left[ \frac{u^{k^Q}}{\sqrt{n}} \xi_n^Q \right] \right), 
    \end{align}
    where the second equality follows the definition of the quantum analog of a renormalized empirical process $\xi_n^Q$ in Eq. \eqref{eq:xi_n^Q}.
    The above second and third terms are calculated as
    \begin{equation}
        \mathbb{E}_\theta\left[ r(u) u^{2k^Q} \right] - \mathbb{E}_\theta\left[ \frac{u^{k^Q}}{\sqrt{n}} \xi_n^Q \right] 
        = \frac{1}{n}\left( r_{CQ} \lambda + r_{CQ} \frac{1}{2} \mathbb{E}_\theta\left[ \sqrt{t} \xi_n(u) \right] - \mathbb{E}_\theta\left[ \sqrt{t^Q} \xi_n^Q \right] \right).
    \end{equation}
    Moreover, following the same calculation in Lemma \ref{lem:higher_order_scaling_singular}, we have $\Tr(\hat{\rho} \mathbb{V}_{\theta}[\log \sigma(\theta)]) = O_p(1/n)$, which implies
    \begin{align}
        \Tr(\left(\frac{1}{n}\sum_{i=1}^{n} \hat{\rho}_{x_i}\right) \mathbb{V}_{\theta}[\log \sigma(\theta)]) = \Tr(\rho \mathbb{V}_{\theta}[\log \sigma(\theta)]) + o_p(1). 
    \end{align}
    Lastly, applying Lemma \ref{lem:higher_order_scaling_singular} and plugging the above into the equations in Theorem \ref{thm:q_basic} both for $G_n^Q$ and $T_n^Q$ completes the proof.
\end{proof}
Due to Assumptions \ref{ass:S1} and \ref{ass:S2}, the proof closely resembles that of the classical case.
One of the notable differences is a constant $r_{CQ}$ originating from $r(u)$ in the simultaneous resolution of singularities of $K(\theta)$ and $K^Q(\theta)$.
This constant may be regarded as the higher order version of the ratio of classical and quantum Fisher information in the realizable case, considering that $r_{CQ} \lambda$ corresponds to $\lambda^Q$ defined in Corollary \ref{cor:expansion of GnQ and TnQ for regular cases} in regular cases.
The other is $\mathbb{E}_\theta[\sqrt{t^Q} \xi_n^Q]$ in the expansion of $T_n^Q$ due to the empirical process $\xi_n^Q$, which will be detailed later.
As in regular cases, we can rewrite the above presentations more intrinsically.
\begin{thm}\label{thm:q_expansion formulas for the expectations for singular cases}
    Suppose that Assumptions \ref{ass:S1} and \ref{ass:S2} are satisfied.
    Then, 
    \begin{align*}
        \mathbb{E}_{X^n}[G_n^Q] &= - \Tr(\rho \log \sigma(\theta_0)) + \frac{1}{n}\left( r_{CQ} \lambda + r_{CQ} \nu - \nu^Q \right) + o\left(\frac{1}{n}\right), \\
        \mathbb{E}_{X^n}[T_n^Q] &= - \Tr(\rho \log \sigma(\theta_0)) + \frac{1}{n}\left(r_{CQ} \lambda + r_{CQ} \nu - 2 \chi^Q - \nu^Q \right) + o\left(\frac{1}{n}\right),
    \end{align*}
    with the real log canonical threshold $\lambda$ and singular fluctuation $\nu$ defined in Definition \ref{defn:invariants in singular learning theory}, and 
    \begin{align}
        \nu^Q \coloneqq \frac{n}{2} \mathbb{E}_{X^n} \left[ \Tr(\rho \mathbb{V}_{\theta}[\log \sigma(\theta)]) \right], \quad 
        \chi^Q \coloneqq \frac{1}{2} \mathbb{E}_{X^n}\left[ \mathbb{E}_\theta\left[ \sqrt{t^Q} \xi_n^Q \right] \right].
    \end{align}
\end{thm}
\begin{proof}
    As a consequence of Theorem \ref{thm:q_expansion formulas of G^Q and T^Q for singular cases}, this theorem is immediately derived from \cite[Lemma 22]{watanabe2018mathematical}.
\end{proof}
The quantities $\nu^Q$ and $\chi^Q$ defined here first appear in the context of quantum state estimation. As noted at the end of the paper, more detailed studies on these, as in the case of classical singular learning theory, are eagerly awaited in the future.

To obtain the final form and property of QWAIC, now, let us study $\mathbb{E}_\theta[ (u^{k^Q}/\sqrt{n}) \xi_n^Q ]$ observed in the proof of Theorem \ref{thm:q_expansion formulas of G^Q and T^Q for singular cases}.
For the later discussion, here we introduce the following quantity, inspired by the fluctuation of the renormalized posterior distribution \cite[Definition 20]{watanabe2018mathematical}.
\begin{defn}
\label{defn:quantum analog of invariants}
    By using the renormalized posterior distribution (Eq. \eqref{eq:renormalized_posterior_dist}), we define 
    \begin{align}
        C^Q(\xi_n) &\coloneqq \mathbb{E}_X\left[ \mathrm{Cov}_{\theta}\left[ \sqrt{t}a(X,u), \sqrt{t^Q}a^Q(\hat{\rho}_X,u) \right]\right], 
        \label{eq:q_func_cov}
    \end{align}
    where $t$ and $a(x,u)$ are defined in singular learning theory (Definition \ref{defn:invariants in singular learning theory}).
    The notion $\xi_n(u)$ is the renormalized empirical process defined in Eq. \eqref{eq:renormalized empirical process}.
\end{defn}
Note that $C^Q(\xi_n)$ is a functional of $\xi_n$ because the posterior covariance $\mathrm{Cov}_{\theta}[\cdot,\cdot]$ depends on $\xi_n$.
The original definition of the fluctuation of the renormalized posterior distribution in \cite{watanabe2018mathematical} is given by the posterior variance of $\sqrt{t}a(X,u)$, instead of the posterior covariance of $\sqrt{t}a(X,u)$ and $\sqrt{t^Q}a^Q(\hat{\rho}_X,u)$ in Definition \ref{defn:quantum analog of invariants}. 
This difference arises because the definition of the generalization and training loss has changed in our task of quantum state estimation.

In the following lemma, we show the relation between two quantities $\mathbb{E}_\theta[\sqrt{t^Q}\xi_n^Q(u)]$ appeared in Theorem \ref{thm:q_expansion formulas of G^Q and T^Q for singular cases} and $C^Q(\xi_n)$ in Definition \ref{defn:quantum analog of invariants} in the asymptotic limit.
\begin{lem}\label{lem:q_empirical process}
    Suppose that Assumptions \ref{ass:S1} is satisfied.
    Let $\xi(u)$ and $\xi^Q(u)$ be the Gaussian processes referred to in Proposition \ref{prop:classical convergence law} and \ref{prop:quantum convergence law}, respectively.
    Then, the following relation holds:
    \begin{equation}
        \mathbb{E}_{\xi}\left[ \mathbb{E}_\theta\left[ \sqrt{t^Q} \xi^Q(u) \right]\right] = \mathbb{E}_{\xi}\left[ C^Q(\xi) \right].
        \label{eq:asym_exp_C^Q}
    \end{equation}
\end{lem}
\begin{proof}
    Let $\{g_i\}_{i=1}^{\infty}$ and $\{g_i^Q\}_{i=1}^{\infty}$ be independent Gaussian random variables on $\mathbb{R}$ which satisfy $\mathbb{E}[g_i]=0$, $\mathbb{E}[g_ig_j] = \delta_{ij}$, $\mathbb{E}[g_i^Q]=0$, and $\mathbb{E}[g_i^Qg_j^Q] = \delta_{ij}$.
    For such random variables $g_i$, Stein's lemma implies 
    \begin{align}
        \mathbb{E}\left[g_i F(g_i)\right] = \mathbb{E}\left[\frac{\partial}{\partial g_i}F(g_i)\right]
        \label{eq:exp_rv_g}
    \end{align}
    for a differentiable absolutely continuous function of $F(\cdot)$.
    Since $L^2(q)$ is a separable Hilbert space, there exists a complete orthonormal system $\{e_i(x)\}_{i=1}^{\infty}$. 
    Putting 
    \begin{align}
        b_i(u) = \int a(x,u) e_i(u) q(x) dx, \quad b_i^Q(u) = \int a^Q(\hat{\rho}_x,u) e_i(u) q(x) dx,
    \end{align}
    it follows that 
    \begin{align}
        a(x,u) = \sum_{i=1}^{\infty} b_i(u) e_i(x), \quad
        a^Q(\hat{\rho}_x,u) = \sum_{i=1}^{\infty} b_i^Q(u) e_i(x), 
    \end{align}
    and 
    \begin{align}
        \mathbb{E}_X[a(X,u)a(X,v)] = \sum_{i=1}^{\infty} b_i(u) b_i(v),\quad \mathbb{E}_X[a(X,u)a^Q(\hat{\rho}_X,v)] = \sum_{i=1}^{\infty} b_i(u) b_i^Q(v).
        \label{eq:exp_a_a}
    \end{align}
    Consequently, the covariances are given by
    \begin{align}
        \mathbb{E}_\xi[\xi(u)\xi(v)] = \mathbb{E}_X[a(X,u)a(X,v)] = \sum_{i=1}^{\infty} b_i(u) b_i(v), \quad
        \mathbb{E}_\xi[\xi(u)\xi^Q(v)] = \mathbb{E}_X[a(X,u)a^Q(\hat{\rho}_X,v)] = \sum_{i=1}^{\infty} b_i(u) b_i^Q(v).
    \end{align}
    The first equality follows from Eq. \eqref{eq:cov_xi} and due to the fact that $u^k = 0$ because the supports of the posterior distribution are included in the set of $u$ satisfying $K(g(u)) = 0$.
    In addition, for Gaussian processes defined by
    \begin{align}
        \xi^*(u) = \sum_{i=1}^{\infty} b_i(u) g_i, \quad 
        {\xi^Q}^*(u) = \sum_{i=1}^{\infty} b_i^Q(u) g_i^Q
    \end{align}
    to have the same expectation and covariance matrices as $\xi(u)$ and $\xi^Q(u)$, we further require $\mathbb{E}[g_ig_j^Q]=\delta_{ij}$, which yields
    \begin{align}\label{eq:same_g}
        g_i=g_i^Q \quad \mathrm{for\ all}\ i.
    \end{align} 
    Next, based on the renormalized posterior distribution (Eq. \eqref{eq:renormalized_posterior_dist}), let us temporarily define 
    \begin{align}
        S[\cdot] = \int du \, D(u) \int_0^\infty dt \, t^{\lambda-1} \exp(-t) [\cdot].
    \end{align}
    Then, the left-hand side of Eq. \eqref{eq:asym_exp_C^Q} is calculated as
    \begin{align}
        \mathbb{E}_{\xi}\left[ \mathbb{E}_\theta\left[ \sqrt{t^Q} \xi^Q \right]\right] 
        &= \mathbb{E}_{\xi}\left[ \frac{S\left[\sqrt{t^Q}\xi^Q\exp(\sqrt{t}\xi)\right]}{S\left[\exp(\sqrt{t}\xi)\right]} \right] \\
        &= \sum_{i=1}^{\infty} \mathbb{E}_{\xi}\left[ \frac{\partial}{\partial g_i} \frac{S\left[\sqrt{t^Q}b_i^Q(u)\exp(\sqrt{t}\xi)\right]}{S\left[\exp(\sqrt{t}\xi)\right]} \right] \\
        &= \sum_{i=1}^{\infty} \left\{ \mathbb{E}_{\xi}\left[ \frac{S\left[\sqrt{tt^Q}b_i(u)b_i^Q(u)\exp(\sqrt{t}\xi)\right]}{S\left[\exp(\sqrt{t}\xi)\right]} \right]  - \mathbb{E}_{\xi}\left[ \frac{S\left[\sqrt{t^Q}b_i^Q(u)\exp(\sqrt{t}\xi)\right]S\left[\sqrt{t}b_i(u)\exp(\sqrt{t}\xi)\right]}{S\left[\exp(\sqrt{t}\xi)\right]^2} \right] \right\} \\
        &= \mathbb{E}_{\xi}\left[ \mathbb{E}_X\left[ \mathbb{E}_\theta\left[ \sqrt{tt^Q}a(X,u)a^Q(\hat{\rho}_X,u) \right] - \mathbb{E}_\theta\left[ \sqrt{t^Q}a^Q(\hat{\rho}_X,u) \right] \mathbb{E}_\theta\left[ \sqrt{t}a(X,u) \right]\right]\right] \\
        &= \mathbb{E}_{\xi}[C^Q(\xi)].
    \end{align}
    The second equality uses Eqs. \eqref{eq:exp_rv_g} and \eqref{eq:same_g}.
    The fourth equality follows from Eq. \eqref{eq:exp_a_a}.
    This completes the proof.
\end{proof}

\begin{rem}
\label{rem:q_suzuki36_finite}
    For a finite $n$, the following holds:
    \begin{align}
        \mathbb{E}_{\xi_n}\left[ \mathbb{E}_\theta\left[ \sqrt{t^Q} \xi_n^Q(u)\right]\right] = \mathbb{E}_{\xi_n}\left[ C^Q(\xi_n) \right] + o(1),
    \end{align}
    because of the convergence in distribution $\xi_n \to \xi$ and $\xi_n^Q \to \xi^Q$ (Propositions \ref{prop:classical convergence law} and \ref{prop:quantum convergence law}).
\end{rem}

From Theorem \ref{thm:q_expansion formulas for the expectations for singular cases} and Remark \ref{rem:q_suzuki36_finite}, we can see the difference \[\mathbb{E}_{X^n}[G_n^Q] - \mathbb{E}_{X^n}[T_n^Q]\] is $\mathbb{E}_{X^n}[C^Q(\xi_n)/n]$ up to $o(1/n)$.
Thus, the remaining problem to prove that QWAIC is an asymptotically unbiased estimator of $G_n^Q$ is to show
\begin{align}
    \mathbb{E}_{X^n}[ C_n^Q ] = \mathbb{E}_{X^n}\left[ \frac{C^Q(\xi_n)}{n} \right] + o\left(\frac{1}{n}\right), 
\end{align}
which is proven in the following lemma.
\begin{lem}
\label{lem:convergence_C_n^Q}
    Suppose that Assumptions \ref{ass:S1} is satisfied.
    Then,
    \begin{align}
        \mathbb{E}_{X^n}[ n C_n^Q ] \to \mathbb{E}_{\xi}[ C^Q(\xi) ] \quad (n \to \infty)
    \end{align}
    where $C_n^Q$ is defined in Eq. \eqref{eq:C_n^Q}.
\end{lem}
\begin{proof}
    By definition,
    \begin{align}
        C_n^Q &\coloneqq \frac{1}{n} \sum_{i=1}^{n} \mathrm{Cov}_{\theta}[\log p(x_i|g(u)),\Tr(\hat{\rho}_{x_i} \log \sigma(g(u)))] = \frac{1}{n} \sum_{i=1}^{n} \mathrm{Cov}_{\theta}\left[\sqrt{\frac{t}{n}} a(x_i,u), \sqrt{\frac{t^Q}{n}} a^Q(\hat{\rho}_{x_i},u)\right]. 
    \end{align}
    This leads us to 
    \begin{align}
        \mathbb{E}_{X^n}[ n C_n^Q ] &= \mathbb{E}_{X^n}\left[ \frac{1}{n} \sum_{i=1}^{n} \mathrm{Cov}_{\theta}\left[\sqrt{t} a(X_i,u), \sqrt{t^Q} a^Q(\hat{\rho}_{X_i},u)\right] \right] \\
        &\to \mathbb{E}_\xi\left[ \mathbb{E}_X\left[ \mathrm{Cov}_{\theta}\left[\sqrt{t} a(X,u), \sqrt{t^Q} a^Q(\hat{\rho}_{X},u)\right] \right] \right] \\
        &= \mathbb{E}_\xi\left[ C^Q(\xi) \right].
    \end{align}
    The above convergence holds due to the law of large numbers and the convergence in distribution of $\xi_n$ and $\xi_n^Q$, simultaneously.
    The last equality follows from the definition of $C^Q(\xi)$ in Definition \ref{defn:quantum analog of invariants}.
\end{proof}

Summarizing the above computation, we obtain the following desired theorem.
It establishes an information criterion for quantum singular models.
\begin{thm}\label{thm:q_QWAIC is unbiased estimator for singular cases}
    Suppose that Assumptions \ref{ass:S1} and \ref{ass:S2} are satisfied.
    Then, $\mathrm{QWAIC}$ is an asymptotically unbiased estimator for $G_n^Q$:
    \[  \mathbb{E}_{X^n}[G_n^Q] = \mathbb{E}_{X^n}[\mathrm{QWAIC}] + o\left(\frac{1}{n}\right).\]
\end{thm}
\begin{proof}
     From Lemma \ref{lem:q_empirical process} and Remark \ref{rem:q_suzuki36_finite}, we find 
     \begin{align}
         \chi^Q = \frac{1}{2} \mathbb{E}_{X^n}\left[ C^Q(\xi_n) \right] + o(1).
     \end{align}
     Then, Theorem \ref{thm:q_expansion formulas for the expectations for singular cases} and Lemma \ref{lem:convergence_C_n^Q} ensures 
     \begin{align}
         \mathbb{E}_{X^n}[\mathrm{QWAIC}] &= \mathbb{E}_{X^n}\left[ T_n^Q + C_n^Q \right]\\
         &= - \Tr(\rho \log \sigma(\theta_0)) + \frac{1}{n} \left( r_{CQ} \lambda + r_{CQ} \nu - \nu^Q \right) + o\left(\frac{1}{n}\right) \\
         &= \mathbb{E}_{X^n}[G_n^Q] + o\left(\frac{1}{n}\right),
     \end{align}
     completing the proof.
\end{proof}

\begin{rem}
    QWAIC, defined in this paper, generalizes the classical WAIC for probability density estimation to quantum state estimation, accounting for the additional penalty introduced by quantum measurements. 
    In fact, we can confirm that QWAIC and WAIC behave the same by regarding a true quantum state $\rho$ and quantum statistical models $\sigma(\theta)$ as classical states and replacing the tomographic complete measurement with the computational basis measurement.
    This criterion can be used for model selection in quantum statistical inference, providing an asymptotically unbiased estimate of the quantum generalization loss.
\end{rem}

It would be interesting to consider the geometric meanings of $\nu^Q$, $\nu'^Q$, and $\chi^Q$ appearing in the analysis of QWAIC, similar to $\lambda$ and $\nu$ in singular learning theory.
This will be explored in future work.

\newpage

\end{document}